\documentclass[12pt]{article}

\usepackage{setspace}

\usepackage{color}
\usepackage{amsfonts}
\usepackage{amssymb}
\usepackage{amsmath}
\usepackage[mathscr]{eucal}
\usepackage{amsfonts}
\usepackage{amsmath,amsxtra,latexsym,amsthm, amssymb, amscd}
\usepackage[colorlinks,citecolor=blue,filecolor=black,linkcolor=blue,urlcolor=black]{hyperref}
\usepackage{pgfplots} 
\usepackage{url}
\usepackage[round]{natbib}
\usepackage{fancyhdr}
\usepackage{tikz}

\usepackage{graphicx}
\usepackage{hyperref}
\usepackage{subcaption}
\usepackage{caption}
\usepackage{epstopdf}

\usepackage{soul}
\numberwithin{equation}{section}
\newtheorem{theorem}{Theorem}

\newtheorem{corollary}{Corollary}

\newtheorem{definition}{Definition}
\newtheorem{example}{Numerical simulation}

\newtheorem{lemma}{Lemma}

\newtheorem{proposition}{Proposition}
\newtheorem{remark}{Remark}

\newtheorem{assum}{Assumption}

\newcommand{\rr}{\mathbb{R}}

\newcommand{\s}{\mathcal{S}}

\newcommand{\ee}{\mathcal{E}}

\newcommand{\summ}{\sum\limits}

\newcommand{\ma}{\max\limits}

\DeclareMathOperator*{\argmax}{arg\,max}

\begin{document}
\title{(Non-Monotonic) Effects of Productivity and Credit Constraints on Equilibrium Aggregate Production in General Equilibrium Models with Heterogeneous Producers\footnote{The author would like to thank Cuong Le Van and Phu Nguyen-Van for helpful comments. Thanks also to participants at ISER Seminar (Osaka University), EWET 2024, PET 2023, SAET 2023 for insightful discussions and suggestions.}
}

\author{Ngoc-Sang PHAM\footnote{{\it Emails}: pns.pham@gmail.com, npham@em-normandie.fr. Phone:  +33 2 50 32 04 08. Address:  EM Normandie (campus Caen), 9 Rue Claude Bloch, 14000 Caen, France.}\\ \vspace*{0.5cm}EM Normandie Business School, M\'etis Lab (France)}
\date{\today}
\maketitle

\begin{abstract}
We show that, in a market economy, the aggregate production level depends not only on the aggregate variables but also on the distribution of individual characteristics (e.g., productivity, credit limit, ...). We prove that, due to financial frictions, the equilibrium aggregate production may be non-monotonic in both individual productivity and credit limit. We provide conditions (based on exogenous parameters) under which this phenomenon happens. By consequence, improving productivity or relaxing credit limit of firms may not necessarily be beneficial to economic development. 
\\


\end{abstract}


\section{Introduction} 
We investigate two basic questions in economics: what are the impacts of (individual and aggregate) productivity and financial changes on the aggregate output? 

Looking back to the literature, on the one hand, the productivity is widely viewed as one of the most important determinants of economic growth. In economics textbooks and classical papers \citep{solow57, romer86, romer90}, an increase of productivity generates a positive effect on the aggregate output and economic growth. On the other hand, one can expect that relaxing credit limits would have positive impact on the aggregate output as argued by several papers (for example,  \cite{kt13} (section VI. C), \cite{mx14} (section II.B), \cite{moll14} (Proposition 1), and \cite{cchst17}). 

We provide a novel view: whether a rise of productivity or credit limit generates a positive (or negative) effect on the aggregate output depends on the distribution of productivity, the size of these rises and  the level of financial imperfection. In order to explore our insights, we build general equilibrium models with credit constraints and heterogeneous producers (having their own productivity) and provide conditions under which the equilibrium aggregate production is decreasing (or increasing) in the producers' productivity and credit limit.


Let us first explain the role of productivity in a static framework. We prove that when the productivity of all agents increases,  this change improves the aggregate production if either (1) the productivity growth rates are the same or (2) there is no financial friction. 
However, the more interesting and realistic case is when the productivity of producers increases at different rates and there is a credit constraint (these two styled facts are documented by several  studies).\footnote{See, for instance, \cite{Syverson11}, \cite{acg15}, 
\cite{Kehrig15}, \cite{bbdf16}, \cite{decker18}, \cite{bbc17}, \cite{bcl21}, \cite{lw21}, \cite{gb22}.} In this case, we argue that the aggregate production may decrease. This may happen if the TFP of less productive agents increases faster than that of more productive agents. Indeed, in such a case, less productive agents absorb more capital and produces more. Since the aggregate capital supply is limited, other producers (who are more productive) get less capital (because of market imperfections) and so they produce less. By consequence, the net effect may be negative. This happens if (1) the TFP of less productive agents is far from that of more productive producers, i.e., the productivity dispersion is high,\footnote{\label{footnote1}\cite{acg15} use a harmonised cross-country dataset, based on underlying data from the
OECD-ORBIS database \citep{gal13}, to analyze the characteristics of firms that operate at the global productivity frontier and their relationship with other firms in the economy. \cite{acg15} document growing productivity dispersion for several developed countries over the 2000s. \cite{bcl21} present empirical evidence showing an increase in productivity dispersion between French firms during the period 1991-2016, with a growing
productivity gap between frontier and laggard firms. See \cite{gklw24} for an excellent review on the slowdown in productivity growth.} (2) the productivity rise  is quite small, (3) the credit constraint is tight. 

Regarding the role of credit constraints, we argue that, while a homogeneous rise of credit limit improves the aggregate output, an asymmetric rise   of credit limits can reduce the output.  The intuition behind this result is similar to that in the case of productivity effects we have mentioned above: If credit limits of less productive agents increase faster that those of more productive ones, less productive agents get more capital and more productive agents get less capital, hence the aggregate output may decrease.  It should be noticed that although the aggregate output is not necessarily monotonic in credit limits of producers, it does not exceed that in the frictionless economy which is in line with the existing literature.

In the second part of our paper, we investigate our above questions in infinite-horizon models \`a la Ramsey. Before doing this, we prove the existence of intertemporal equilibrium. To do so, we adopt the following approach:\footnote{See \cite*{bblvs15} and  \cite*{lvp16} among others.} (1) we prove the existence of equilibrium for each $T-$ truncated economy $\ee^T$; (2) we show that this sequence of equilibria converges for the product topology to an equilibrium of our original economy. 

We show that the non-monotonic effect of productivity and credit limit on the aggregate output cannot appear at the steady state. The reason is that the steady state interest rate only depends on the rate of time preferences of agents. Therefore, we focus on the global dynamics of intertemporal equilibrium.  Technically, this task is far from trivial and very few papers do this.\footnote{The existing literature focuses on the balanced-growth path, recursive equilibrium or provides  analyses around the steady-state equilibrium. See \cite{lvp16} for intertemporal equilibrium in a model with heterogeneous households and a representative producer.} However, we manage to obtain several insights. First, our findings suggest that a permanent increase of productivity of less productive agents improves the aggregate output in the long run. However, when this productivity rise is quite small and credit constraints are tight, the aggregate output may decrease in the short-run and then increase from some period on.

Second, we look at the effects of credit limits. Recall that in the static model, an increase in the most productive agent's credit limit is always beneficial for the aggregate output. However, along intertemporal equilibrium, we show that an increase of the credit limit of the most productive producer may reduce the output at every period. The intuition behind is that when her(his) credit limit goes up, the equilibrium interest rate increases, and hence, her(his) repayment also increases. This in turn reduces her(his) net worth in the next period. By consequence, her(his) saving and hence the production decrease.  The economic mechanism can be summarized by the following schema:
\begin{align}
&\text{Credit limit } \uparrow \quad  \Rightarrow  \text{ Interest rate  } \uparrow \quad  \Rightarrow  \text{ Agent's net worth  } \downarrow   \quad \Rightarrow \notag\\
 & \Rightarrow  \text{ Saving  } \downarrow \quad\Rightarrow  \text{ Production  } \downarrow   \quad  \Rightarrow  \cdots
\end{align}
As in the static model, this mechanism can happen because the credit limit of the most productive agent remains low and the productivity dispersion is high.

Third, we show how the equilibrium interest rate and the outcomes of intertemporal equilibrium (in particular in the long run) depend on the distribution of initial endowments, credit limit and productivity as well as of the discount factors. Recall that in a standard Ramsey model with one representative producer, the most patient household owns the entire capital of the economy after some finite time - this is the so-called {\it Ramsey conjecture} - and the equilibrium interest rate in the long run depends only on the rate of preference time of the most patient agent \citep{beckermitra12, bdm14, bbd15}. In our models with many potential producers, along the intertemporal equilibrium, in particular in the long run, there may be several producers sharing the aggregate capital. We point out that whether an agent holds the capital depends on the distribution of discount factor, credit limit, productivity and initial capital. Precisely, the capital holding of a producer is increasing in each of these parameters. 

\subsubsection*{Link to the literature} Our article is  related to a growing literature on general equilibrium models with heterogeneous producers and financial frictions.\footnote{The reader is referred to \citet{mat07}, \citet{quadrini11},  \citet{bes13}  for more complete reviews on the  macroeconomic effects of financial frictions and to \citet{bks15} for the relationship between entrepreneurship and financial frictions.}
 Let us mention some of them.\footnote{While we focus on producer heterogeneity, there is a growing literature studying the roles of household heterogeneity in macroeconomics (the reader is refereed to \citet{kv18} for an excellent review on this topic).} 
\citet{mx14} consider a two-sector model with a collateral constraint that requires the debt of producer does not exceed a fraction of its capital stock. They focus on balanced growth equilibrium to study the role of collateral constraint in determining TFP. Their parameterizations consistent with the data imply fairly small losses from misallocation, but potentially sizable losses from inefficiently low levels of entry and technology adoption. 
\cite{kt13} develop a dynamic stochastic general equilibrium with a representative household and heterogeneous firms facing a borrowing constraint (slightly different from ours) and focus on recursive equilibrium. They find that a negative shock to borrowing conditions can generate a large and persistent recession through disruptions to the distribution of capital. 
 \cite{bs13} develop a model with individual-specific technologies and collateral constraints to investigate the role of the misallocation and reallocation of resources in macroeconomic transitions.  \cite{bs13} find that collateral constraints have a large impact along the transition to the steady state.  \cite{moll14} studies the effect of collateral constraints on capital misallocation and aggregate productivity in a general equilibrium with a continuum of heterogeneous firms and financial frictions (modeled by a collateral constraint). 
 Proposition 1 in \cite{moll14} shows that the aggregate TFP  is increasing in the leverage ratio which is the common across firms.\footnote{In both \cite{bs13}, \cite{moll14}, the collateral constraint, which is slightly different from ours, states that the capital of a firm does not exceed a leverage ratio of its financial wealth.}  

Our paper differs from this literature in two points. First, the credit limit is individualized in our model while all credit parameters in the above studies are common across producers. Second, we argue that this credit heterogeneity plays an important role in the distribution of capital and of income as well as in the aggregate output. Indeed, we prove that the aggregate output and the aggregate TFP in our model may not be monotonic functions of the credit limits which are different across agents; they may display an inverted-U form.\footnote{Our finding is related to \citet{aetal18}. They consider a model of firm dynamics and innovation with entry, exit, and credit constraints, based on \cite{kk04}, \citet{aetal15}.  They assume that intermediate firms (monopolist) cannot invest more than $\mu$ times their current market value in innovation. They argue that the credit access may harm productivity growth because it allows less efficient incumbent firms to remain longer on the market,  which discourages entry of new and potentially more efficient innovators.} However, we show that, if agents have the same credit limit, the aggregate output and the aggregate TFP are increasing functions of this common credit limit; this finding is consistent with the above literature. 

Our paper is related to \cite{bf20} who build a general equilibrium model where productivity and wedge are exogenous parameters to study how the impact of (productivity and wedge) shocks can be decomposed into a pure technology effect and an allocative efficiency effect. There are  some differences between \cite{bf20} and the present paper.  First,  \cite{bf20} model frictions by wedge  while we model frictions by a credit constraint and the credit limit is our exogenous parameter. Second, \cite{bf20} provide a quantitative analysis by applying their approach to the firm-level markups in the U.S. but they do not provide conditions (based on exogenous parameters) under which the aggregate output is increasing or decreasing in productivity and friction level (wedge in their framework). Although we do not provide quantitative applications of our results, we show several conditions (based on exogenous parameters) under which the aggregate output is increasing or decreasing in productivity and friction level (credit limit in our framework). We also run some simulations and extend our analyses in infinite-horizon models while \cite{bf20} do not do this. 

Our paper also concerns the literature on the welfare effects of financial constraints. \cite{jp94,jp99} consider overlapping generations models with liquidity constraints and households living for three periods and argue that liquidity constraints may increase or decrease welfares.  The central point in \cite{jp94,jp99} is that liquidity constraints have two opposite effects on welfare: "they force the consumption of young below the unconstrained level but raise their permanent income by fostering capital accumulation".  
\cite{homs11} considers a general equilibrium with heterogeneous households (whose borrowings are bounded by an exogenous limit) and a representative firm. He argues that the borrowing limit has a negative on the welfare of borrower if its {\it quantity effect} dominates its {\it price effect}.  As in \cite{jp94,jp99}, the mechanism of \cite{homs11} relies on the role of supply of credit to households who need to smooth their consumption. By contrast, our mechanism focuses on credit to firms who need credit to finance their productive investment. Moreover, \cite{homs11} considers exogenous borrowing limits while we focus on credit constraints and our model has endogenous borrowing limits.

\cite{cchst17} build a dynamic general equilibrium model with heterogeneous firms and collateral constraints. They focus on the steady state and provide estimates suggesting that lifting financial frictions (modeled by collateral constraints) would increase aggregate welfare by 9.4\%  and aggregate output by 11\%. Our paper differs from \cite{cchst17} in two aspects. First, although we also find that the aggregate output in the frictionless economy is higher than that in the economy with financial frictions, it is not a monotonic function of the degree of financial friction. Second, both individual and social welfares may not be monotonic in the degree of financial friction. Interestingly, lifting credit constraint may decrease the welfare of some agents.

Last but not least, our paper contributes to the debate concerning the slowdown in aggregate productivity growth that has been documented by several studies such as \cite{acg15}, \cite{bcl21}, \cite{gklw24}; see Footnote \ref{footnote1}. 
 Our above analyses suggest that the interplay between credit constraints, high heterogeneity of productivity, asymmetry of productivity and financial shocks  may generate a slowdown in aggregate productivity growth. We argue that the aggregate productivity growth rate may be far from that of most productive firms. It may be even lower than the smallest productivity growth rate of firms. Our approach, which is different from those in the literature, is based on the general equilibrium theory with financial frictions and heterogeneous producers.


The rest of our article is organized as follows. Section \ref{motivating} presents a motivating example with two agents while Section \ref{simple} present a two-period general equilibrium framework with many producers to study the effects of productivity and credit limits. Section \ref{extensions} explores our analyses  in infinite-horizon general equilibrium models \`a la Ramsey. Section \ref{conclu} concludes. Formal proofs are gathered in the appendices.

\section{A motivating example}
\label{motivating}

In this section, we consider a deterministic two-period economy with a two agents $i=1,2$.   There is a single good (num\'eraire) which can be consumed or used to produce. Each agent $i$ has exogenous initial wealth ($S_i$ units of good) at the initial date.  To keep the model as simple as possible, we assume that agents just maximize their consumption in the second period and we focus on the output in this period. 

Agents have two ways for investing. On the one hand, agent $i$ can buy $k_i$ units of physical capital at the initial date to produce $F_i(k_i)$ units of good at the second date, where $F_i$ is the production function. Assume that $F_i(k)=A_ik$, $\forall k\geq 0$, with $0<A_1<A_2$. 

On the other hand, she can invest in a financial asset with real return $R$ which is endogenous. Denote $b_i$ the asset holding of agent $i$. She can also borrow and then pay back $Rb_i$ in the next period. However, there is a borrowing constraint. The maximization problem of agent i can be described as follows:
 \begin{subequations}
\begin{align}
(P_i):\quad & \pi_i = \ma_{k_i, b_i} {[F_i(k_i) - Rb_i]}\\
\label{concave1}\text{subject to: }& 0\leq k_i \leq S_i + b_i \text{ (budget constraint)}\\
\label{concave2g}&  Rb_i \leq \gamma_iF_i(k_i)  \text{ (borrowing constraint)}
\end{align}
\end{subequations}
where \textcolor{blue}{$\gamma_i\in (0,1)$} is an exogenous parameter.  Borrowing constraint (\ref{concave2g}) means that the repayment does not exceed the market value of the borrower's project.\footnote{Here, we follow \citet{k98} by assuming that  the debtor is required to put her project as collateral in order to borrow: If she does not repay, the creditor can seize the collateral. Due to the lack of commitment (or just because the debtor is not willing to help the creditor take the whole value of the debtor's project), the creditor can only obtain a fraction $\gamma_i$ of the total value of the project. Anticipating the possibility of default, the creditor limits the amount of credit so that the debt repayment will not exceed a fraction $\gamma_i$ of the debtor's project value.}\footnote{\citet{mat07} (Section 2) considers a model with heterogeneous agents, which corresponds to our model with $k_i=1$, $S_i=w, b_i=1-w$. However, different from our setup, investment projects in \citet{mat07} are non-divisible.} This is similar to the collateral constraint (4) in \citet{k98} or the so-called {\it earnings-based constraint} in \cite{lm21}.\footnote{
Some authors \citep{bs13,moll14} set $k_i\leq \theta w_i$, where $w_i\geq 0$ is the agent $i$'s wealth and interpret that $\theta$ measures the degree of credit frictions (credit markets are perfect if $\theta=\infty$ while $\theta=1$ corresponds to financial autarky, where all capital must be self-financed by entrepreneurs). In our framework, $S_i$ plays a similar role of wealth $w_i$ in \citet{bs13}, \citet{moll14}. Another way to introduce credit constraint is to set that $b_i\leq \theta k_i$. This corresponds to constraint (3) in \citet{mx14}. Other authors \citep{kocherlakota92, homs11} consider exogenous borrowing limits by imposing a short sales constraint: $b_i\leq B$ for any $i$.  Under these three settings, the asset holding $b_i$ is bounded from above by an upper bound which does not depend on the interest rate $R
$. \cite{cgv09} present a two-period general equilibrium model with uncertainty, num\'eraie assets,  and participation constraints described by functions of agent's choices and prices. \cite{cgv09} prove the existence of equilibrium and study indeterminacy but do not provide comparative statics.} The better the commitment, the higher value of $\gamma_i$, the larger the set of feasible allocations of the agent $i$. \citet{k98} interprets $\gamma_i$ as the collateral value of investment. In our paper, we call $\gamma_i$ the {\it credit limit} of agent $i$. 

The following table from the \cite{es}'s panel datasets suggests that borrowing and collateral constraints matter for the development of firms.

\begin{figure}[h!]
\centering
\includegraphics[width=13.5cm,height=4.5cm]{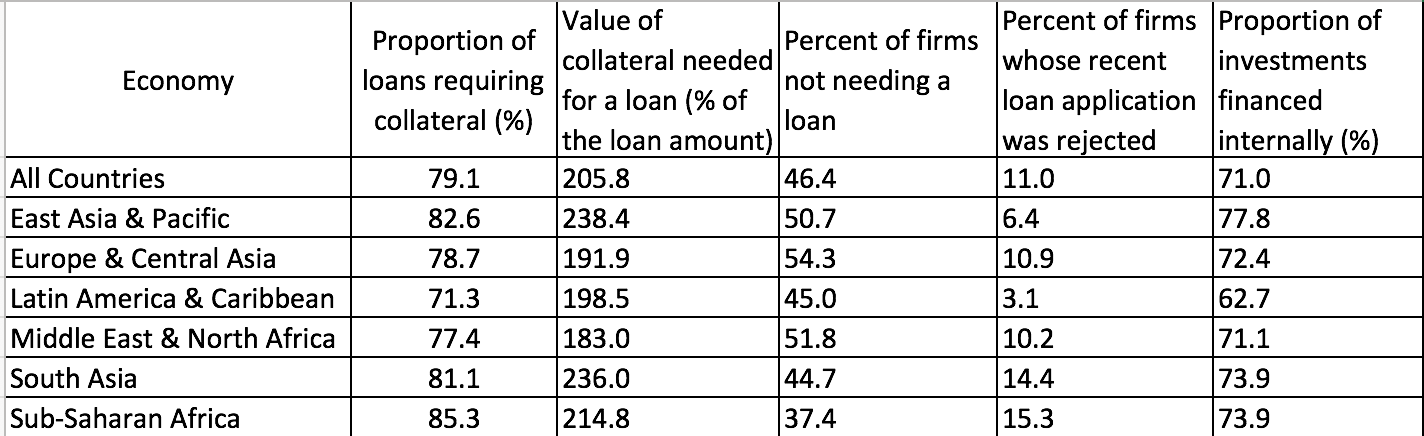}
  \label{ES}
\end{figure}

An economy $\mathcal{E}$ with credit constraints is characterized by a list of fundamentals
$$\mathcal{E} \equiv (A_i, f_i, \gamma_i,S_i)_{i=1,2}.$$
\begin{definition}\label{defi}
A list $(R, (k_i,b_i)_i)$ is an equilibrium if (1) for each $i$, given $R$, the allocation $(b_i, k_i)$ is a solution of the problem $(P_i)$, and (2) financial market clears $\sum_ib_i = 0$.
\end{definition}


In our example with linear production function, we can explicitly compute the equilibrium interest rate and aggregate output (see Theorem \ref{cate-r} in Appendix \ref{append-linear}):
\begin{lemma}In the above economy with 2 agents and linear production function, the equilibrium interest rate and aggregate output are determined by
 \begin{align}
\label{y} Y=& \begin{cases}
A_2 (S_1 + S_2) &\mbox{if }   A_1< \gamma_2A_2 \frac{S_1+ S_2} {S_1}\\
A_1S_1 + A_2S_2\dfrac{A_1(1 - \gamma_2)}{A_1 - \gamma_2A_2} &\mbox{if } A_1\geq \gamma_2A_2 \frac{S_1+ S_2} {S_1}
\end{cases}\\
R=&\begin{cases}A_2  &\text{ if } S_1 \leq \frac{\gamma_2}{1- \gamma_2}S_2\\
\frac{\gamma_2A_2(S_1+S_2)}{S_t} & \text{ if } \frac{\gamma_2}{1- \gamma_2}S_2 < S_1 <  \frac{\gamma_2A_2}{A_1- \gamma_2A_2}S_2  \\
A_1 &\text{ if }  S_1 \geq  \frac{\gamma_2A_2}{A_1- \gamma_2A_2}S_2, \text{ or, equivalently, } A_1\geq \gamma_2A_2 \frac{S_1+ S_2} {S_1}  \\
\end{cases}
\end{align}
\end{lemma}
This allows us to fully investigate the effects of productivity changes. First, we look at the individual level.
\begin{proposition}[effects of individual productivity changes]\label{motiprop1}
\begin{enumerate}
\item The aggregate output is always  increasing in $A_2$ - the productivity of the most productive agent.
\item  When $A_1< \gamma_2A_2 \frac{S_1+ S_2} {S_1}$, the aggregate output does not depend on $A_1$. When $\gamma_2A_2 \frac{S_1+ S_2} {S_1}<A_1$, we have $\frac{\partial Y}{\partial A_1}=S_1-\frac{(1-\gamma_2)\gamma_2A_2^2S_2}{(A_1-\gamma_2A_2)^2}$, and, by consequence,
\begin{align}\label{pre1}
\frac{\partial Y}{\partial A_1}\geq 0 \Leftrightarrow \frac{S_1}{S_2}\big(\frac{A_1}{A_2}-\gamma_2\big)^2 \geq (1-\gamma_2)\gamma_2
\end{align}\end{enumerate}
\end{proposition}
So, the aggregate output displays an U-shape as a function of the least productive agent's credit limit. It is increasing in $A_1$ if the productivity ratio $A_1/A_2$ is higher than a threshold (or, equivalently, the productivity gap $A_2/A_1$ is lower than a threshold). Figure \ref{effect-A_1} illustrates an example. In this numerical simulation, we set $S_1=1, S_2=0.7$, $A_2=1, \gamma_2=0.2$, and let $A_1$ vary from $\gamma_2A_2 \frac{S_1+ S_2} {S_1}=0.34$ to $A_2=2$. Then the output, as a function of $A_1$,  is decreasing on the interval $(0.34, 0.54]$ and then increasing in the interval $(0.54,1)$. \begin{figure}[h!]
\centering
   \includegraphics[width=8cm,height=6cm]{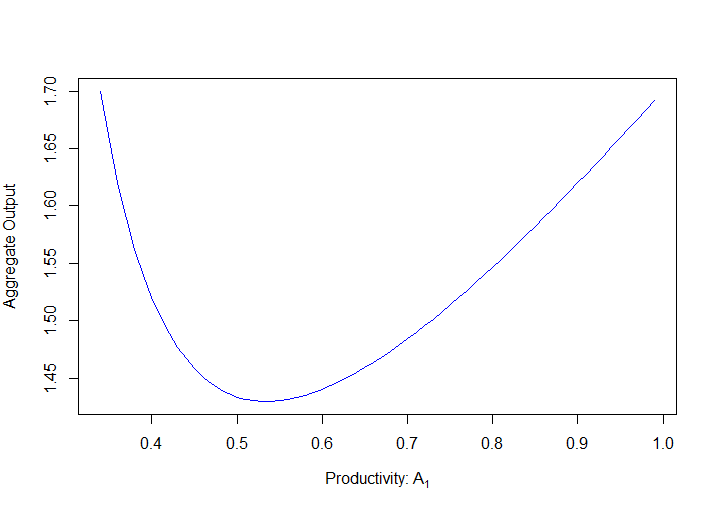}
   \caption{Non-monotonic effect of the agent 1's productivity.}
  \label{effect-A_1}
\end{figure}

We now let both productivities $A_1$ and $A_2$ vary. 
\begin{proposition}[effects of productivity changes]\label{asymmetry}
Consider a two-agent economy having linear technologies  $F_i(k)=A_ik$ $\forall i=1,2$ with $\gamma_2<A_1<A_2$, and borrowing constraints: $Rb_i\leq \gamma_iA_ik_i$.   



Assume that there is a productivity shock that changes the productivity of agents from $(A_1,A_2)$ to $(A_1',A_2')$. Assume that $A'_2>A_1'$. Assume that the credit constraint of  agent $2$ is low so that $\gamma_2<   \frac{ A_1}{ A_2} \frac{S_1}{S_1+ S_2}$ and $\gamma_2<  \frac{ A'_1}{A'_2} \frac{S_1}{S_1+ S_2}$. Then, the output change is
\begin{align}
Y(A_1',A_2')-Y(A_1,A_2)=(A_1'-A_1)S_1+A_2S_2(1-\gamma_2)\frac{A_1A_2'-A_1'A_2}{(A_1-\gamma_2A_2)(A_1'-\gamma_2A_2')}
\end{align}
(1) We have that: 
\begin{align}\label{A2fast}
\text{If } \frac{A_2'}{A_2}\geq \frac{A_1'}{A_1}\geq 1, \text{ then }Y(A_1',A_2')\geq Y(A_1,A_2)
\end{align}
(2) Assume that
\begin{align}\label{dispersion}
S_2A_2(1 - \gamma_2) \dfrac{\gamma_2A_2}{(A_1 - \gamma_2  A_2)^2}-S_1>0, \text{ i.e., } \frac{S_1}{S_2}\big(\frac{A_1}{A_2}-\gamma_2\big)^2
<(1-\gamma_2)\gamma_2
\end{align}


Then, there is a neighborhood $\mathcal{B}$ of $(A_1,A_2)$ such that
\begin{subequations}
\begin{align}\label{cs}
&\frac{Y(A_1',A_2')-Y(A_1,A_2)}{A_1'-A_1}<
0\\
\label{rate}&\forall (A_1',A_2')\in \mathcal{B} \text{ satisfying }
\frac{\frac{A'_2}{A_2}-1}{\frac{A'_1}{A_1}-1}<\frac{\gamma_2A_2}{A_1}-\frac{S_1(A_1-\gamma_2A_2)^2}{S_2A_1A_2(1-\gamma_2)}\text{ and }  A'_1\not=A_1.
\end{align}
\end{subequations}

\end{proposition}
\begin{proof}
See Appendix \ref{proof3.1}.
\end{proof}
Condition (\ref{A2fast}) says that the aggregate output increases if the productivity of both producers increases and the productivity of the most productive agent increases faster than that of the less productive one. 

Let us now focus on point 2 of Proposition \ref{asymmetry}. Here, condition (\ref{dispersion}) plays a very important role. It is satisfied if the ratio $\frac{A_1}{A_2}$ is low in the sense that $\frac{A_1}{A_2}<\gamma_2+\big(\frac{\gamma_2(1-\gamma_2)S_2}{S_1}\big)^{0.5}$. This can be interpreted as a {\it high productivity dispersion.} Under this condition, we see that $\frac{\gamma_2A_2}{A_1}-\frac{S_1(A_1-\gamma_2A_2)^2}{S_2A_1A_2(1-\gamma_2)}\in (0,1)$.  According to conditions (\ref{dispersion}) and (\ref{cs}), under a positive shock that improves the TFP of all agents, the aggregate output may decrease: 
$$Y(A_1', A_2')<Y(A_1, A_2), \forall A_1'>A_1, A'_2>A_2, (A'_1,A'_2)\in \mathcal{B} \text{ satisfying (\ref{rate})}.$$ 

Let us explain the economic intuition behind this result. Assume that the productivity dispersion is high and let us consider a small positive shock (both the TFP of both agents increases). If the productivity of the less productive agent increases faster than that of the most productive agent (i.e., $\frac{A_2'}{A_2}$  is low with respect to $\frac{A_1'}{A_1}$, see condition (\ref{rate})), the first agent absorbs more physical capital and the most productive agent gets less capital (i.e., $k_2(A_1', A_2')<k_2(A_1, A_2)$). By consequence, the aggregate output may decrease.

\section{A two-period model with many agents}\label{simple}
We now extend the two-period model in Section \ref{motivating} by allowing for a finite number $(m)$ of heterogeneous agents and general production functions $F_i: \rr_+\to \rr_+$.\footnote{We can interpret the one-factor production function $F_i$ as a reduced form for a setting with other factors of production. Indeed, suppose that the producer has a two-factor production function, say capital and labor, $G_i(k,N)$. For a given level of capital $k_i$, the firm chooses labor quantity $N_i$ to maximize its profit $ \max_{N_i\geq 0} {[G_i(k_i,N_i) -wN_i]}$. The first order condition writes $\frac{\partial G_i}{\partial N}(k_i,N_i)=w$. This implies that $N_i=N_i(k_i,w)$. So, $G_i(k_i,N_i)=G_i(k_i,N_i(k_i,w))$. We now define  $F_i(k_i)\equiv G_i(k_i,N_i(k_i,w))$.}

We require standard assumptions on the production function.
\begin{assum}\label{assumption1} The production function $F_i$ is concave, strictly increasing, $F_i(0)=0$. The credit limit $\gamma_i$ belongs the interval $(0,1)$ for any $i$.
\end{assum}

We define the notion of equilibrium as in Definition  \ref{defi}. Under the above assumption, we can prove the equilibrium existence.

\begin{proposition}\label{existence}
Under Assumption \ref{assumption1}, there exists an equilibrium. 
\end{proposition}
\begin{proof}
See Appendix.
\end{proof}

Given an equilibrium $(R, (k_i,b_i)_i)$, the aggregate output is $Y=\sum_iF_i(k_i)$. This depends on the forms of functions $(F_i)$, the initial wealths $(S_i)$, and the credit limits $(\gamma_i)$. 

Note that in an economy with perfect financial market, the aggregate production is simply determined by 
\begin{align}\label{y-1}
Y^{perfect}\equiv \max_{(k_i)\geq 0}\sum_iF_i(k_i) \quad 
\text{subject to}: \quad \sum_ik_i \leq S\equiv \sum_iS_i.
\end{align}
$Y^{perfect}$ is increasing in $A_i$, $\forall i$, and in $S$. 

In equilibrium, we have $\sum_ik=S$. So, we have that $Y\leq Y^{perfect}$.  This is consistent with a number of studies on the macroeconomic effects of financial constraints  \citep{bs13,kt14,mx14,moll14,cchst17}. 

However, an interesting open issue is whether the aggregate output is increasing or decreasing in agents' productivity $A_i$ and credit limit $\gamma_i$. In the following sections, we will investigate how the aggregate output changes when productivities $(A_i)$ and credit limits $(\gamma_i)$ vary.

\subsection{Effects of productivity changes}
\label{productivity}
We study conditions the aggregate production is increasing or decreasing when productivity changes take place. 
 Since we are interested in the effect of productivity changes, we assume that the production functions take the following form: 
\begin{align}F_i(k)=A_if_i(k),
\end{align} where the parameter $A_i>0$ represents the productivity of agent $i$ while $f_i$ is the original production function.

Assume that the TFP of agents depends on an exogenous variable $x\in \rr$ in the sense that $A_i=A_i(x)$ where $A_i$ is a differentiable function of $x$. Since we focus on positive changes, we assume that $A_i'(x)>0$, $\forall i$.

 We wonder how the aggregate output changes when $x$ varies. Note that the equilibrium physical capital, denoted by $k_i(x)$,  depends on $A_i(x)$ and the equilibrium interest rate $R$ which in turn depend on all productivities $A_1(x), \ldots, A_m(x)$. Assume that we have the differentiability. So, we can compute
\begin{align} \label{decompo-x}
k_i'(x)&=\underbrace{\frac{\partial k_i}{\partial R}}_\text{$<0$}\underbrace{\frac{\partial R}{\partial x}}_\text{$>0$}+ \underbrace{\frac{\partial k_i}{\partial A_i}}_\text{$>0$} \underbrace{\frac{\partial A_i}{\partial x}}_\text{$>0$}, \quad 
\frac{\partial R}{\partial x}=\sum_j\frac{\partial R}{\partial A_j}\frac{\partial A_j}{\partial x}
\end{align}

Notice that $ \frac{\partial k_i}{\partial A_i}\geq 0,  \frac{\partial k_i}{\partial R}\leq 0,  \frac{\partial A_i}{\partial x}\geq 0, 
\frac{\partial R}{\partial x} \geq 0$.  
By consequence, we can expect that $k_i'(x)$ may have any sign. However, we have $\sum_ik_i'(x)=0$ because $\sum_ik_i=S$ in equilibrium.

We now look at the aggregate output:
\begin{align} 
Y(x)&=\sum_{i}A_i(x)f_i\big(k_i(x)\big)=\sum_{i}A_i(x)f_i\Big(k_i\big(A_i(x),R(A_1(x),\ldots,A_m(x))\big)\Big)\notag\\
\label{dx}\frac{\partial Y}{\partial x}=&\sum_iA_i'(x)f_i\big(k_i(x)\big)+\sum_{i}A_i(x)f_i'(k_i(x))k_i'(x).
\end{align}

 By using (\ref{decompo-x}) and the fact that $\sum_ik_i'(x)=0$, we obtain two decompositions. 
\begin{proposition}[effects of productivity changes - general decompositions]
\label{prop2}
  Consider an equilibrium and assume that the equilibrium outcomes are differentiable functions. We have
\begin{subequations}\begin{align} 
\label{d1}\frac{\partial Y}{\partial x}=&\underbrace{\sum_iA_i'(x)f_i\big(k_i(x)\big)+\sum_{i: k_i'(x)\geq 0}A_i(x)f_i'(k_i(x))k'_i(x)}_\text{Added production of some agents}\notag\\
&+\underbrace{ \sum_{i: k_i'(x)<0}A_i(x)f_i'(k_i(x))k'_i(x)}_\text{Production losses of other agents}\\
\label{d2} \frac{\partial Y}{\partial x}=&\sum_i\underbrace{A_i'(x)f_i\big(k_i(x)\big)+\sum_iA_i(x)f_i'(k_i(x))\underbrace{\frac{\partial k_i}{\partial A_i}}_\text{$>0$} \underbrace{\frac{\partial A_i}{\partial x}}_\text{$>0$} }_\text{Quantity effect} +\underbrace{\sum_iA_i(x)f_i'(k_i(x))\underbrace{\frac{\partial k_i}{\partial R}}_\text{$<0$}\underbrace{\frac{\partial R}{\partial x}}_\text{$>0$}}_\text{Price effect}.
\end{align}
\end{subequations}

\end{proposition}

Proposition \ref{prop2} provides different interpretations of the effects of  productivity changes and helps us understand how the aggregate output may be increasing or decreasing in the exogenous change $x$. Look at (\ref{d1}). When $x$ increases, it generates a direct and positive effect on the productivity of agents, which is represented by the terms $\sum_iA_i'(x)f_i\big(k_i(x)\big)>0$. However, since the capital supply is fixed, we have $\sum_ik_i'(x)=0$. So, some agents get more input (i.e., $k_i'(x)\geq 0$) and produce more. However, others get less (i.e., $k_i'(x)<0$) and produce less. Therefore, the aggregate production can increase or decrease. The second decomposition (\ref{d2}) shows us the quantity and price effects. Indeed, the equilibrium physical capital $k_i$ is increasing in the productivity $A_i(x)$ which contribute to the quantity effect. However, it is  decreasing in the interest rate $R$; see (\ref{decompo-x}).  Since the interest rate is increasing in $x$, agents pay higher cost when borrowing, which generates the price effect.

Proposition \ref{prop2} leads to the following result showing the effect of individual productivity change.
\begin{corollary}[effect of individual productivity changes]
\label{prop2-cor}  Consider an equilibrium with $k_j>0$. Let only $A_j$ vary and assume that the equilibrium outcomes are differentiable functions. We have
\begin{align} 
\label{yaj}\frac{\partial Y}{\partial A_j} =&\underbrace{f_j(k_j)}_\text{Productivity effect}+ \underbrace{\sum_{i\not=j}\big(A_jf_j^\prime(k_j)-A_if_i^\prime(k_i)\big)\underbrace{\frac{-\partial k_i}{\partial R}}_\text{$\geq 0$}\underbrace{\frac{\partial R}{\partial A_j}}_\text{$\geq 0$}}_\text{Allocation effect}.
\end{align}

By consequence, ${\partial Y}/{\partial A_j}\geq 0$, $\forall j\in\mathcal{I}$, where $\mathcal{I}=\argmax_{i>n} \{A_if_i'(k_i)\}$. The aggregate output increases in $A_i$ if the producer $i$ has the highest total marginal factor productivity.

\end{corollary}

Conditions in Proposition \ref{prop2} and Corollary \ref{prop2-cor} are based on endogenous variables. We can go further by providing conditions based on exogenous parameters, shows the role of credit limit on the effect of productivity change. 

Firstly, we consider linear production functions.  The following result is a generalization of Proposition \ref{motiprop1}. 
\begin{proposition}[effects of productivity changes - linear technology]
\label{tfp-rn} 
Assume linear production functions  $F_i(k)=A_ik$ $\forall i, \forall k$, where $A_1<\cdots <A_m$. 
\begin{enumerate}
\item We have $Y\leq Y^{perfect}\equiv A_m\sum_iS_i$. Moreover, 
$Y=Y^*$ if and only if 
$f_mA_m\geq A_{m-1}(1-\frac{S_m}{S})$. 

\item 
Assume that $A_n>\max_i(\gamma_iA_i)$ and $\sum_{i=n+1}^m\frac{A_nS_i}{A_n-\gamma_iA_i}\leq S\leq \sum_{i=n}^m\frac{A_nS_i}{A_n-\gamma_iA_i}$.  Then, the equilibrium interest rate equals $A_n$\footnote{In Appendix \ref{proof3.1}, we present also the case where $R\in (A_{n-1},A_n)$ for some $n\in \{1,\ldots,m\}$. In such a case, the output is increasing in the productivity of any producer.} and the aggregate output equals $=A_n\sum_{i=1}^nS_i+\sum_{i=n+1}^m\dfrac{A_n(1-\gamma_i)}{A_n-\gamma_iA_i}A_iS_i.$

We also have that $ \frac{\partial Y}{\partial A_j}>0, \text{ $\forall $} j> n$, and  
\begin{align}
  \frac{\partial Y}{\partial A_n}=\sum_{i=1}^nS_i-\sum_{i=n+1}^m\frac{(1-\gamma_i)\gamma_iS_i}{(\frac{A_n}{A_i}-\gamma_i)^2} \label{yn-an}
  \end{align}
  \end{enumerate}
\end{proposition}
We can see clearly that $\frac{\partial Y}{\partial A_n}$ may have any sign. Since $\frac{\partial Y}{\partial A_n}$  is increasing in $A_n$, it can be negative when $A_n$ is low and positive when $A_n$ is high. This is consistent with our insights mentioned in Section \ref{motivating}.

Secondly, we investigate the case of strictly concave production function. We require standard assumptions.
\begin{assum}\label{assum-concave}For any $i$,
 the function $f_i$ is strictly increasing, strictly concave, twice continuously differentiable, $f_i(0)=0$, $f_i(\infty)=\infty$, $f_i'(0)=\infty$, $f_i'(\infty)=0$.
\end{assum}
\begin{assum}\label{Hi} For any $i$, the function $\frac{kf_i'(k)}{f_i(k)}$ is increasing in $k$.
\end{assum}
\begin{definition}\label{knb}  Given $R,\gamma_i,A_i,S_i$, denote $k^n_i=k^n_i(R/A_i)$ the unique solution to the equation  $A_if_i'(k)=R$ and $k^b_i=k^b_i(\frac{R}{\gamma_iA_i},S_i)$ the unique solution to $R(k-S_i)=\gamma_iA_if_i(k)$.
  \end{definition}
Under Assumption \ref{assum-concave}, $k^n_i$ and $k^b_i$ are uniquely defined. Observe that $k^n_i$ (resp., $k^b_i$) represents the optimal physical capital of agent $i$ when her borrowing constraint is not binding (resp., binding).

The following result explores conditions under which the equilibrium aggregate output increases or decreases in agents' productivity.
\begin{proposition}[effects of productivity changes - strictly concave technology]
\label{A1effect-general-manyagents}  Consider the case of strictly concave technology and let Assumptions \ref{assum-concave} and \ref{Hi} be satisfied. 
\begin{enumerate}
\item  The equilibrium outcomes coincide to those in the economy without frictions, (and hence, the equilibrium aggregate output is increasing in each individual productivity $A_i$) if one of the two following conditions 
\begin{enumerate}
\item[(a)] The credit limit of any agent  is high, in the sense that $\gamma_i> \lim_{x\rightarrow \infty}\frac{xf_i'(x)}{f_i(x)}, \forall i$.
\item[(b)] $\gamma_i< \lim_{x\rightarrow \infty}\frac{xf_i'(x)}{f_i(x)}, \forall i$, $R_1<R_2<\cdots <R_m$, and  $S<\sum_{i=1}^mk^n_i(R_m/A_i)$, where  $R_i$ is  the unique value satisfying 
\begin{align}\label{Ridef}
 R_i \frac{k^n_i(R_i/A_i)-S_i}{A_if_i(k^n_i(R_i/A_i))}= \gamma_i.
\end{align}
\end{enumerate}

\item Assume now that  $\gamma_i<\lim_{k \to \infty}\frac{kf_i'(k)}{f_i(k)}$, $\forall i$, and $R_2<R_3<\cdots < R_m$.  We look at the role of $A_1$.
\begin{enumerate}
\item There exists $\bar{A}_1>0$ such that the equilibrium output $Y$ is increasing in $A_1$ on the interval $(\bar{A}_1,\infty)$. 
\item Consider the case when $A_1$ is small. Denote 
\begin{align*}
D_2&=k^n_2(\frac{R_2}{A_2})+\sum_{i=3}^mk^b_i(\frac{R_2}{\gamma_iA_i},S_i), \quad D_3=\sum_{i=2}^3k^n_i(\frac{R_3}{A_i})+\sum_{i=3}^mk^b_i(\frac{R_3}{\gamma_iA_i},S_i), \ldots\\
D_m&=\sum_{i=2}^mk^n_i(\frac{R_m}{A_i})
\end{align*}
Since $R_2<R_3<\cdots < R_m$, we have $D_2>D_3>\cdots > D_m>0$.
\begin{enumerate}
\item If $S<D_m$, then the output is increasing in $A_1$ when $A_1$ is small enough.  
\item Assume that 
\begin{subequations}
\begin{align}
&D_n>S>D_{n+1}\\
\label{g1}
&\gamma_i \frac{f_i(k)}{kf_i'(k)}<\frac{S_1}{S_1+\sum_{t\geq n+1}S_t}, \forall i=n+1,\ldots,m, \forall k\in(0,S)\\
&\lim_{x\to +\infty}\frac{x}{f_1^{\prime \prime}(x)}<0
\end{align}
\end{subequations}
Then, for any $A_1$ small enough, we have that $\textcolor{blue}{
\frac{\partial Y}{\partial A_1}<0.}$
\end{enumerate}
\end{enumerate}
\end{enumerate}
\end{proposition}
\begin{proof}
See Appendix \ref{proof3.1}.
\end{proof}

Proposition \ref{A1effect-general-manyagents} explores the role of two important factors: credit limits $(\gamma_i)$ and productivity $A_1$.\footnote{In Appendix \ref{additional}, we provide more detailed analyses for the case of two agents with strictly concave technologies.}

Look at firstly on part 1 of Proposition \ref{A1effect-general-manyagents}. Condition $\gamma_i> \lim_{x\rightarrow \infty}\frac{xf_i'(x)}{f_i(x)}$ ensures that agent $i$'s borrowing constraint is not binding (see Lemma \ref{gammaalpha}  in Appendix \ref{general1-proof}). By consequence, the equilibrium coincides to that in the economy without frictions. Therefore, the output is increasing in each productivity.

Under condition 1.(b) of Proposition \ref{A1effect-general-manyagents}, Theorem \ref{general1} in Appendix \ref{proof3.1} implies that  the equilibrium coincides to that in the economy without frictions (this is similar to part 1 of Proposition \ref{tfp-rn}). Our proof is based on the key result: Agent $i$'s borrowing constraint is binding if and only if $R\leq R_i$ (see Lemma \ref{indi-general} in  Appendix \ref{34-proof}). 

Observe that $\sum_{i=1}^mk^n_i(R_m/A_i)$ is increasing in $\gamma_m$  
because $R_i/A_i$ does not depend on $A_i$ and $k^n_i(R/A_i)$ is decreasing in $R/A_i$. So, condition $S<\sum_{i=1}^mk^n_i(R_m/A_i)$ is more likely to be satisfied if the credit limit $\gamma_m$ of the agent $m$ who needs the credit the most (in the sense that $R_m>R_i$, $\forall i$) is quite high, then the credit constraints of this agent and of all other ones are not binding. 
 
 To better understand point 1.b, we  
 look at the case where $F_i(k)=A_ik^{\alpha}$, $\forall i,\forall k$, with $\alpha >\gamma_i$ $\forall i$. We can compute that $R_m=\alpha A_mS_m^{\alpha-1}(1-\frac{\gamma_m}{\alpha})^{1-\alpha}$,\footnote{See Remark \ref{cdtr} in  Appendix \ref{34-proof}.} 
 and hence 
\begin{align*}
\sum_{i=1}^mk^n_i(\frac{R_m}{A_i})=\sum_{i=2}^m(\frac{\alpha A_i}{R_m})^{\frac{1}{1-\alpha}}=\sum_{i=1}^m(\frac{ A_i}{A_m S_m^{\alpha-1}(1-\frac{\gamma_m}{\alpha})^{1-\alpha}})^{\frac{1}{1-\alpha}}=\sum_{i=1}^m(\frac{ A_i}{A_m})^{\frac{1}{1-\alpha}}\frac{S_m}{1-\frac{\gamma_m}{\alpha}}.
\end{align*}
So, we get that: $$S<\sum_{i=1}^mk^n_i(R_m/A_i) \Leftrightarrow \sum_{i=1}^mS_i<\sum_{i=1}^m(\frac{ A_i}{A_m})^{\frac{1}{1-\alpha}}\frac{S_m}{1-\frac{\gamma_m}{\alpha}}.$$
This can be satisfied if $\gamma_m$ is high in the sense that it is closed to $\alpha$.\footnote{For instance, we can take $\gamma_i=\gamma<\alpha$, $S_i=s$, $\forall i$, and $A_1<\cdots <A_m$. Then $R_1<\cdots <R_m$. Moreover, $S<\sum_{i=1}^mk^n_i(R_m/A_i)$ becomes $m(1-\frac{\gamma_m}{\alpha})<\sum_{i=1}^m(\frac{ A_i}{A_m})^{\frac{1}{1-\alpha}}$, which is satisfied if $\gamma$ is closed to $\alpha$.}

We now explain part 2 of Proposition \ref{A1effect-general-manyagents}. According to point 2.a, when the productivity $A_1$ is high, a positive productivity change is good for the aggregate output. The intuition behind is that when $A_1$ is high enough, the marginal productivity $A_1f_1'(k_1)$ of this agent is the highest total marginal factor productivity, and hence, decomposition (\ref{yaj}) ensures that $\frac{\partial Y}{\partial A_1}>0$.

Regarding point $2.b.i$ of Proposition \ref{A1effect-general-manyagents}, condition $S<D_m$ is non-empty and it can be satisfied with a large class of parameter.\footnote{Indeed, let $F_i(k)=A_ik^{\alpha}$, $\forall i,\forall k$, with $\alpha >\gamma_i$. We have $R_m=\alpha A_mS_m^{\alpha-1}(1-\frac{\gamma_m}{\alpha})^{1-\alpha}$, and hence $D_m=\sum_{i=2}^m(\frac{ A_i}{A_m})^{\frac{1}{1-\alpha}}\frac{S_m}{1-\frac{\gamma_m}{\alpha}}.$
When $S_i=s,\gamma_i=\gamma$, $\forall i$, and $A_2<\cdots <A_m$, then we have $R_2<\cdots<R_m$.  Condition $S<D_m$ is equivalent to $m(1-\frac{\gamma}{\alpha})<\sum_{i=2}^m(\frac{ A_i}{A_m})^{\frac{1}{1-\alpha}}$ which can be satisfied.
} Observe that $D_m$ is increasing in $A_2,\ldots, A_{m-1}$ but decreasing in $A_m$ because $R_i/A_i$ does not depend on $A_i$ and $k^n_i(R/A_i)$ is decreasing in $R/A_i$. Moreover, $D_m$ is increasing in agent $m'$s credit limit $\gamma_m$. In other words, condition $S<D_m$ is likely to be satisfied if $\gamma_m$ is quite high. In such a case,  point 2.a ensures that, when $A_1$ is small enough, the credit constraints of all agents are not binding and hence the aggregate output is increasing in $A_i$, $\forall i\geq 1$.

Let us now look at point 2.b.ii. Condition $D_n>S>D_{n+1}$ ensures that when $A_1$ is low enough, the credit constraint of any agent $i\geq n+1$ is binding while that of any agent $i\leq n$ is not.   Condition (\ref{g1}) means that agents whose credit constraints are binding have a very low credit limit. In such a case,  the aggregate output may be decreasing in productivity $A_1$ when $A_1$ is small enough. The fact that $A_1$ is very small ensures that the productivity dispersion is high. This is consistent with condition (\ref{pre1}) in our motivating example.

\subsubsection{Homogeneous versus heterogeneous productivity changes}

When the TFP of all producers changes at the same rate, we have the following result.
\begin{corollary}[homogeneous productivity changes]\label{TFP-homo}Consider an equilibrium. Assume that an exogenous change makes the individual TFP vary from $A_i$ to $ A_i(x)=x A_i$, $\forall i$, where $x>0$. Then, for this new economy, there is an equilibrium where $Y(x)=xY$, i.e., the aggregate output changes at the same rate.

\end{corollary}
\begin{proof}
Denote $(R,(k_i,b_i))$ an equilibrium for the economy $\mathcal{E} \equiv (A_i, f_i, \gamma_i,S_i)_{i=1,\ldots,m}$ with borrowing constraints: $Rb_i\leq \gamma_iA_if_i(k_i)$. We can check that $(R(x),(k_i,b_i))$, where $R(x)\equiv xR$, is an equilibrium for the new economy $\mathcal{E}(x) \equiv (A_i(x), f_i, \gamma_i,S_i)_{i=1,\ldots,m}.$ In equilibrium, the new aggregate output is $Y(x)=\sum_iA_i(x)f_i(k_i)=xY.$ 
\end{proof}
Next, we consider the case where productivity changes are not proportional. In such a case, we argue that positive productivity changes may reduce the aggregate output. Indeed, by using Taylor's theorem and Proposition \ref{A1effect-general-manyagents}, we obtain the following result.
\begin{corollary}[asymmetric productivity changes]
\label{propcsg1}Consider an economy which satisfies conditions in case 2.(b) in Proposition \ref{A1effect-general-manyagents}, and $A_1>0$ small enough. Then, there exist $g\in (0,1)$ and a neighborhood $\mathcal{G}$ of $(A_1,\ldots, A_m)$  such that
\begin{align}
\label{csg1}&\frac{Y(A_1',\ldots, A_m')-Y(A_1,\ldots, A_m)}{A_1'-A_1}<0,  \\
&\forall (A_1',\ldots, A_m')\in \mathcal{G} \text{ satisfying } |\frac{A_i'-A_i}{A_1'-A_1}|<g, \forall j.\notag
\end{align}
\end{corollary}
\begin{proof}
Denote $A\equiv (A_1,\ldots,A_m)$ and $A'\equiv (A'_1,\ldots,A'_m)$. 
By Taylor's theorem, we have
\begin{align*}
Y(A')-Y(A)=\frac{\partial Y(A)}{\partial{A_1}}(A'_1-A_1)+\sum_{i\geq 2}\frac{\partial Y(A)}{\partial{A_i}}(A'_i-A_i)+\sum_{i}h_i(A,A')(A'_i-A_i)
 \end{align*}
 where $\lim_{A'\to A}h_i(A,A')=0$. 

We can choose $\epsilon<0, g<1$ and $(A'_i)$ such that $|\frac{A_i'-A_i}{A_1'-A_1}|<g$ for any $i\not=1$  and $\frac{\partial Y(A)}{\partial{A_1}}+\sum_{i\geq 2}\frac{\partial Y(A)}{\partial{A_i}}\frac{A'_i-A_i}{A'_1-A_1}<\epsilon <0$. In this case, we get (\ref{csg1}).
\end{proof}

There are two key points that ensure (\ref{csg1}). The first condition is $\frac{\partial Y(A)}{\partial{A_1}}<0$, i.e., the output is decreasing in $A_1$ in a neighborhood of $(A_1,\ldots, A_m)$; notice that this may happen only if $A_1$ is small enough. Of course, we have $\frac{Y(A_1',\ldots, A_m')-Y(A_1,\ldots, A_m)}{A_1'-A_1}>0$ if $\frac{\partial Y(A)}{\partial{A_i}}>0$, $\forall i$. 
The second condition is $|\frac{A_i'-A_i}{A_1'-A_1}|<g$, i.e., the productivity does not change at the same rate and that the productivity of the less productive agent (agent 1) increases faster than that of the most productive agents. This implies that agent $1$ absorbs more capital than other ones.

\subsection{Effects of credit limits}
\label{effect-fi}





In this section, we investigate the effects of credit limits $(\gamma_t)$ on the aggregate production, which help us to understand better the relationship between finance and economic growth. A meaningful question is whether financial development has positive effects on the economic growth.  In our model,  relaxing credit limit (i.e., increasing $\gamma_i$) can be interpreted as reduction of financial friction or improvement of the financial sector.

Assume that the credit limit of all agents depends on an exogenous variable $x\in \rr$ in the sense that $\gamma_i=\gamma_i(x)$ where $\gamma_i$ is a differentiable function of $x$ and $\gamma_i'(x)>0$.  

We wonder how the aggregate output changes when $x$ varies. The equilibrium physical capital of agent $i$, which depends on $x$, is denoted by $k_i(x)$. We write $k_i(x)=k_i\big(\gamma_i(x),R(\gamma_1(x),\ldots,\gamma_m(x))\big)$, where $R=R(\gamma_1(x),\ldots,\gamma_m(x))$ is the equilibrium interest rate which depends on the credit limit $(\gamma_i(x))_i$. We can write the aggregate output as follows:
\begin{align} 
Y(x)&=\sum_{i}F_i\big(k_i(x)\big)=\sum_{i}F_i\Big(k_i\big(\gamma_i(x),R(\gamma_1(x),\ldots,\gamma_m(x))\big)\Big).
\end{align}

Assume the differentiability, we have 
\begin{align} 
k_i'(x)&=\frac{\partial k_i}{\partial \gamma_i}\frac{\partial \gamma_i}{\partial x}+\frac{\partial k_i}{\partial R}\frac{\partial R}{\partial x}, \quad 
\frac{\partial R}{\partial x}=\sum_j\frac{\partial R}{\partial \gamma_j}\frac{\partial \gamma_j}{\partial x}
\end{align}

Recall that $
 \frac{\partial k_i}{\partial \gamma_i}\geq 0,  \frac{\partial k_i}{\partial R}\leq 0,  \frac{\partial \gamma_i}{\partial x}\geq 0, 
\frac{\partial R}{\partial x} \geq 0$ because $\frac{\partial R}{\partial x}=\sum_j\frac{\partial R}{\partial \gamma_j}\frac{\partial \gamma_j}{\partial x}$ and $\frac{\partial R}{\partial \gamma_j}\geq 0$, $\forall j$. So, we see that $k_i'(x)$ may have any sign. However, we know $\sum_ik_i'(x)=0$ because $\sum_ik_i=S$ in equilibrium. By consequence, we obtain two decompositions which help us to understand why the aggregate output may be increasing or decreasing in the exogenous change $x$. 

\begin{proposition}[effects of credit changes]\label{creditdecompo}
Consider an equilibrium. 

\begin{enumerate}
\item The equilibrium outcomes do not depend on credit limits $\gamma_i(x)$ of agents whose borrowing constraints are not binding.   
\item 
For any agent $j$ whose borrowing constraint is binding, let $x$ vary and assume that the equilibrium outcomes are differentiable functions. Then, we have decompositions:
\begin{align} 
\label{df1}\frac{\partial Y}{\partial x}=&\underbrace{\sum_{i: k_i'(x)\geq 0}F_i'(k_i(x))k'_i(x)}_\text{Added production of agent $j$} +\underbrace{ \sum_{i: k_i'(x)<0}F_i'(k_i(x))k'_i(x)}_\text{Production losses of other agents}\\
\label{df2} =&\underbrace{\sum_iF_i'(k_i(x))\underbrace{\frac{\partial k_i}{\partial \gamma_i}}_\text{$>0$} \underbrace{\frac{\partial \gamma_i}{\partial x}}_\text{$>0$} }_\text{Quantity effect} +\underbrace{\sum_iF_i'(k_i(x))\underbrace{\frac{\partial k_i}{\partial R}}_\text{$<0$}\underbrace{\frac{\partial R}{\partial x}}_\text{$>0$}}_\text{Price effect}
\end{align}

\item Consider a particular case where only $A_j$ varies (other being fixed). We have that
\begin{align}
 \label{decom-difference}
\frac{\partial Y}{\partial \gamma_j}
=&\underbrace{\frac{\partial R}{\partial \gamma_j}}_\text{$\geq 0$} \sum_{i\not=j}\big(F_i'(k_j)-F_i'(k_i)\big)\underbrace{\frac{-\partial k_i}{\partial R}}_\text{$\geq 0$}
\end{align}
\end{enumerate}

\end{proposition}

While we directly get (\ref{df1}) and (\ref{df2}) by taking the derivative of $x$ with respect to $\gamma_j$, condition (\ref{decom-difference}) is a  consequence of (\ref{df1}) and the fact that $\sum_ik_i=S$.\footnote{Indeed, notice that $k_i$ depends on $R$ and $\gamma_i$, taking the derivative of both sides of $\sum_ik_i=S$ with respect to $\gamma_j$, we have 
 $\Big(\sum_{i=1}^m\frac{\partial k_i}{\partial R}\Big)\frac{\partial R}{\partial \gamma_j}+\frac{\partial k_j}{\partial \gamma_j}=0$, which imply that $\frac{\partial k_j}{\partial R}\frac{\partial R}{\partial \gamma_j}+\frac{\partial k_j}{\partial \gamma_j}=-\sum_{i\not=j}\frac{\partial k_i}{\partial R}\frac{\partial R}{\partial \gamma_j}\geq 0.$ Combining this with the equation $Y=\sum_iF_i(k_i)$, we get (\ref{decom-difference}).} Proposition \ref{creditdecompo}  has a similar insight as in Proposition \ref{prop2} and Corollary \ref{prop2-cor}. This directly leads to the following result.
\begin{corollary}\label{creditco}
Denote $\mathcal{I}_n=\argmax_{i} \{F_i'(k_i)\}$. Thus, we have that
${\partial Y}/{\partial \gamma_j}\geq 0$ $\forall j\in\mathcal{I}_n $, i.e.,  the aggregate output is increasing in the credit limit of agents having the highest marginal productivity.

\end{corollary}


We now provide conditions under which the aggregate output may be decreasing in credit limits.
 \begin{proposition}[effects of individual credit limit]\label{yn-fi-co} 
  Assume that $F_i(k)=A_ik$ $\forall i,k$.  Assume that $\max_i(\gamma_iA_i)<A_1<\cdots <A_m$. Consider the case where the equilibrium interest rate is belong to the interval $(A_{n-1},A_{n})$.  Then, we have that:
\begin{enumerate}
\item $
\frac{\partial Y}{\partial \gamma_{n}}< 0< \frac{\partial Y}{\partial \gamma_m}$  if $n<m$.\footnote{Moreover, if $n=m$ (i.e., only agent $m$ produces), we have $\frac{\partial Y}{\partial \gamma_m}= 0$.}

\item Consider an entrepreneur $i$ with $n< i< m$, we have that: 
\begin{subequations}
\begin{align}
\label{ynfi1}\frac{\partial Y}{\partial \gamma_{i}}&>0 \text{ if  $A_i$ is high enough, i.e., } \frac{A_i-A_{i-1}}{A_m-A_i}>\frac{\sum_{t=i+1}^m\frac{\gamma_{t}A_tS_t}{(A_{n-1}-\gamma_{t}A_t)^2}}{\sum_{t=n}^{i-1}\frac{\gamma_{t}A_tS_t}{(A_{n}-\gamma_{t}A_t)^2}}\\
\label{ynfi2}\frac{\partial Y}{\partial \gamma_{i}}&<0 \text{ if $A_i$ is low enough, i.e., } \frac{A_i-A_{n}}{A_{i+1}-A_i}<\frac{\sum_{t=i+1}^m\frac{\gamma_{t}A_tS_t}{(A_{n}-\gamma_{t}A_t)^2}}{\sum_{t=n}^{i-1}\frac{\gamma_{t}A_tS_t}{(A_{n-1}-\gamma_{t}A_t)^2}}.
\end{align}
\end{subequations}
\end{enumerate}
\end{proposition}
\begin{proof}See Appendix \ref{effect-fi-proof}.
\end{proof}
Condition $\frac{\partial Y}{\partial \gamma_{n+1}}< 0$ indicates that an increasing of the credit limit of the least productive producer harms the aggregate output while condition  $ \frac{\partial Y}{\partial \gamma_m}>0)$ has a similar interpretation as in Corollary \ref{creditco}.

According to (\ref{ynfi1}) and (\ref{ynfi2}), the aggregate output is more likely to be increasing (resp., decreasing) in the credit limit of an agent if the TFP of this producer is quite close to those of more productive entrepreneurs (resp., that of the least productive entrepreneur) or/and credit limits and initial wealths of more productive agents $(\gamma_{t})_{t>i}$ are low. 

We complement our above points by a numerical example.
\begin{example}\label{ex-f2-2}{\normalfont
Consider a three-agent economy with linear production functions $F_i(k)=A_ik$, $\forall i, \forall k$, and borrowing constraints are $Rb_i\leq \gamma_iA_ik_i$. In Appendix \ref{effect-fi-proof}, we completely compute the equilibrium. Assume now that fundamentals are given by $S_1=4$, $S_2=4$, $S_3=3$,  $A_1=1$, $A_2=1.2$, $A_3=1.5$, $\gamma_1=0.2$.

First, we set $\gamma_3=0.3$ and we let $\gamma_2$ vary.  Figure \ref{effect-f2-2} shows the effects of the agent $2$'s credit limit $\gamma_2$ on the equilibrium interest rate and the aggregate output.  When $\gamma_2$ varies from $0.15$ to $0.45$, the interest rate varies from $A_1=1$ to $A_2=1.2$.  The aggregate output is not monotonic functions of $\gamma_2$. Indeed, it is increasing in $\gamma_2$ in the regime $\mathcal{A}_1$ where the interest rate $R=A_1$, but decreasing in $\gamma_2$ in the regime $\mathcal{R}_1$ where the interest rate $R=R_1$ (consistent with Proposition  \ref{yn-fi-co}), and then constant in the regime $\mathcal{A}_1$ where $R=A_2$. 
 \begin{figure}[h!]
\centering
    \includegraphics[width=8cm,height=6cm]{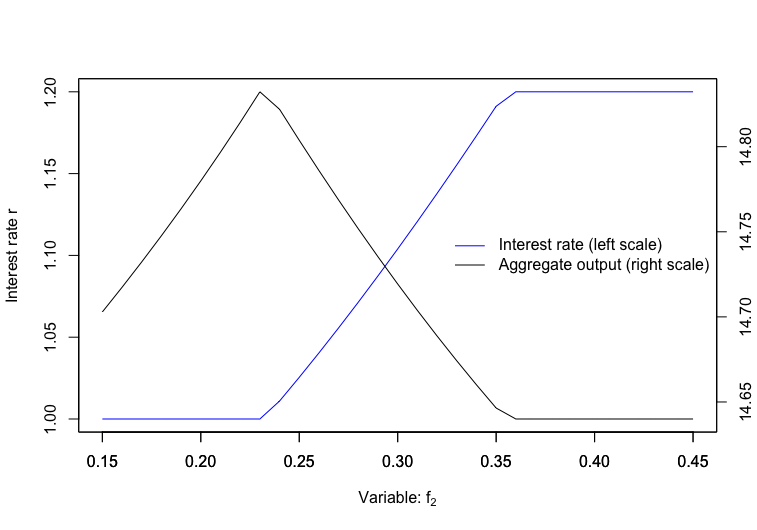}
  \caption{Non-monotonic effects of credit limit $\gamma_2$.}
  \label{effect-f2-2}
\end{figure}

Second, we set $\gamma_2=0.3$ and let $\gamma_3$ vary.  Figure \ref{effect-f3-2} shows the effects of the most productive agent's credit limit $\gamma_3$ on the equilibrium interest rate and the aggregate output.  The output is increasing in $\gamma_3$ (this is consistent with point 1 of Proposition  \ref{yn-fi-co}).
 \begin{figure}[h!]
\centering
   \includegraphics[width=8cm,height=6cm]{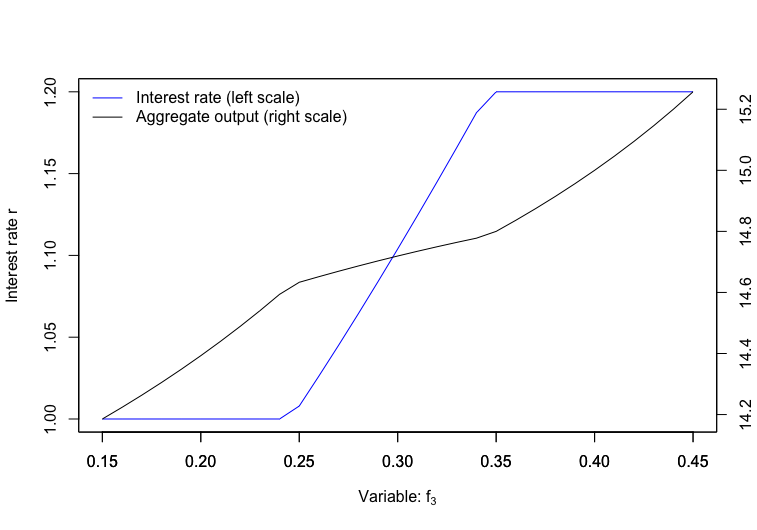}
  \caption{Monotonic effects of credit limit $\gamma_3$.}
  \label{effect-f3-2}
\end{figure}



}

\end{example}

\subsubsection{Homogeneous versus heterogeneous credit changes}


We firstly consider the case of homogeneous credit change. 

\begin{corollary}[homogeneous credit change]\label{hete-f}Assume either $F_i(k)=A_ik$, $\forall i, \forall k$ or Assumption \ref{assum-concave}  is satisfied. Assume also that $\gamma_i=\gamma\in (0,1)$, $\forall i$. Then the equilibrium aggregate output is an increasing function of the credit limit $\gamma$.
\end{corollary} 
\begin{proof}
See Appendix \ref{effect-fi-proof}.
\end{proof}

The intuition of the result is  simple: all credit-constrained producers, who have higher marginal productivity, can borrow more from other agents who have lower marginal productivity, and hence produce more.  This point is consistent with those in in \cite{kt13} (section VI. C), \cite{mx14} (section II.B), \cite{moll14} (Proposition 1), and \cite{cchst17}.


We now assume that there is an aggregate change on credit limits under which the new credit limits are $(\gamma_i')_i$.   Our novel point is that, even $\gamma_i'>\gamma_i$ $\forall i$, the new aggregate output $Y'=Y(\gamma_1',\ldots, \gamma_m')$ may be lower than $Y=Y(\gamma_1,\ldots, \gamma_m)$. Formally, we have the following result.

 \begin{corollary}[general credit changes]\label{credit-hete}
  Assume that $F_i(k)=A_ik$ $\forall i,k$, and $\max_i(\gamma_iA_i)<A_1<\cdots <A_m$. 
  Consider the case where the equilibrium interest rate is in the interval $[A_n,A_{n+1})$.  Consider an agent $i$ such that $n+1<i<m$ and assume that condition (\ref{ynfi2}) holds. Then there exist $g\in (0,1)$ and a neighborhood $\mathcal{G}$ of $(\gamma_1,\ldots, \gamma_m)$ such that
\begin{align}
&\frac{Y(\gamma_1',\ldots, \gamma_m')-Y(\gamma_1,\ldots, \gamma_m)}{\gamma_i'-\gamma_i}<0,  \\
&\forall (\gamma_1',\ldots, \gamma_m')\in \mathcal{G} \text{ satisfying } |\frac{\gamma_j'-\gamma_j}{\gamma_i'-\gamma_i}|<g, \forall j\not=i.\notag
\end{align}
\end{corollary}

We can apply the same argument used in Corollary \ref{propcsg1} to prove Corollary \ref{credit-hete}.

Corollary \ref{credit-hete} shows that the aggregate output may be reduced even the credit limits of all agents increase (i.e., $\gamma_i'>\gamma_i, \forall i$). It complements  Proposition \ref{hete-f}, Proposition \ref{yn-fi-co}, and those in \cite{bs13}, \cite{kt13}, \citet{mx14},  \cite{moll14},  \cite{cchst17}. Recall that these studies provide conditions under which relaxing credit limits has positive impact on the aggregate output.  

\subsection{Productivity growth, productivity dispersion and credit constraint}
\begin{definition}[aggregate production function and aggregate TFP]
If we assume that $F_i(k)=A_if(k)$ where $A_i$ represents the individual productivity of agent $i$ and $f$ is a production function, then we can define the aggregate production function $G$ and the aggregate TFP $A$ by
\begin{subequations}
\begin{align}
\label{TFP}
\text{the aggregate TFP: }
A&\equiv \frac{Y}{f(S)}\\
\text{the aggregate production function: }G(S)&\equiv Y=Af(S).
\end{align}
\end{subequations}
\end{definition}

 Assume that there is a shock (technical progress, for instance) that changes productivity from $A_i$ to $A_i'$ and credit limit from $\gamma_i$ to $\gamma_i'$. The new TFP of the economy is $TFP'=Y'/f(S)$.  We have 
\begin{align*}
\frac{TFP'}{TFP}=\frac{\frac{Y(A_1',\ldots,A_m',\gamma_1',\ldots,\gamma_m')}{f(S)}}{\frac{Y(A_1,\ldots,A_2,\gamma_1,\ldots,\gamma_m)}{f(S)}}=\frac{Y(A_1',\ldots,A_m',\gamma_1',\ldots,\gamma_m')}{Y(A_1,\ldots,A_2,\gamma_1,\ldots,\gamma_m)}
\end{align*}

We aim to understand the relationship between the aggregate productivity growth $\frac{TFP'}{TFP}$ and individual ones $\frac{A_1'}{A_1},\ldots,\frac{A_m'}{A_m}$.

In the economy without frictions, by using the definition (\ref{y-1}) we have that
 \begin{align*}
 \frac{TFP'}{TFP}=\frac{\max\{\sum_iA'_if(k_i): k_i\geq 0, \sum_ik_i\leq S\}}{\max\{\sum_iA_if(k_i): k_i\geq 0, \sum_ik_i\leq S\}}
 \end{align*}
Observe that $min_i\{\frac{A_i'}{A_i}\}A_if(k_i)\leq A_i'f(k_i)\leq max_i\{\frac{A_i'}{A_i}\}A_if(k_i)$. So, obtain that  $ min_i\{\frac{A_i'}{A_i}\}\leq \frac{TFP'}{TFP}\leq max_i\{\frac{A_i'}{A_i}\}$. 
 
 However, when we consider economies with credit constraints, our above analyses (see Propositions
\ref{asymmetry}, \ref{A1effect-general-manyagents},  and Corollaries \ref{propcsg1}, \ref{credit-hete}) show that the aggregate productivity growth $\frac{TFP'}{TFP}$ may be less than $min_i\{\frac{A_i'}{A_i}\}$. Indeed, for instance, we can choose $(A_i)$ and $(A_i')$ so that all conditions in Proposition \ref{propcsg1} are satisfied and $min_i\{\frac{A_i'}{A_i}\}>1$. In this case, we have $Y'-Y<0$, or, equivalently, $\frac{TFP'}{TFP}<1$. We summarize our points in the following result. 
\begin{proposition}[productivity growth, productivity dispersion and credit constraint]
Consider the case $F_i(k)=A_if(k)$, $\forall i, \forall k$. Assume that there is a shock that changes productivity from $A_i$ to $A_i'$ and credit limit from $\gamma_i$ to $\gamma_i'$.  
\begin{enumerate}
\item In the economy without frictions, we always have  
 \begin{align}\label{tfp-nofriction}
 min_i\{\frac{A_i'}{A_i}\}\leq \frac{TFP'}{TFP}\leq max_i\{\frac{A_i'}{A_i}\}
 \end{align}
\item Consider economies with credit constraints $\mathcal{E} \equiv (A_i, f_i, \gamma_i,S_i)_{i=1,\ldots,m}$. 
\begin{enumerate}
\item If $\frac{A_i'}{A_i}=g>0$, $\forall i$, then Corollary \ref{TFP-homo} implies that $\frac{TFP'}{TFP}=g$.
\item However, under some situations as in  Propositions
\ref{asymmetry}, \ref{A1effect-general-manyagents},  \ref{propcsg1}, \ref{credit-hete}, we may have that \begin{align}\label{tfp-friction}
  \frac{TFP'}{TFP}< min_i\{\frac{A_i'}{A_i}\}.
 \end{align}
 \end{enumerate}
\end{enumerate}
\end{proposition}
By consequence, the aggregate productivity growth rate may be far from that of most productive firms. It may be even lower than the smallest productivity growth rate of firms.

Our points contribute to the debate concerning the slowdown in aggregate productivity growth. For instance, by using data in 23 OECD countries over the 2000s, \cite{acg15} document a slowdown in aggregate productivity growth, a rising productivity gap between the global frontier and other firms, and that productivity growth at the global frontier remained robust. 

 The following graphic from \cite{bcl21} 
 shows the median productivity level of frontier firms and laggard firms, over the period 1991-2016 in France, productivity being measured by TFP. We see that the productivity dispersion tends to increase over time.
\begin{figure}[h!]
\centering
      \includegraphics[width=8cm,height=6cm]{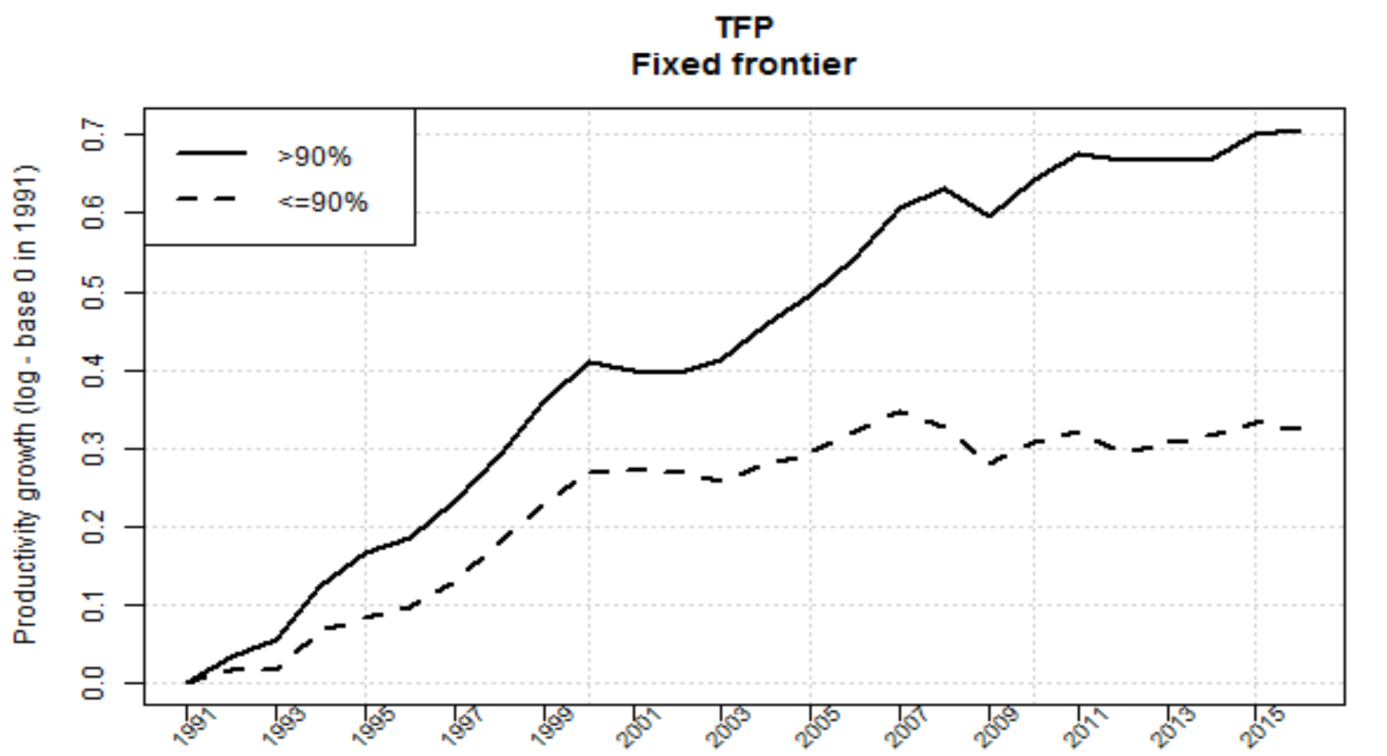}
  \caption{TFP growth. Source: \cite{bcl21}}
  \label{effect-f2f3}
\end{figure}

As recognized by  \cite{gklw24}, there is no single reason for the slowdown in aggregate productivity growth.  We provide a supply-side point of view by using a general equilibrium model with credit constraint. Our above analyses suggest that the interplay between credit constraints, high heterogeneity of productivity, asymmetry of productivity and financial shocks may generate a slowdown in the aggregate productivity growth, and eventually a decrease in the aggregate productivity.




\section{Extension: Infinite-horizon models \`a la Ramsey}
\label{extensions}

We now extend our previous models by considering  infinite-horizon  models \`a la Ramsey.  Each agent $i$ maximizes her intertemporal utility subject to budget and borrowing constraints:
\begin{subequations}
\begin{align}
 &\ma_{(c_i,k_i,b_i)} \sum_{t=0}^{\infty}\beta_i^tu_i(c_{i,t})\\
\label{bc-i}\text{subject to: }& c_{i,t}+k_{i,t}-(1-\delta)k_{i,t-1}+R_tb_{i,t-1} \leq f_{i,t+1}(k_{i,t-1}) + b_{i,t}\\
\label{cc-i}&  R_{t+1}b_{i,t}\leq \gamma_i\Big(f_{i,t}(k_{i,t})+(1-\delta)k_{i,t}\Big),\\
&c_{i,t}\geq 0, \quad k_{i,t}\geq 0,
\end{align}
\end{subequations}
where $\delta\in [0,1]$ is the depreciation rate. We assume that $b_{i,-1}=0$, $\forall i$,  and denote the exogenous initial wealth $w_{i,0}=F_{i,0}(k_{i,-1})$.

Note that we allow for non-stationary production functions. Let us define the function $F_{i,t}: \rr_+\to \rr_+$ by \textcolor{blue}{$$F_{i,t}(k)=f_{i,t}(k)+(1-\delta)k.$$}


\begin{definition}
An intertemporal equilibrium is a list $((c_{i,t},k_{i,t},b_{i,t})_i,R_t)_{t\geq 0}$ satisfying two conditions: (1) given $(R_t)$, the allocation $(c_{i,t},k_{i,t},b_{i,t})$ is a solution of the above maximization problem, and (2) markets clear: $\sum_ib_{i,t}=0$, $\sum_{i}(c_{i,t}+k_{i,t})=\sum_iF_{i,t}(k_{t-1})$, $\forall t$.
\end{definition}


Before providing the equilibrium analyses, we prove the existence of intertemporal equilibrium. To so so, we require standard assumptions.
\begin{assum}[endowments] $k_{i,-1}>0$  and $b_{i,-1}=0$ for any $i$. 
\end{assum}
\begin{assum}[borrowing limits] $\gamma_i\in (0,1)$ for any $i$.
\end{assum}

\begin{assum}[production functions] For each $i$, the function $F_{i,t}$ is concave, continuously differentiable, $F'_{i,t}>0$, $F_{i,t}(0)=0$.
\end{assum}
\begin{assum}[utility functions]For each $i$ and for each $t\geq 0$, the function $u_{i}:\rr_+\rightarrow \rr_+$ is continuously differentiable, concave, strictly increasing. 
\end{assum}

\begin{assum}[finite utility] \label{assum_utility_finite} For each $i\in \{1,\cdots,m\}$,
\begin{align}\label{utility_finite-sto}
\ma_{c_{i,t},k_{i,t}\geq 0} \Big\{\sum_{t=0}^{\infty}\beta_i^tu_i(c_{i,t}): c_{i,t}+k_{i,t}\leq F_{i,t}(k_{i,t-1})\Big\}&>-\infty\\
\summ_{t\geq 0}\beta_i^tu_{i}(B_{K,t})&<\infty.
\end{align}
where we define the exogenous sequence $(B_{K,t})$ as follows:
\begin{align}
B_{K,-1}&=\sup_{(k_i):\sum_ik_i\leq \sum_{i}k_{i,-1}; k_i\geq 0, \forall i} \sum_iF_{i,0}(k_{i})\\
B_{K,t}&=\max_{(k_i):\sum_ik_i\leq B_{K,t-1}; k_i\geq 0, \forall i} \sum_iF_{i,t}(k_{i}).
\end{align}

\end{assum}
\begin{theorem}\label{prop1} Under the above assumptions, there exists an intertemporal equilibrium.
\end{theorem}

The detailed proof is presented in Appendix \ref{existencequilibrium}. Let us explain the main idea. First, we prove the existence of equilibrium for each $T-$ truncated economy $\ee^T$ where there is no activity from date $T+1$. Second, we show that this sequence of equilibria converges for the product topology to an equilibrium of our economy. The main difficulty is to bound the volume of financial asset holding of agents. Thanks to borrowing and budget constraints, we can do this.


 



 \subsection{Effects of productivity changes}
 \label{sec41}
We firstly look at the steady state.
 \begin{proposition}[steady state analysis] \label{ss} Consider the above infinite-horizon model. Assume that $F_{i,t}=F_i$, i.e., does not depend on time. Consider a steady state equilibrium with $k_i>0$, $\forall i$. 
\begin{enumerate}
\item The steady state interest rate is $R=1/\max_i\{\beta_i\}.$
\item 
Assume, in addition, that $\beta_1>\beta_i$, $\forall i\geq 2$. 
Then  $A_1F_1'(k_1)=R=1/\beta_1$, agent $1$'s borrowing constraint is not binding, and  for any $i\geq 2$, the capital $k_i$ is determined by 
 \begin{align*}
\frac{R-\gamma_iF_{i}'(k_{i})}{R}&=F_{i}'(k_{i})(1-\gamma_i)
\end{align*}
Hence, $k_i$ is increasing in $A_i$. Since $R\beta_i\leq 1$, the value $k_i$ is increasing in credit limit $\gamma_i$. By consequence, the steady state output $Y=\sum_iF_i(k_i)$ is increasing in TFP $A_i$ and credit limit $\gamma_i$ for any $i$.

\end{enumerate}
\end{proposition}
\begin{proof}
See Appendix \ref{extensions-proof}.
\end{proof}
In Proposition \ref{ss}, the long run interest rate is determined by the time preference rate of the most patient agent. According to Proposition \ref{ss}, the non-monotonic effect of productivity and credit limit on the aggregate output can only be appeared along the global dynamics of the economy. Therefore, we will focus on global dynamics, i.e., the dynamic properties of the intertemporal equilibrium. 

In general, it is challenge to provide comparative statics of intertemporal equilibrium in infinite-horizon models.  For the sake of tractability, we assume that $u_i(c)=ln(c)$ and  $F_{i,t}(k)=A_{i,t}k$. Thanks to this specification, we can, in some cases, explicitly compute the equilibrium. 

Firstly, we look at the economy without financial frictions. It is easy to prove the following result.
\begin{lemma}[economy without credit constraint]\label{perfecty} Assume that $u_i(c)=ln(c)$ and  $F_{i,t}(k)=A_{i,t}k$ where $A_{1,t}<A_{2,t}\cdots <A_{m,t}$, $\forall t$. Consider an economy without credit constraints. Then, in equilibrium, we must have $R_t=A_{m,t}$ and the output equals denoted by $Y^*_t$ and  growth rate $(G^*_t)$ of this economy are determined by 
\begin{subequations}
\begin{align}
Y^*_t=&A_{m,t}\cdots A_{m,1}\sum_{i=1}^m\beta_i^{t-1}s_{i,0}, \\
G^*_{t+1}=&A_{m,t+1}\frac{\sum_{i =1}^m\beta_i^{t}}{\sum_{i =1}^m\beta_i^{t-1}}.
\end{align}
\end{subequations}
where we denote $s_{i,0}\equiv \beta_iw_{i,0}, \forall i$.
\end{lemma}

For the economy with credit constraints, the following result provides conditions under which the equilibrium interest rate equals the TFP of some agent. Other cases will be presented latter.
\begin{lemma}[economy with credit constraint]\label{infinite-a11}
Consider an infinite-horizon model with utility function $u_i(c)=ln(c)$ $\forall c$, $\forall t$, $\forall i$ and production functions $F_{i,t}(k)=A_{i,t}k$, $\forall k$, $\forall t$, $\forall i$. 
 Assume that $\max_i\gamma_iA_{i,t}<A_{1,t}<\cdots<A_{m,t}$ $\forall t$.
\begin{enumerate}
\item 
  Assume that there is an agent $h$ so that \begin{align}\label{mh}
\frac{\beta_h^t s_{h,0}}{1-\gamma_h}\geq &\sum_{i\leq h}\beta_i^t
s_{i,0}-\sum_{j>h}\beta_j^t\frac{\gamma_jA_{j,t+1}}{A_{h,t+1}-\gamma_jA_{j,t+1}}\frac{(1-\gamma_j)A_{j,t}}{A_{h,t}-\gamma_jA_{j,t}} \cdots \frac{(1-\gamma_j)A_{j,1}}{A_{h,1}-\gamma_jA_{j,1}} s_{j,0}>0.
\end{align}
where $s_{i,0}=\beta_iw_{i,0}$. 

Then there exists an equilibrium with $R_t=A_{h,t},$ $\forall t$.  In such an equilibrium, the aggregate output at date $t$, $(t\geq 1)$, is 
 \begin{align}\label{linear-output}
Y_t&=A_{h,t}\cdots A_{h,1}\sum_{i \leq h}\beta_i^{t-1}
s_{i,0}\\
&+A_{h,t}\cdots A_{h,1}\sum_{j>h}\beta_j^{t-1}(1-\gamma_j)^{t}\frac{A_{j,t}}{A_{h,t}-\gamma_jA_{j,t}}\frac{A_{j,t-1}}{A_{h,t-1}-\gamma_jA_{j,t-1}}\cdots \frac{A_{j,1}}{A_{h,1}-\gamma_jA_{j,1}} s_{j,0}.\notag
 \end{align}
 
\item In particular, when \begin{align}\label{mhm}
\frac{\beta_m^t s_{m,0}}{1-\gamma_m}\geq &\sum_{i\leq m}\beta_i^t
s_{i,0}, \forall t,
\end{align}
then there exists an equilibrium which coincides with the equilibrium in the economy without credit constraints: the interest rate equals $R_t=A_{m,t},$ $\forall t$ and the aggregate output is $Y_t=A_{h,t}\cdots A_{h,1}\sum_{i \leq h}\beta_i^{t-1}s_{i,0}.$ 
 
\end{enumerate}
\end{lemma}
\begin{proof}
See Appendix \ref{extensions-proof}.
\end{proof}
The right hand side of condition (\ref{mh}) ensures that agent $h$ produces, i.e., $k_{h,t}>0$ while the left hand side ensures agent h's borrowing constraint. Under these conditions, we can compute the equilibrium outcome.

\subsubsection{Effect of permanent productivity changes}
 Lemma \ref{infinite-a11} allows us to investigate the effects of productivity changes.  First, assume that, for several reasons, the productivity of producers increase (or decrease) at any date. We explore how this change affect the aggregate output and the growth rate along the intertemporal equilibrium.

\begin{proposition}\label{infinite-akm}
Assume that $F_{i,t}(k)=A_ik$, $\forall i,\forall k\geq 0$ with $\max_i{\gamma_iA_i}<A_1<A_2<\ldots<A_m$, and utility function $u_i(c)=ln(c)$ $\forall i$.
 Assume that \begin{align}\label{mhs}
\frac{\beta_h^t s_{h,0}}{1-\gamma_h}&\geq \sum_{i\leq h}\beta_i^t
s_{i,0}-\sum_{j>h}\frac{\gamma_jA_{j}}{A_{h}-\gamma_jA_{j}}\big(\frac{\beta_j(1-\gamma_j)A_{j}}{A_{h}-\gamma_jA_{j}}\big)^t s_{j,0}>0,\forall t\geq 0\\
\label{mhs2}\beta_h=\max_{i\leq h}\beta_i&>\max_{j>h}\frac{\beta_j(1-\gamma_j)A_{j}}{A_{h}-\gamma_jA_{j}}.
\end{align}
for some agent $h$.

Then, there is an equilibrium with the interest rate $R_t=A_h$, $\forall t$. In this equilibrium,  we have that:
\begin{enumerate}
\item The aggregate output equals  \begin{align}\label{outputt}
Y_t&=A_h^t\sum_{i\leq h}\beta_i^{t-1}
s_{i,0}+\sum_{j>h}\beta_j^{t-1}(1-\gamma_j)^{t}A_{j}^t\Big(\frac{A_{h}}{A_{h}-\gamma_jA_{j}}\Big)^t s_{j,0}.
 \end{align}
 This is increasing in $A_j$ for any $j>h$. However, for agent $h$, we have that:
   \begin{align}\label{ytah}
 \frac{\partial Y_t}{\partial A_h}>0 \Leftrightarrow \sum_{i\leq h}\beta_i^{t-1}
s_{i,0} -\sum_{j>h}\frac{(1-\gamma_j)\gamma_j A_{j}^2}{(A_h-\gamma_jA_j)^2}\Big(\frac{\beta_j(1-\gamma_j)A_j}{A_{h}-\gamma_jA_{j}}\Big)^{t-1} s_{j,0}>0
 \end{align}
 and this condition can be satisfied for some value of parameters $(s_{i,0},\gamma_i,A_i)_i$.

\begin{enumerate}
\item If $\sum_{i: \beta_i=\beta_h}s_{i,0}>\sum_{j>h}\frac{(1-\gamma_j)\gamma_j A_{j}^2}{(A_h-\gamma_jA_j)^2} s_{j,0}$, then $ \frac{\partial Y_t}{\partial A_h}>0$ for any $t$.
\item If $\sum_{i\leq h}s_{i,0}<\sum_{j>h}\frac{(1-\gamma_j)\gamma_j A_{j}^2}{(A_h-\gamma_jA_j)^2} s_{j,0}$, then  $\frac{\partial Y_1}{\partial A_h}< 0 $ at date $1$ but there exists a date  $t_0$ such that $\frac{\partial Y_t}{\partial A_h}\geq 0 $, $\forall t> t_0$.
\end{enumerate}

 \item The growth rate $G_{t+1}\equiv \frac{Y_{t+1}}{Y_t} $ equals
\begin{align}\label{Gth}
G_{t+1}\equiv \frac{Y_{t+1}}{Y_t} = A_h\frac{\sum_{i\leq h}\beta_i^{t}
s_{i,0}+\sum_{j>h}\Big(\frac{\beta_j(1-\gamma_j)A_j}{A_{h}-\gamma_jA_{j}}\Big)^{t+1} \frac{s_{j,0}}{\beta_j}}{\sum_{i\leq h}\beta_i^{t-1}
s_{i,0}+\sum_{j>h}\Big(\frac{\beta_jA_{j}(1-\gamma_j)}{A_{h}-\gamma_jA_{j}}\Big)^t \frac{s_{j,0}}{\beta_j}}.
\end{align} 
and it  converges to $A_h\beta_h$. Moreover, for agent $j$, with $j>h$, there exists a date $t_1$ such that the growth rate $G_{t+1}$ is decreasing in the productivity $A_j$ for any date $t\geq t_1$.

\end{enumerate}

\end{proposition}
\begin{proof}
See Appendix \ref{extensions-proof}.
\end{proof}
Observe that the right hand side of (\ref{ytah}) is increasing in $A_h$. So,  the aggregate output $Y_t$ is more likely to be increasing in the TFP $A_h$ (i.e., $\frac{\partial Y_t}{\partial A_h}>0$) if (1) the productivity gap $\frac{A_{j}}{A_{h}}$ (for $j>h$) is low or (2) the initial income gap $\frac{s_{j,0}}{s_{i,0}}$ (for $j>h$, $i\leq h$) is low or (3) the time preference gap $\frac{\beta_{j}}{\beta_{i}}$ (for $j>h$, $i\leq h$) is low.

 Condition (\ref{mhs}) ensures that agent h still produces, i.e.,  $k_{h,t}>0$, $\forall t$.\footnote{Notice that condition (\ref{mhs}) requires that $
\max_{i\geq h}\beta_i\geq \max_{j>h}\frac{\beta_j(1-\gamma_j)A_{j}}{A_{h}-\gamma_jA_{j}}$ and $ \max\Big(\beta_h,\max_{j>h}\frac{\beta_j(1-\gamma_j)A_{j}}{A_{h}-\gamma_jA_{j}}\Big)\geq \max_{i< h}\beta_i.$}
 This happens if its TFP $A_h$ is not too low and the rate of time preference $\beta_h$ is high enough. Notice that  Condition (\ref{mhs2}) ensures imply that $\beta_hA_h>\beta_jA_j$, $\forall j>h$. This ensures that agent $1$ still produces and the growth rate $\frac{Y_{t+1}}{Y_t}$ converges to $\beta_hA_h$. 

Proposition \ref{infinite-akm} allows us to understand the impact of a shock on the TFP of the less productive agent. Observe that, if $A_h$ increases, then the output will increase in the long run in both cases cases 1.a and 1.b of Proposition \ref{infinite-akm}. However, in the short-run, point 1.b of Proposition \ref{infinite-akm} indicates that, the output may decrease in the short run if the productivity dispersion is high.


\begin{example}{\normalfont
To complement our theoretical findings presented above, we run a simulation in a two-agent model with linear production function $F_i(k)=A_ik$, and $s_{1,0}=200, s_{2,0}=100, \beta_1=0.99, \beta_2=0.4, A_1=1.5, A_2=2.25$. The credit limit of agent $2$ is $\gamma_2=0.4$. Let us denote $Y_t(A_1)$ the equilibrium aggregate output of the economy when the productivity of agent $1$ is $A_1$. 
The following graphics show how the difference between $Y_t(A_1+h)-Y_t(A_1)$ changes over time, where $h$ is a productivity change. 

First, when the productivity of agent $1$ increases from 1.5 to 1.53 (a small productivity change), the output goes down and then goes up. Precisely, $Y_t(1.5+0.03)-Y_t(1.5)<0$ for $t=1,2,3,4$ and then $Y_t(1.5+0.03)-Y_t(1.5)>0$, $\forall t\geq 5.$
\begin{figure}[h!]
\centering
   \includegraphics[width=7cm,height=5cm]{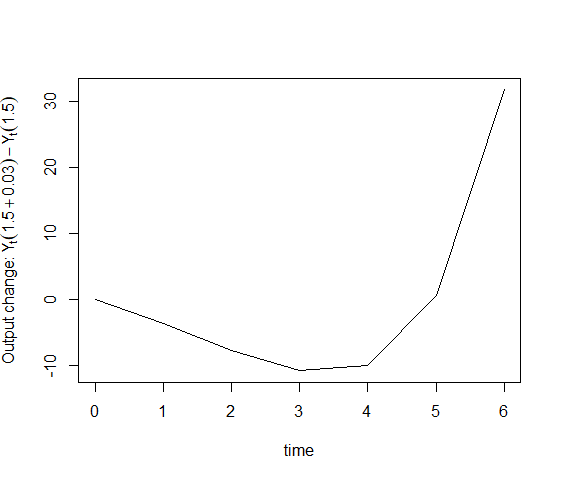}
    \includegraphics[width=7cm,height=5cm]{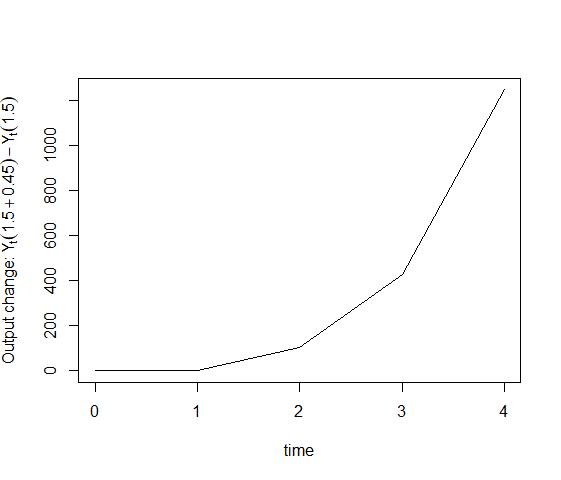}
\end{figure}

Second, when there is a high productivity change so that the productivity of agent $1$ increases from 1.5 to 1.95, the output goes up at any period: $Y_t(1.5+0.45)-Y_t(1.5)>0$, $\forall t\geq 1.$ This is consistent with the insights in Proposition \ref{asymmetry}.
} \end{example}

\subsubsection{Effect of temporary productivity changes}

 Let us look at the effects of temporary productivity changes. Assume that there is a productivity change only at date $1$, which affects the TFP of agent $h$. We would like to understand how the aggregate output changes when  $A_{h,1}$ varies. 
 

According to Lemma \ref{infinite-a11}, we have 
 \begin{align*}
\frac{\partial Y_t}{\partial A_{h,1}}=&A_{h,t} \cdots A_{h,2}\sum_{i\leq h}\beta_{i}^{t-1}s_{i,0}\\
&+A_{h,t}\cdots A_{h,2}\sum_{j>h}\beta_j^{t-1}(1-\gamma_j)^{t}\frac{A_{j,t}}{A_{h,t}-\gamma_jA_{j,t}}\cdots\frac{A_{j,2}}{A_{h,2}-\gamma_jA_{j,2}}\frac{\partial \big(\frac{A_{h,1}}{A_{h,1}-\gamma_jA_{j,1}}\big)}{\partial A_{h,1}}  A_{j,1}s_{j,0}\\
=&A_{h,t} \cdots A_{h,2}\sum_{i\leq h}\beta_{i}^{t-1}s_{i,0}\\
&-A_{h,t}\cdots A_{h,2}\sum_{j>h}\beta_j^{t-1}(1-\gamma_j)^{t}\frac{A_{j,t}}{A_{h,t}-\gamma_jA_{j,t}}\cdots\frac{A_{j,2}}{A_{h,2}-\gamma_jA_{j,2}}
\frac{\gamma_jA_{j,1}^2}{(A_{h,1}-\gamma_jA_{j,1})^2}
s_{j,0}
 \end{align*}

For the sake of simplicity, we focus on the case where $A_{i,t}=A_i$, $\forall t$, $\forall i\not=h$ and $A_{h,t}=A_h$, $\forall t\not=1$. We only let the productivity $A_{h,1}$ of agent $h$ at date $1$ vary. In this case, we can compute that    \begin{align}
&Y_t=A_{h}^{t-1}A_{h,1}\Big(\sum_{i \leq h}\beta_i^{t-1}
s_{i,0}+\sum_{j>h}\beta_j^{t-1}(1-\gamma_j)^{t}\big(\frac{A_{j}}{A_{h}-\gamma_jA_{j}}\big)^{t-1} \frac{A_{j,1}}{A_{h,1}-\gamma_jA_{j,1}} s_{j,0}\\
&\label{tempo1}\frac{\partial Y_t}{\partial A_{h,1}}=A_{h}^{t-1} \Big(\sum_{i\leq h}\beta_{i}^{t-1}s_{i,0}-\sum_{j>h}\beta_j^{t-1}(1-\gamma_j)^{t}\big(\frac{A_{j}}{A_{h}-\gamma_jA_{j}}\big)^{t-1}
\frac{\gamma_jA_{j,1}^2}{(A_{h,1}-\gamma_jA_{j,1})^2}
s_{j,0}\Big).
 \end{align}
We can also prove that the growth rate $\frac{Y_{t+1}}{Y_t}$ again converges to $\beta_hA_h$. However, the output can decrease when $A_{h,1}$ increases. According to (\ref{tempo1}), this happens  when the productivity gaps $\frac{A_j}{A_h}$ are high. The insights are consistent with the effects of permanent productivity shocks. 

\subsection{Effects of credit limits}
 \label{sec42}
In this section, we explore the effects of credit limits on the aggregate output in intertemporal equilibrium. To simplify our exposition, we focus on the case of stationary linear technology $F_i(k)=A_ik$. Since $A_1<A_2<\cdots<A_m$, the equilibrium interest rate is between $A_1$ and $A_m$. We distinguish two cases: (1) the interest rate equals the TFP of some producer and (2) the interest rate is between the TFPs of two producers. The following result considers an equilibrium in the first case.
\begin{proposition}\label{ahf}
Assume that the technology is stationary: $A_{i,t}=A_i,\forall i,\forall t$. Let assumptions in Lemma \ref{infinite-a11} be satisfied. Then, there exists an equilibrium with $R_t=A_{h},$ $\forall t$. In equilibrium, we have that:
\begin{enumerate}
\item  The aggregate output $Y_t$ is increasing in the credit limit $\gamma_j$ of agent $j$ for any $j>h$. Moreover, the output in  (\ref{outputt}) in the economy with credit constraints is lower that the output in the economy without credit constraints.

\item  However, the growth rate determined by (\ref{Gth}) is not necessarily increasing in the credit limit $\gamma_j$. It converges to $A_h\beta_h$ which  is  higher than $A_m\beta_m$ - the growth rate of the economy without credit constraint.

\end{enumerate}
\end{proposition}
\begin{proof}
Observe that $\frac{1-\gamma_j}{A_{h,t}-\gamma_jA_{j,t}}$ is increasing in $\gamma_j$.\footnote{Indeed, we have
\begin{align*}
\frac{\partial \big(\frac{1-\gamma_j}{A_{h,t}-\gamma_jA_{j,t}}\big)}{\partial \gamma_j}=\frac{-1(A_{h,t}-\gamma_jA_{j,t})+(1-\gamma_j)A_{j,t}}{(A_{h,t}-\gamma_jA_{j,t})^2}=\frac{A_{j,t}-A_{h,t}}{(A_{h,t}-\gamma_jA_{j,t})^2}>0.
\end{align*}} So, according to (\ref{linear-output}), the aggregate output $Y_t$ is increasing in each $\gamma_j,\forall j>h$.
\end{proof}
Point 1 is consistent with the insights in the literature concerning the macroeconomic effects of credit constraint (\cite{kt13} (section VI. C), \cite{mx14} (section II.B), \cite{moll14} (Proposition 1), and \cite{cchst17}).

However, the insight of point 2 of Proposition \ref{ahf} is new. It indicates when producers' credit limits are low, a rise in credit limit may decrease the growth rate of the economy. This is consistent with the empirical fact: the rate of growth of developing countries (with more severe credit constraints of firms) is in general higher than the grow rate of developed countries).

We now investigate a question: Along an intertemporal equilibrium, does relaxing credit limit always  improve or, in some cases, reduce the aggregate output? The full answer is complicated. The following result provides the first part of our answer: conditions (based on exogenous parameters) under which the aggregate output is a decreasing function of the credit limit.

\begin{proposition}[\textcolor{blue}{intertemporal equilibrium with $R_1\in (A_{m-1},A_m)$, $R_t=A_m, \forall t\geq 2$}]
\label{credit-infinite-magents}  Assume that $u_i(c)=ln(c)$, $\forall i,\forall c>0$, $F_{i,t}(k)=A_ik$, $\forall i,\forall k\geq 0$ with $\max_i{\gamma_iA_i}<A_1<A_2<\ldots<A_m$, and 
 \begin{subequations}\begin{align}
 \label{magentm1}
\frac{\sum_{i<m}\beta_i^ts_{i,0}}{\sum_{i<m}s_{i,0}}&\leq \beta_m^{t},\forall t,\\
\gamma_m&<\frac{\sum_{i\not=m}s_{i,0}}{S_0}\\
\frac{A_{m-1}}{A_m}&<\gamma_m\frac{S_0}{\sum_{i\not=m}s_{i,0}}.
\end{align} \end{subequations}
Then, there exists an equilibrium where the interest rates are determined by
\begin{align}
R_1&=\gamma_mA_m\frac{S_0}{\sum_{i\not=m}s_{i,0}} \in (A_{m-1},A_m), \quad 
R_t=A_m, \forall t\geq 2,
\end{align}
where $S_0\equiv \sum_{i=1}^ms_{i,0}$.

The aggregate capital is
 \begin{subequations}
\begin{align}
K_0&\equiv\s\sum_is_{i,0}=\sum_i\beta_iw_{i,0}\\
K_t&=k_{m,t}=S_0A_m^{t}\Big(\gamma_m\frac{\sum_{i\not=m}\beta_i^ts_{i,0}}{\sum_{j\not=m}s_{j,0}}+\beta_m^t(1-\gamma_m)\Big), \forall t\geq 1
\end{align}
 \end{subequations}
and the aggregate output
 \begin{subequations}\begin{align}
Y_1&=A_mk_{m,0}=A_mS_0\\
\label{ytm}Y_t&=A_mk_{m,t-1}=S_0A_m^{t}\Big(\gamma_m\frac{\sum_{i\not=m}\beta_i^{t-1}s_{i,0}}{\sum_{j\not=m}s_{j,0}}+\beta_m^{t-1}(1-\gamma_m)\Big).
\end{align}
 \end{subequations}
\end{proposition}
\begin{proof}
See Appendix \ref{extensions-proof}.
\end{proof}
In such an equilibrium, only the most productive agent produces. Notice that her borrowing constraint at date $1$ is binding but her borrowing constraints from date $2$ on are not necessarily binding.\footnote{In this equilibrium, borrowing constraint $R_{t+1}b_{m,t}\leq \gamma_mA_mk_{m,t}$ is equivalent to  $\sum_{i\not=m}\beta_i^t\frac{s_{i,0}}{\sum_{j\not=m}s_{j,0}}\leq \beta_m^{t}$.} From date $2$ on, the equilibrium interest rate equals the productivity of the most productive agent: $R_t=A_m,$ $\forall t\geq 2$. However, the interest rate between the initial date and date 1 equals $R_1$ which is lower than $A_m$ because the credit limit $\gamma_m$ of agent $m$ is not so high (in the sense that $\gamma_m<\frac{\sum_{i\not=m}s_{i,0}}{S_0}$) and the productivity gap is high (in the sense  that $\frac{A_{m-1}}{A_m}<\gamma_m\frac{S_0}{\sum_{i\not=m}s_{i,0}}$). Notice that \cite{k98}'s Section 2 only focuses on the case where the equilibrium interest rate equals the rate of return on investment of unproductive agents, i.e., $R_t=A_1$, $\forall t$.

We now look at the equilibrium aggregate output. 
\begin{enumerate}
\item 
First, according to Lemma \ref{perfecty}, the output in the economy without credit constraints  is $Y^*_t=A_m^{t}\sum_{i=1}^m\beta_i^{t-1}s_{i,0}$. So, the output at date 1 in our economy coincides to $Y^*_1$. However, we can verify, by using $\gamma_m<\frac{\sum_{i\not=m}s_{i,0}}{S_
0}$ and $\frac{\sum_{i<m}\beta_i^ts_{i,0}}{\sum_{i<m}s_{i,0}}\leq \beta_m^{t}$, that $Y_t<Y^*_t$ for any $t\geq 2$.\footnote{Indeed, we have $
Y_t= S_0A_m^{t}\Big(\frac{\sum_{i\not=m}s_{i,0}}{S_0}\frac{\sum_{i\not=m}\beta_i^{t-1}s_{i,0}}{\sum_{i\not=m}s_{i,0}}+\beta_m^{t-1}(1-\frac{\sum_{i\not=m}s_{i,0}}{S_0})\Big)=A_m^{t}\sum_{i=1}^m\beta_i^{t-1}s_{i,0}=Y^*_t$.} It means that, the output in the economy with credit constraints is lower than the output in the economy without credit constraints. This is consistent with the existing literature. 

\item  Second, according to (\ref{ytm}) and our assumption (\ref{magentm1}), we have that:
\textcolor{blue}{
\begin{align}\label{yt2}
\frac{\partial Y_t}{\partial \gamma_m}<0, \forall t\geq 2.
\end{align}}
It means that, from date $2$ on, the aggregate output  decreases when the most productive agent's credit limit increases.\footnote{Notice that the aggregate output at date $1$ does not depend on the credit limit $\gamma_m$ of the most productive agent. The equilibrium in Proposition \ref{credit-infinite-magents} does not depend on the credit limits $(\gamma_i)_{i<m}$ of less productive agents because these agents neither borrow nor produce.} This interesting result is new with respect to the standard view on the effects of financial constraints as shown in \cite{bs13}, \cite{kt13}, \citet{mx14},  \cite{moll14},  \cite{cchst17}. 

\end{enumerate}

Let us explain the intuition of our finding (\ref{yt2}). 
Denote $W_{i,t}\equiv F_i(k_{i,t-1})-R_tb_{i,t-1}$ the net worth of agent $i$ at date $t$.   In equilibrium in Proposition \ref{credit-infinite-magents}, the net worth of the most productive agent is given by
\begin{align*}W_{m,1}&=A_mk_{m,0}-R_1b_{m,0}=(1-\gamma_m)\sum_is_{i,0}\\
\text{$t\geq 2$: }W_{m,t}&=A_mk_{m,t-1}-R_tb_{m,t-1}=A_ms_{m,t-1}=A_m(\beta_mA_m)^{t-1}s_1\\
&=(\beta_mA_m)^t(1-\gamma_m)A_mS_0.
\end{align*}
where we denote the individual saving: $s_{i,t}\equiv k_{i,t}-b_{i,t}$.

We see that the net worth is decreasing in the credit limit $\gamma_m$. The reason behind is that when $\gamma_m$ increases, the interest rate $R_1$ goes up which makes the repayment $R_1b_{m,0}$ increase. However, the capital $k_{m,0}$ of agent $m$ is already equal to the aggregate savings $\sum_is_{i,0}$ which can no longer increase. By consequence, the net worth $W_{m,1}=A_mk_{m,0}-R_1b_{m,0}$ decreases. 
This makes the saving of agent $m$ go down, and, hence, the output decreases. The mechanism can be summarized by the following schema:
\begin{align}
&\text{\bf Credit limit } \gamma_m 	\quad \uparrow \quad  \Rightarrow  \text{\bf Interest rate  } \uparrow \quad  \Rightarrow  \text{\bf Agent m's net worth  } \downarrow   \quad \Rightarrow \notag\\
 & \Rightarrow  \text{\bf Saving  } \downarrow \quad  \Rightarrow  \text{\bf Production  } \downarrow   \quad  \Rightarrow  \cdots
\end{align}
However, this mechanism does not happen when the credit limit $\gamma_m$ of agent $m$ is high enough (if this happens, we recover the equilibrium in part 2 of Lemma \ref{infinite-a11} where the output of our economy coincides to the output of the economy without credit constraints).

In Proposition \ref{credit-infinite-magents}, the most productive agent is the unique producer at date 1 and, thanks to this, the output at date $1$ equals the output in the economy without credit constraints. When there are more than 2 producers, the effects of credit limits $(\gamma_i)$ of different agents become more interesting. We attempt to understand what would happen in this case. Let us start with an intermediate step.
\begin{lemma}[\textcolor{blue}{intertemporal equilibrium with $R_1\in (A_{n-1},A_{n}), R_t=A_h$, $\forall t\geq 2$, $h\geq n$}]
\label{credit-infinite-nh}
 Assume that $u_i(c)=ln(c)$, $\forall i,\forall c>0$, $F_{i,t}(k)=A_ik$, $\forall i,\forall k\geq 0$ with $\max_i{\gamma_iA_i}<A_1<A_2<\ldots<A_m$, and 
\begin{align}
\label{nh1}
\sum_{j\geq n}\frac{\gamma_jA_{j}}{A_{n}-\gamma_jA_{j}}s_{j,0}<&\sum_{i<n}s_{i,0}<\sum_{j\geq n}\frac{\gamma_jA_{j}}{A_{n-1}-\gamma_jA_{j}}s_{j,0}\\
\label{nh2}\beta_h^{t}\frac{A_h}{R_1-\gamma_hA_{h}}s_{j,0}\geq &\sum_{i<n}\beta_i^t s_{i,0}+\sum_{n\leq i\leq h}\beta_j^{t}(1-\gamma_j)\frac{A_j}{R_1-\gamma_jA_{j}}s_{j,0}\notag\\
&-\sum_{j>h}\Big(\beta_j\frac{(1-\gamma_j)A_{j}}{A_h-\gamma_jA_{j}}\Big)^t \frac{\gamma_jA_j}{R_1-\gamma_jA_j} s_{j,0}\quad \geq 0,\quad  \forall t\geq 1.
\end{align}
for some agent $h$ with $n\leq h\leq m$. Then, there exists an equilibrium  with the interest rates
\begin{subequations}
\begin{align}
\label{ahR1}R_1&\in (A_{n-1},A_n)  \text{ is determined by  }\sum_{i<n}s_{i,0}=\sum_{j\geq n}\frac{\gamma_jA_{j}}{R_{1}-\gamma_jA_{j}}s_{j,0}\\
\label{ahRt}R_t&=A_h, \forall t\geq 2.
\end{align}\end{subequations}
\end{lemma}

\begin{proof}
See Appendix \ref{extensions-proof}.
\end{proof}
Condition (\ref{nh1}) ensures that the equilibrium interest rate $R_1$ is determined by (\ref{ahR1}) while (\ref{nh2}) ensures (\ref{ahRt}). The first inequality  in (\ref{nh2}) means that the borrowing constraint of agents $h$ are satisfied while the second inequality is equivalent to  $k_{h,t}\geq 0$. Note that condition (\ref{nh2}) requires  that\footnote{Because $\max\Big(\beta_h,\max_{j>h}\beta_j\frac{(1-\gamma_j)A_{j}}{A_h-\gamma_jA_{j}}\Big)\geq \max_{i<h}\beta_i$ and $\max_{i\leq h}\beta_i\geq \max_{j>h}\beta_j\frac{(1-\gamma_j)A_{j}}{A_h-\gamma_jA_{j}}$.} 
\begin{align}
\beta_h\geq \max\big(\max_{i<h}\beta_i,\max_{j>h}\beta_j\frac{(1-\gamma_j)A_{j}}{A_h-\gamma_jA_{j}}\big)>\max_{j>h}\beta_i.
\end{align}
So, agent $h$ has the highest discount factor.

 In such an equilibrium in Lemma \ref{credit-infinite-nh}, the capital of producers and the aggregate output are determined by
\begin{align}\label{capital-nhm}
k_{j,0}&=\begin{cases}0, &\forall j<n\\
 \frac{R_1}{R_1-\gamma_jA_{j}}s_{j,0}, &\forall j\geq n
\end{cases}\\
k_{j,t}&=\begin{cases}0, &\forall j<h\\
\sum_{i<n}\beta_i^t A_h^{t-1} R_{1}s_{i,0}+\sum_{n\leq j\leq h}\beta_j^{t}A_h^{t-1}(1-\gamma_j)\frac{A_jR_1}{R_1-\gamma_jA_{j}}s_{j,0}\\
\quad \quad \quad \quad \quad -\sum_{j>h}A_h^{t-1}\Big(\beta_j\frac{(1-\gamma_j)A_{j}}{A_h-\gamma_jA_{j}}\Big)^t \frac{\gamma_jA_jR_1}{R_1-\gamma_jA_j} s_{j,0}, &\text{for } j=h,\\
\frac{A_h}{A_h-\gamma_jA_{j}}\Big(\beta_j\frac{(1-\gamma_j)A_{j}A_h}{A_h-\gamma_jA_{j}}\Big)^{t-1}\Big(\beta_j\frac{(1-\gamma_j)A_{j}R_1}{R_1-\gamma_jA_{j}}\Big) s_{j,0}, &\forall j>h.
\end{cases}
\end{align}
We are now ready to state our result showing the effects of credit limits.
\begin{proposition}\label{ytanalysis-infinite}
Let assumptions in Lemma \ref{credit-infinite-nh} be satisfied. \\
\textcolor{blue}{1. For date $t=1$}, the aggregate output equals $Y_1=\sum_{j\geq n}A_jk_{j,0}$.
\begin{enumerate}
\item[1.1.] $
\frac{\partial Y_j}{\partial \gamma_{n}}< 0< \frac{\partial Y_j}{\partial \gamma_m}$  if $n<m$.\footnote{Moreover, if $n=m$ (i.e., only agent $m$ produces), we have $\frac{\partial Y}{\partial \gamma_m}= 0$.}

\item[1.2.] Consider any producer $i$ with $n< i< m$, we have that: 
\begin{subequations}
\begin{align}
\label{ynfi1}\frac{\partial Y_1}{\partial \gamma_{i}}&>0 \text{ if  $A_i$ is high enough, i.e., } \frac{A_i-A_{i-1}}{A_m-A_i}>\frac{\sum_{j=i+1}^m\frac{\gamma_{j}A_js_{j,0}}{(A_{n-1}-\gamma_{j}A_j)^2}}{\sum_{j=n}^{i-1}\frac{\gamma_{j}A_js_{j,0}}{(A_{n}-\gamma_{j}A_j)^2}}\\
\label{ynfi2}\frac{\partial Y_1}{\partial \gamma_{i}}&<0 \text{ if $A_i$ is low enough, i.e., } \frac{A_i-A_{n}}{A_{i+1}-A_i}<\frac{\sum_{j=i+1}^m\frac{\gamma_{j}A_js_{j,0}}{(A_{n}-\gamma_{j}A_j)^2}}{\sum_{j=n}^{i-1}\frac{\gamma_{j}A_js_{j,0}}{(A_{n-1}-\gamma_{j}A_j)^2}}.
\end{align}
\end{subequations}
\end{enumerate}
\textcolor{blue}{2. From second date ($t\geq 2$)}.\\
2.1. For $v\in \{n,\ldots,h\}$, this agent produces only at date 1. We have that 
\begin{align*}
\frac{1}{A_h^t}\dfrac{\partial Y_{t+1}}{\partial \gamma_v}\frac{1}{\frac{\partial R_1}{\partial \gamma_v}}&=\sum_{i<n}\beta_i^t s_{i,0}-\sum_{n\leq j\leq h}\beta_j^{t}\frac{(1-\gamma_j)\gamma_jA_j^2}{(R_1-\gamma_jA_{j})^2}s_{j,0}-
\sum_{j> h}\Big(\beta_j\frac{(1-\gamma_j)A_{j}}{A_h-\gamma_jA_{j}}\Big)^{t}\frac{(1-\gamma_j)\gamma_jA_j^2}{(R_1-\gamma_jA_{j})^2}s_{j,0}
\Big)\\
&+\beta_v^{t}(A_v-R_1)\sum_{j\geq n}\frac{\gamma_jA_{j}}{(R_{1}-\gamma_jA_{j})^2}s_{j,0}, \forall t\geq 1.
\end{align*}
By consequence, if 
$\beta_j>\max_{i\not=j}\beta_i$, then there exists $t_0$ such that $\dfrac{\partial Y_{t+1}}{\partial \gamma_v}<0,\forall t\geq t_0$.

2.2. For agent $v>h$, this agent produces at any date. We have, for any $t\geq 1$, that
\begin{align*}
\frac{1}{A_h^t}\dfrac{\partial Y_{t+1}}{\partial \gamma_v}\frac{1}{\frac{\partial R_1}{\partial \gamma_v}}&=\sum_{i<n}\beta_i^t s_{i,0}-\sum_{n\leq j\leq h}\beta_j^{t}\frac{(1-\gamma_j)\gamma_jA_j^2}{(R_1-\gamma_jA_{j})^2}s_{j,0}-
\sum_{j> h}\Big(\beta_j\frac{(1-\gamma_j)A_{j}}{A_h-\gamma_jA_{j}}\Big)^{t}\frac{(1-\gamma_j)\gamma_jA_j^2}{(R_1-\gamma_jA_{j})^2}s_{j,0}
\Big)\\
&+\Big(\frac{\beta_v(1-\gamma_v)A_v}{A_h-\gamma_vA_v}\Big)^t\Big(\frac{t(R_1-\gamma_vA_v)(A_v-A_h)}{(A_h-\gamma_vA_v)}+A_v-R_1\Big)\Big(\sum_{j\geq n}\frac{\gamma_jA_{j}}{(R_{1}-\gamma_jA_{j})^2}s_{j,0}\Big)
\end{align*}
By consequence, if $\beta_h>\max\{\frac{\beta_v(1-\gamma_v)A_v}{A_h-\gamma_vA_v},\max_{i<n}\beta_i\}$, then there exists $t_0$ such that $\dfrac{\partial Y_{t+1}}{\partial \gamma_v}<0,\forall t\geq t_0$.
 
\end{proposition}
Proposition \ref{ytanalysis-infinite}  allows us to understand why the aggregate output is decreasing or increasing in the credit limits of producers. It depends not only on the distribution of productivity and of credit limits but also on the distribution of initial capital of agents. 

Since the insight of part 1 is similar to Proposition \ref{yn-fi-co} in the two-period model, let us explain the intuition of part 2. Note that the aggregate output does not depend on the credit limits of non-producers. So, we only look at the producers in equilibrium. At date 1, producers are any agent $v\geq n$. From date $2$ on, producers are any agent $v\geq h$.  In both cases of part 2, from some date on, the output will be decreasing in the credit limit of any producer if the discount factor $\beta_h$ is high. This finding is consistent with (\ref{yt2}). The basic intuition behind is the input is used by less productive agents. Indeed, at the date $1$, because low credit limits and high productivity dispersion (see condition (\ref{nh1})), we have $R_1<A_n<A_h$, so we have a capital misallocation. When agent $h$ has the highest discount factor $\beta_h$, this agent absorbs capital in the long run which makes the misallocation persistent over time and the output decrease.


\begin{remark}[additional analyses] In Appendix \ref{additional}, we present two additional results.  First, Proposition \ref{m-1mh} shows that the aggregate output is increasing in the credit limits of producers for the case $R_1\in (A_{m-1},A_m), R_t=A_h, \forall t\geq 2$, with $h<m$. 

Second, Proposition \ref{m-1m} provides conditions under which there exists an equilibrium with $R_1\in (A_{m-1},A_m), \forall t\geq 1$. In this case, there is only one producer in equilibrium and the output is increasing in the credit limit of this agent. 

The intuition behind these two results is that the equilibrium interest rate is not so high low (it is lower that $A_m$). Hence, the borrowing cost of producers is not so high. This helps producers borrow more and produce more.
\end{remark}

\begin{example}
We complement our theoretical result by a numerical simulation (Figure \ref{effect-f2f3}). Consider a model with 3 agents. In this simulation, we set that $\beta_1=0.2, 
\beta_2=0.2, \beta_3= 0.95,$
$s_{1,0}=4=\beta_1w_{1,0}, s_{2,0}=4=\beta_2w_{2,0}, s_{3,0}=3=\beta_3w_{3,0},$ $\gamma_1=0.2, \gamma_3=0.3$. Productivity: $A_1=1, A_2=1.2, A_3=1.5$.  We draw the output path for two cases: $ \gamma_2=0.3$ and $ \gamma_2=0.35$. 
We observe that 
\begin{align*}
Y_t(\gamma_2=0.35)<Y_t(\gamma_2=0.30), \forall t \geq 1.
\end{align*}
It means that when the credit limit $\gamma_2$ of agent 2 increases from 0.30 to 0.35, the aggregate output will be lower at any period of time.

\begin{figure}[h!]
\centering
      \includegraphics[width=11cm,height=8cm]{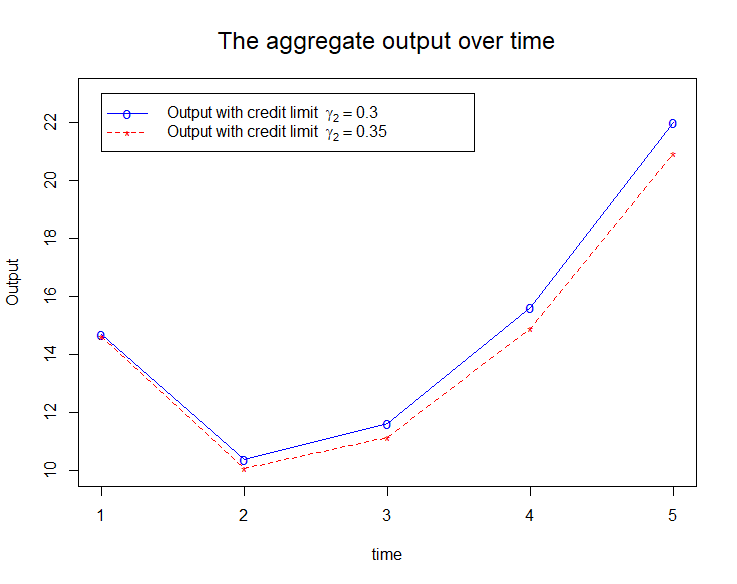}
  \caption{Effects of credit limits $\gamma_2$ on the aggregate output.}
  \label{effect-f2f3}
\end{figure}
 \end{example}

\section{Conclusion} \label{conclu}

We have build general equilibrium models with borrowing constraints to explain why the aggregate output may be decreasing (increasing, respectively) when the productivity or credit limit of producers increases (decreases, respectively).  A positive homogeneous (productivity or financial) shock has a positive impact on the aggregate output. This is consistent with the insights in economic textbooks and several articles. Our new insight is that positive asymmetric (productivity or financial) shocks may reduce the aggregate production. Overall, not only productivity but also financial frictions and the productivity gap (or dispersion of productivity distribution) matter for the economic development.

The contribution of the present paper is primarily theoretical. A promising avenue for future research would be to develop a quantitative model calibrated with empirical data to reassess the effects of asymmetric (productivity and financial) shocks and the persistence of shocks on equilibrium dynamics.



\appendix
{\small
\section*{Appendices}
\section{Proofs for Section  \ref{productivity}}
\label{proof3.1}

  \subsection{Characterization of general equilibrium}
  \subsubsection{Linear technology}
When the production functions are linear, it is easy to compute the optimal allocation of agents as a function of the interest rate (see Lemma  \ref{lem1} in Appendix \ref{append-linear}). Therefore, the key point is to determine the equilibrium interest rate. To state our characterization of equilibrium, we introduce some notations. 
\begin{align}
\label{bdn}\mathbb{D}_{n}&\equiv \sum_{i=n}^{m}\frac{A_{n}S_i}{A_{n}-\gamma_iA_i} \text{ }\forall n\geq 1, \quad
\mathbb{B}_n \equiv \sum_{i=n+1}^{m}\frac{A_nS_i}{A_n-\gamma_iA_i} \text{ }\forall n\geq 1.
\end{align}
where by convention,  $\sum_{i=n}^mx_i=0$ if $n>m$.
 
Denote $R^L_n$ the greatest solution of the following equation:\footnote{It should be noticed that the function $f(x)\equiv \sum_{i=n+1}^m\frac{xS_i}{x- \gamma_iA_i}$ is not continuous at point $\gamma_iA_i$ with $i\geq n+1$.   However, it is continuous and decreasing in the interval $(\max_{i\geq n+1}(\gamma_iA_i),\infty)$. Then, the equation $f(x)=S$ has a unique solution in such interval.}
\begin{align}\label{rn} 
\underbrace{\sum_{i=n+1}^m\dfrac{\gamma_iA_i}{R- \gamma_iA_i}S_i}_{\text{Asset demand}}=\underbrace{\sum_{i=1}^nS_i}_{\text{Asset supply}}  \text{ or equivalently }
 \underbrace{\sum_{i=n+1}^m\dfrac{RS_i}{R- \gamma_iA_i}}_{\text{Capital demand}}=\underbrace{S}_{\text{Capital supply}}
\end{align} 

\begin{definition}
\begin{enumerate}
\item 
the regime $\mathcal{A}_n$ (with $n\in \{1,\ldots,m\}$) is the set of all economies satisfying $A_n>\max_i(\gamma_iA_i)$ and $\mathbb{B}_n\leq S\leq \mathbb{D}_n$
\item  the regime $\mathcal{R}_n$ (with $n\in \{1,\ldots,m-1\}$) is the set of all economies satisfying 
\begin{enumerate}
\item 
either $\max_i(\gamma_iA_i)< A_n< R^L_n< A_{n+1}$ (or equivalently $\max_i(\gamma_iA_i)< A_n$ and $\mathbb{D}_{n+1}<S<\mathbb{B}_n$) 
\item or  $A_n\leq \max_i(\gamma_iA_i)<R^L_n<A_{n+1}$ (or equivalently $A_n\leq \max_i(\gamma_iA_i)<R^L_n$ and $\mathbb{D}_{n+1}<S$).
\end{enumerate}
\end{enumerate}
\end{definition}
We now provide a characterization of general equilibrium.
\begin{theorem}[characterization of general equilibrium with  linear technologies]
\label{cate-r} Assume that $F_i(K)=A_iK$ $\forall i$ and  $A_1<\cdots<A_m$. 
Then, there exists a unique equilibrium. The equilibrium interest rate is determined by the following:
\begin{align}
\label{r-linear-deter}
R=\begin{cases}A_i &\text{ in the regime } \mathcal{A}_i.\\
R^L_i &\text{ in the regime } \mathcal{R}_i.
\end{cases}
\end{align}

\end{theorem}
\begin{proof}See  Appendix \ref{append-linear}.  
  \end{proof}

 \subsubsection{Strictly concave technology}

Before providing the characterization of equilibrium, we state an assumption about the credit limit.
\begin{assum}\label{gammaalpha-assum}$\gamma_i<\lim_{k \to \infty}\frac{kf_i'(k)}{f_i(k)}$, $\forall i$.
\end{assum}
As proved in Lemma \ref{gammaalpha} in Appendix \ref{general1-proof}, under Assumptions \ref{assum-concave} and \ref{Hi}, if agent i's borrowing constraint is binding, we must have $\gamma_i\leq \lim_{x\rightarrow \infty}\frac{xF_i'(x)}{F_i(x)}$. 

We are now ready to state the characterization of equilibrium.
\begin{theorem}[characterization of general equilibrium: strictly concave technologies]
\label{general1}Under Assumption \ref{assum-concave}, there exists a unique equilibrium. Assume, in addition, that Assumption \ref{Hi} and \ref{gammaalpha-assum} hold and $R_1<R_2<\ldots<R_m$,  where $R_i$ is the unique value satisfying 
\begin{align}\label{Ri}
H_i(R_i)\equiv R_i \frac{k^n_i(R_i/A_i)-S_i}{A_if_i(k^n_i(R_i/A_i))}= \gamma_i.
\end{align}
Then the unique equilibrium is determined as follows:
\begin{enumerate}
\item In the regime $\mathcal{R}_m$, i.e., when $
S<\sum_{i=1}^mk^n_i(R_m/A_i)$, 
 credit constraint of any agent is not binding. In this case, the equilibrium coincides to that of the economy without credit constraints, and the interest rate is $R=R^*>R_m$. Agent $i$ borrows ($k_i\geq S_i$) if and only if $F'_i(S_i)\geq R^*$.

\item In the regime $\mathcal{R}_n$ (with $1\leq n\leq m-1$), i.e., when 
\begin{align*}
\sum_{i=1}^nk^n_i(\frac{R_n}{A_i})+\sum_{i=n+1}^mk^b_i(\frac{R_n}{\gamma_iA_i},S_i) >  S\geq \sum_{i=1}^{n+1}k^n_i(\frac{R_{n+1}}{A_i})+\sum_{i=n+2}^mk^b_i(\frac{R_{n+1}}{\gamma_iA_i},S_i),*\end{align*}
then the equilibrium interest rate  is determined by  the following equation
\begin{align}
\label{rn-concave}
 \sum_{i=1}^nk^n_i(\frac{R}{A_i})+\sum_{i=n+1}^mk^b_i(\frac{R}{\gamma_iA_i},S_i)=S\equiv \sum_iS_i
\end{align}
while agents' capital is
\begin{align*}k_i=
\begin{cases}k^n_i(\frac{R}{A_i})\quad \text{ if } i\leq n \\
k^b_i(\frac{R}{\gamma_iA_i},S_i)\quad \text{ if } i\geq n+1.
 \end{cases}
\end{align*}
Notice that $R_n<R\leq R_{n+1}$ in this case. Any agent $i$ $(i\geq n+1)$ borrows and her credit constraint  is binding. The credit constraint of any agent $i\leq n$ is not binding. Moreover, agent $i$ $(i\leq n)$ borrows if and only if $F_i'(S_i)\geq  R$. 
\end{enumerate}
\end{theorem}
\begin{proof}See  Appendix \ref{general1-proof}.
\end{proof}

\begin{proof}[{\bf Proof of Proposition \ref{A1effect-general-manyagents}}]
{\bf Part 1}. Point (a) is a direct consequence of Lemma \ref{gammaalpha}  in Appendix \ref{general1-proof}. Point (b) is a direct consequence of  Theorem \ref{general1}. 

{\bf Part 2}.  Since the production functions satisfy Inada's condition, all agents produce in equilibrium.  According to (\ref{yaj}), we have
\begin{align}\label{ya1n} \frac{\partial Y}{\partial A_1}=&\underbrace{f_1(k_1)}_\text{Productivity effect}+ \underbrace{\sum_{i\not=1}\big(A_1f_1^\prime(k_1)-A_if_i^\prime(k_i)\big)\underbrace{\frac{-\partial k_i}{\partial R}}_\text{$\geq 0$}\underbrace{\frac{\partial R}{\partial A_1}}_\text{$\geq 0$}}_\text{Allocation effect}.\end{align}

According to FOCs, we have
\begin{align*}
[k]&: (1 +\mu_i \gamma_i)F_i'(k) =\lambda_i\\
[a]&: (1+\mu_i)R = \lambda_i, \quad \mu_i \geq 0, \text{ and } \mu_i(\gamma_i  F_i(k_i)-R_ib_i)=0.
\end{align*}
These equations imply that:
\begin{equation}
\gamma_iA_if_i'(k_i)\leq R=A_if_i'(k_i)\dfrac{1+\gamma_i \mu_i}{1+\mu_i}\leq A_if_i'(k_i), \forall i.
\end{equation}
This implies that $R\geq \max_j\gamma_jF_j'(k_j)\geq  \max_j\gamma_jF_j'(S))$.  Thus, $R\geq \max_j\gamma_jF_j'(S))>0$, $\forall A_1$.

\begin{enumerate}
\item When $A_1$ is high enough. Note that $\lim_{A_1\to \infty}R_1=\infty$. Hence, for $A_1$ high enough, we have that $R_1>S$. We prove that the equilibrium interest rate goes to infinity when $A_1$ goes to infinity. Indeed, if agent 1's borrowing constraint is not binding, we have $R=A_1f_1'(k_1)>A_1f_1'(S)$. If agent 1's borrowing constraint is  binding, we have $R(k_1-S_1)=\gamma_1A_1f_1(k_1)$ which implies that 
\begin{align*}
R=\frac{\gamma_1A_1f_1(k_1)}{k_1-S_1}\geq \frac{\gamma_1A_1f_1(S_1)}{S-S_1}
\end{align*}
Hence, $R\geq \min \big(A_1f_1'(S),\frac{\gamma_1A_1f_1(S_1)}{S-S_1}\big)$. From this, we obtain that $\lim_{A_1\to \infty}R=\infty$.

Now, condition $\lim_{A_1\to \infty}R=\infty$ implies that borrowing constraint of any agent $i\geq 2$ is not binding for $A_1$ high enough. So, $A_1f_1'(k_1)\geq R=A_if_i'(k_i)$, $\forall i\geq 1$. By combining this and condition (\ref{ya1n}), we get that $\frac{\partial Y}{\partial A_1}>0$ for $A_1$ high enough.

\item We will prove that when $A_1$ is small enough, the productivity effect is smaller than the allocation effect. 
To show $\frac{\partial Y}{\partial A_1}<0$ for $A_1$ small enough, we will prove that $\lim_{A_1\to 0}k_1=0$, $\lim_{A_1\to 0}A_1f_1^\prime(k_1)-A_2f_2^\prime(k_2)<0$, $\lim_{A_1\to 0}\frac{-\partial k_2}{\partial R}>0$, and $\lim_{A_1\to 0}\frac{\partial R}{\partial A_1}>0$. 

 Since $A_1f_1'(k_1)\geq R\geq \max_j\gamma_jF_j'(S))>0$, we have $\lim_{A_1\to 0}f_1'(k_1)=\infty$. Therefore, we have 
\begin{align}\label{i}
\lim_{A_1\to 0}k_1=0, \text{ and } \lim_{A_1\to 0}\sum_{i\not=1}k_i=S.
\end{align}

Since $\lim_{A_1\to 0}k_1=0$, we get that 
$\gamma_1A_1f_1(k_1)-Rb_1=\gamma_1A_1f_1(k_1)-Rk_1+RS_1>0$ for $A_1$ small enough. It means that the borrowing constraint of agent $1$ is not binding. To sum up, we have $$R= A_1f_1'(k_1)\geq \max_j\gamma_jF_j'(S))>0, \text{ for $A_1$ small enough}.$$


Denote
\begin{align*}
B_1&=B_1(R_1)\equiv k^n_1(R_1)+\sum_{i=2}^mk^b_i(R_1),\quad 
B_2=B_2(R_2)=\sum_{i=1}^2k^n_i(R_2)+\sum_{i=3}^mk^b_i(R_2)\\
B_m&=B_m(R_m)=\sum_{i=1}^mk^n_i(R_m)
\end{align*}
\textcolor{blue}{where, to simplify notations, we write $k^n_i(R)$ and $k^b_i(R)$ instead of $k^n_i(\frac{R}{A_i})$ and  $k^b_i(\frac{R}{\gamma_iA_i}, S_i)$} (see Definition \ref{knb}). 
We see that $D_i\equiv B_i-k^n_i(R_i)$, $\forall i$. Notice also that $B_1,\ldots,B_m$ depend on $A_1$ but  $D_2,D_3,\ldots, D_m$ do not. Moreover, $\lim_{A_1\to 0}(B_i-D_i)=0$, $\forall i\geq 2$ because $\lim_{A_1\to 0}k^n_1(R_i)=0$, $\forall i\geq 2$.

 Condition $R_2<R_3<\cdots <R_m$ implies that $D_2>\cdots >D_m$. Since $\lim_{A_1\to 0}B_1=+\infty$ and $\lim_{A_1\to 0}(B_i-D_i)=0$, $\forall i\geq 2$, we have $B_1>B_2>\cdots >B_m$ for $A_1$ small enough.
\begin{enumerate}
\item $S<D_m$. Then we have $S<B_m$. According to Theorem \ref{general1},  the equilibrium coincides to that of the economy without frictions. Therefore, the output is increasing in $A_1$.  
\item Let $D_n>S>D_{n+1}$. In this case, we have $B_n>S>B_{n+1}$ for any $A_1$ small enough. According to Theorem \ref{general1}, the equilibrium interest rate $R$ is in the interval $(R_n,R_{n+1}]$ and determined by 
\begin{align}
\label{rn-concave0}
 \sum_{i=1}^nk^n_i(R)+\sum_{i=n+1}^mk^b_i(R)=S\equiv \sum_iS_i
\end{align}
Denote $Z_2(R)=\sum_{i=2}^nk^n_i(R)+\sum_{i=n+1}^mk^b_i(R)$. When $A_1$ tends to zero, we have $\lim_{A_1\to 0} k^n_1(R)=0$ and $\lim_{A_1\to 0}R=R(0)$ where $R(0)>0$ is uniquely determined by $Z_2(R(0))=S$.

For $i\geq n+1$, agent $i$'s borrowing constraint is binding: $R(k_i-S_i)=\gamma_iA_if_i(k_i)$ for any $A_1$ small enough. Let $A_1$ tend to zero, we have $k_i$ tends to $k_i(0)$, $R$ tends to $R(0)$, and $$\gamma_iA_if_i(k_i(0))=R(0)(k_i(0)-S_i).$$
Let $\sigma $ be such that
\begin{align}\label{g11}
\gamma_i \frac{f_i(k)}{kf_i'(k)}<\sigma<\frac{S_1}{S_1+S_{n+1}+\cdots +S_m}, \forall i\geq n+1, \forall k\in(0,S).
\end{align}

  According to condition (\ref{g1}), we have
\begin{align}\label{ii}
R(0)-A_if_i^\prime(k_i(0))&=\frac{\gamma_iA_if_i(k_i(0))}{k_i(0)-S_i}-A_if_i^\prime(k_i(0))\\
&\leq \frac{A_if_i^\prime(k_i(0))}{k_i(0)-S_i}(\sigma k_i(0)-(k_i(0)-S_i)\big)
\end{align}

 By market clearing condition, we have
 \begin{align*}
\sum_{i=n+1}^mk_i=\sum_{i=2}^m(S_i-k_i)+S_1-k_1+\sum_{i=n+1}^mS_i\geq S_1-k_1+\sum_{i=n+1}^mS_i
 \end{align*}
 Let $A_1$ tend to zero, we get that $\sum_{i=n+1}^mk_i(0)\geq  S_1+\sum_{i=n+1}^mS_i$. Thus,
 \begin{align*}
&\sum_{i=n+1}^m \Big(\sigma k_i(0)-(k_i(0)-S_i)\Big)=\sum_{i=n+1}^m\Big(S_i-(1-\sigma) k_i(0)\Big)\\
&\leq \sum_{i=n+1}^mS_i -(1-\sigma)\big(S_1+\sum_{i=n+1}^mS_i\big)<0
 \end{align*} 
  Therefore, there exists $j\in \{n+1,\ldots,m\}$ such that $\sigma k_j(0)-(k_j(0)-S_j)<0$, and hence 
\begin{align}\label{iiig}
R(0)-A_jf_j^\prime(k_j(0))&\leq \frac{A_jf_j^\prime(k_j(0))}{k_j(0)-S_j}(\sigma k_j(0)-(k_j(0)-S_j)\big)<0.
\end{align}
Now, by noting that $A_1f'(k_1)=R$, we have
\begin{align}\label{ya1n2} \frac{\partial Y}{\partial A_1}\leq &f_1(k_1) + \big(R-A_jf_j^\prime(k_j)\big) \frac{-\partial k_j}{\partial R} \frac{\partial R}{\partial A_1}
\end{align}

Again, by the market clearing condition 
\begin{align}
k^n_1(\frac{R}{A_1})+\sum_{i\not= 2}k_i(R)=S
\end{align}
we have that 
\begin{align}
(k^n_1)^{\prime}\big(\frac{R}{A_1}\big)\frac{R'(A_1)A_1-R}{A_1^2}+\sum_{i\not=2}\frac{\partial k_i}{\partial R}R'(A_1)&=0\\
\Leftrightarrow R'(A_1)\Big(\frac{1}{A_1}(k^n_1)^{\prime}\big(\frac{R}{A_1}\big)+\sum_{i\not=2}\frac{\partial k_i}{\partial R}\Big)&=(k^n_1)^{\prime}\big(\frac{R}{A_1}\big)\frac{R}{A_1^2}\notag\\
\Leftrightarrow R'(A_1)A_1\Big(\frac{1}{R}+\frac{A_1}{R}\frac{\sum_{i\not=2}\frac{\partial k_i}{\partial R}}{(k^n_1)^{\prime}\big(\frac{R}{A_1}\big)}\Big)&=1.\label{Rprimea3}
\end{align}
Since $\frac{\partial k_i}{\partial R}<0$, $\forall i\not=1$, and $(k^n_1)^{\prime}\big(\frac{R}{A_1}\big)<0$, we have $R'(A_1)>0$. 

By definition of $k^n_1$, we have $f_1'(k^n_1(x))=x$. So,  $(k^n_1)'(x)f_1^{\prime \prime}(k^n_1(x))=1$, and hence,  
$$\lim_{A_1\to 0}\frac{R}{A_1}(k^n_1)'(\frac{R}{A_1})=\lim_{A_1\to 0}\frac{\frac{R}{A_1}}{f_1^{\prime \prime}(\frac{R}{A_1})}=\lim_{x\to \infty}\frac{x}{f_1^{\prime \prime}(x)}< 0.$$
By combining this with (\ref{Rprimea3}), $\lim_{R\to R(0)}\frac{\partial k_i}{\partial R}<0$, $\forall i$, and $\lim_{A_1\to 0}R=R(0)>0$, we get that
\begin{align}\label{iiigg}
\lim_{A_1\to 0}R'(A_1)=+\infty.
\end{align} 

By combining (\ref{ya1n2}), (\ref{iiig}), (\ref{iiigg}),  and $\lim_{R\to R(0)}\frac{\partial k_j}{\partial R}<0$, we get that $\frac{\partial Y}{\partial A_1}<0$ for any $A_1>0$ small enough.
\end{enumerate}

\end{enumerate}

\end{proof} 
 \subsection{Additional results}
\label{additional} 
In the case of a two-agent model, we have the following result with more details and intuitive conditions.
\begin{proposition}\label{A1effect-general}Consider a two-agent model. 
\begin{enumerate}
\item  Let Assumptions \ref{assum-concave}, \ref{Hi} and \ref{gammaalpha-assum}  be satisfied. Assume also that
\begin{align*}k^n_2(\frac{R_2}{A_2})<S, \quad 
\gamma_2<\frac{S_1}{S_1+S_2}\frac{Sf_2^\prime(S)}{f_2(S)}, \quad \lim_{x\to +\infty}\frac{x}{f_1^{\prime \prime}(x)}<0
\end{align*}
Then, for any $A_1$ small enough, we have that 
$\textcolor{blue}{
\frac{\partial Y}{\partial A_1}<0}$.
\item By consequence, in a two-agent economy with Cobb-Douglas production functions ($F_i(k)=A_ik^{\alpha}$) and  $\gamma_2< \alpha \frac{S_1}{S_1+S_2}$, we have that:  \textcolor{blue}{$
\frac{\partial Y}{\partial A_1}<0 \text{ for $A_1$ small enough}.$}

\end{enumerate}
\end{proposition}




\begin{proof}[{\bf Proof of Proposition \ref{A1effect-general}}]First, we state a corollary of Theorem \ref{general1}.
\begin{corollary}\label{charac2agents}Let Assumptions \ref{assum-concave}, \ref{Hi} and \ref{gammaalpha-assum} be satisfied. Consider a two-agent model and assume that $R_1<R_2$. 
\begin{enumerate}
\item In the regime $\mathcal{R}_2$, i.e., when $S<k^n_1(\frac{R_2}{A_1})+k^n_2(\frac{R_2}{A_2})$, credit constraint of any agent is not binding. 

\item In the regime $\mathcal{R}_1$, i.e., when $S\geq k^n_1(\frac{R_2}{A_1})+k^n_2(\frac{R_2}{A_2})$,\footnote{Notice that we always have that $k^n_1(R_1)=k^b_1(R_1)$,  $k^n_2(R_2)=k^b_2(R_2)$, and $k^n_1(R_1)+k^b_2(R_1)= k^b_1(R_1)+k^b_2(R_1)>S$.} 
the equilibrium interest rate $R$ is determined by 
\begin{align}
 k^n_1(\frac{R}{A_1})+k^b_2(\frac{R}{\gamma_2A_2},S_2)=S\equiv \sum_iS_i.
\end{align}
In this regime, $R_1<R\leq R_{2}$, agent 2 borrows and her credit constraint  is binding while agent $1$ is lender. 

\end{enumerate}

\end{corollary}

Now, we prove part 1 of Proof of Proposition \ref{A1effect-general}. Since Inada condition holds, all agents produce in equilibrium.  According to (\ref{yaj}), we have
\begin{align}\label{ya1} \frac{\partial Y}{\partial A_1}=&f_1(k_1)+ \big(A_1f_1^\prime(k_1)-A_2f_2^\prime(k_2)\big)\underbrace{\frac{-\partial k_2}{\partial R}}_\text{$>0$}\underbrace{\frac{\partial R}{\partial A_1}}_\text{$>0$}.\end{align}
To show $\frac{\partial Y}{\partial A_1}<0$ for $A_1$ small enough, we will prove that $\lim_{A_1\to 0}k_1=0$, $\lim_{A_1\to 0}A_1f_1^\prime(k_1)-A_2f_2^\prime(k_2)<0$, $\lim_{A_1\to 0}\frac{-\partial k_2}{\partial R}>0$, and $\lim_{A_1\to 0}\frac{\partial R}{\partial A_1}>0$. 

According to FOCs, we have
\begin{align*}
[k]&: (1 +\mu_i \gamma_i)F_i'(k) =\lambda_i\\
[a]&: (1+\mu_i)R = \lambda_i, \quad \mu_i \geq 0, \text{ and } \mu_i(\gamma_i  F_i(k_i)-R_ib_i)=0.
\end{align*}
These equations imply that:
\begin{equation}
\gamma_iA_if_i'(k_i)\leq R=A_if_i'(k_i)\dfrac{1+\gamma_i \mu_i}{1+\mu_i}\leq A_if_i'(k_i), \forall i.
\end{equation}
This implies that $R\geq \gamma_2F_2'(k_2)\geq \gamma_2F_2'(S)$.  Thus, $R\geq \gamma_2F_2'(S)$, $\forall A_1$. Since $R\leq A_1f_1'(k_1)$. So, we have $\lim_{A_1\to 0}f_1'(k_1)=\infty$. Therefore, we have 
\begin{align}\label{i}
\lim_{A_1\to 0}k_1=0, \text{ and } \lim_{A_1\to 0}k_2=S.
\end{align}By consequence, we get that 
$\gamma_1A_1f_1(k_1)-Rb_1=\gamma_1A_1f_1(k_1)-Rk_1+RS_1>0$ for $A_1$ small enough. It means that the borrowing constraint of agent $1$ is not binding. To sum up, we have $R= A_1f_1'(k_1)\geq \gamma_2F_2'(S)$ for $A_1$ small enough.

Since $R_2$ does not depend on $A_1$, we observe that $\lim_{A_1\to 0}k_1^n(R_2/A_1)=0$. So, by combining with the assumption $k^n_2(\frac{R_2}{A_2})<S$, we have $k^n_1(\frac{R}{A_1})+k^n_2(\frac{R}{A_2})<S$ for $A_1$ small enough. According to point 3 of Lemma \ref{Ridef}, we have $R_1<R_2$ for $A_1$ small enough. Hence, we can apply Corollary \ref{charac2agents} to obtain that the borrowing constraint of agent $2$ is binding in equilibrium. It means that $\gamma_2A_2f_2(k_2)-Rk_2+RS_2=0$. 

Look at the market clearing condition:  $k^n_1(\frac{R}{A_1})+k^b_2(\frac{R}{\gamma_2A_2},S_2)=S\equiv \sum_iS_i$. When $A_1$ converges to $0$, we have $k^n_1(\frac{R}{A_1})$ converges to $0$. So, $R$ converges to $R(0)$ satisfying $k^b_2(\frac{R(0)}{\gamma_2A_2},S_2)=S$. So, we have 
\begin{align}
\lim_{A_1\to \infty}\big(A_1f_1^\prime(k_1)-A_2f_2^\prime(k_2)\big)=\lim_{A_1\to \infty}\big(R-A_2f_2^\prime(k_2)\big)=R(0)-A_2f_2^\prime(S).
\end{align}


Since agent $2$'s borrowing constraint is binding: $R(k_2-S_2)=\gamma_2A_2f_2(k_2)$ for any $A_1$ small enough. Let $A_1$ tend to zero, we have $\gamma_2A_2f_2(S)=R(0)(S-S_2)=R(0)S_1$. So, we have
\begin{align}\label{ii}
R(0)-A_2f_2^\prime(S)=\frac{\gamma_2A_2f_2(S)}{S_1}-A_2f_2^\prime(S)<0
\end{align}
because we assume that $\gamma_2<\frac{S_1}{S_1+S_2}\frac{Sf_2^\prime(S)}{f_2(S)}$.

Again, by the market clearing condition $k^n_1(\frac{R}{A_1})+k^b_2(\frac{R}{\gamma_2A_2},S_2)=S$, we have 
\begin{align}
(k^n_1)^{\prime}\big(\frac{R}{A_1}\big)\frac{R'(A_1)A_1-R}{A_1^2}+\frac{\partial k^b_2}{\partial x_1}\big(\frac{R}{\gamma_2A_2},S_2\big)\frac{R'(A_1)}{(\gamma_2A_2)^2}=0\notag\\
\Leftrightarrow R'(A_1)\Big(\frac{1}{A_1}(k^n_1)^{\prime}\big(\frac{R}{A_1}\big)+\frac{1}{(\gamma_2A_2)^2}\frac{\partial k^b_2}{\partial x_1}\big(\frac{R}{\gamma_2A_2},S_2\big)\Big)=(k^n_1)^{\prime}\big(\frac{R}{A_1}\big)\frac{R}{A_1^2}\notag\\
\Leftrightarrow R'(A_1)A_1\Big(\frac{1}{R}+\frac{A_1}{(\gamma_2A_2)^2R}\frac{\frac{\partial k^b_2}{\partial x_1}\big(\frac{R}{\gamma_2A_2},S_2\big)}{(k^n_1)^{\prime}\big(\frac{R}{A_1}\big)}\Big)=1\label{Rprimea2}
\end{align}
First, since $\frac{\partial k^b_2}{\partial x_1}<0$ and $(k^n_1)^{\prime}\big(\frac{R}{A_1}\big)<0$, we have $R'(A_1)>0$.


Recall that $f_1'(k^n_1(x))=x$. So, we have $(k^n_1)'(x)f_1^{\prime \prime}(k^n_1(x))=1$, and hence, 
$$\lim_{A_1\to 0}\frac{R}{A_1}(k^n_1)'(\frac{R}{A_1})=\lim_{A_1\to 0}\frac{\frac{R}{A_1}}{f_1^{\prime \prime}(\frac{R}{A_1})}=\lim_{x\to +\infty}\frac{x}{f_1^{\prime \prime}(x)}< 0.$$
By combining this with (\ref{Rprimea2}) and $\lim_{A_1\to 0}R=R(0)>0$, we get that
\begin{align}\label{iii}
\lim_{A_1\to 0}R'(A_1)=+\infty.
\end{align} It is easy to see that, when $A_1$ is small enough, the  $\frac{\partial k_2}{\partial R}=\frac{\partial k_2^b}{\partial R}\big(\frac{R}{\gamma_2A_2},S_2\big)$. Thus, 
 \begin{align}\label{iv}
\lim_{A_1\to 0}\frac{\partial k_2}{\partial R}=\lim_{R\to R(0)}\frac{\partial k_2^b}{\partial R}\big(\frac{R}{\gamma_2A_2},S_2\big)<0.
\end{align}
By combining (\ref{i}), (\ref{ii}), (\ref{iii}), (\ref{iv}) and (\ref{ya1}), we conclude that $\frac{\partial Y}{\partial A_1}<0$ for any $A_1>0$ small enough.

We now consider the Cobb-Douglas production functions.  In such a case, condition $k^n_2(\frac{R_2}{A_2})<S$ becomes $\gamma_2< \alpha \frac{S_1}{S_1+S_2}$.  For the sake of simplicity, we write $k^n_2$ instead of $k^n_2(\frac{R_2}{A_2})$. Recall that $R_2=A_2f_2'(k^n_2)=A_2\alpha (k^n_2)^{\alpha-1}$. Hence, \begin{align*}
&\big(k^n_2-S_2\big)R_2=\gamma_2A_2f_2\big(k^n_2\big)\Leftrightarrow (k^n_2-S_2)R_2=\gamma_2A_2(k^n_2)^{\alpha} \\
&\Leftrightarrow (k^n_2-S_2)A_2\alpha (k^n_2)^{\alpha-1}=\gamma_2A_2(k^n_2)^{\alpha}\\
&\Leftrightarrow(\alpha-\gamma_2)A_2(k^n_2)^{\alpha}=S_2A_2\alpha (k^n_2)^{\alpha-1} \Leftrightarrow (\alpha-\gamma_2)k^n_2=\alpha S_2.
\end{align*}
Therefore, condition $S_1+S_2> k^n_2(R_2/A_2)$ becomes $(S_1+S_2)(\alpha-\gamma_2)> \alpha S_2$, or, equivalently,
$\alpha\frac{S_1}{S_1+S_2}> \gamma_2.$
\end{proof}

\begin{proof}[{\bf Proof of Proposition \ref{tfp-rn}}]
\label{linear-ww-proof}
We make use of Theorem \ref{cate-r}. We firstly consider the regime $\mathcal{R}_{n}$ with $n\leq m-1$. In this regime, we have
\begin{align}Y=Y_n=\sum_{i=n+1}^m\frac{rA_iS_i}{r-f_iA_i}\leq \sum_{i=n+1}^mA_m\frac{rS_i}{r-f_iA_i}=A_mS.
\end{align} 
Notice that $Y=A_mS$ if and only if $n+1=m$.

We now consider the regime $\mathcal{A}_{n}$ with $n\leq m$.  In this regime, we have
\begin{align*}Y&=A_n\sum_{i=1}^nS_i+\sum_{i=n+1}^m\frac{A_n(1-f_i)A_iS_i}{A_n-f_iA_i}=A_n\sum_{i=1}^mS_i+A_n\sum_{i=n+1}^m\frac{A_n(A_i-A_n)}{A_n-f_iA_i}\\
&\leq A_nS+(A_m-A_n)\sum_{i=n+1}^m\frac{A_n}{A_n-f_iA_i}\leq A_nS+(A_m-A_n)S=A_mS.
\end{align*} 
where the last inequality is from the condition $\sum_{i=n+1}^m\frac{A_nS_i}{A_n-f_iA_i}$ in the regime $\mathcal{A}_{n}$.

It is easy to see that $Y=A_mS$ if and only if either (i) $n+1>m$ or (ii) $n+1=m$ and $\frac{A_{m-1}}{A_{m-1}-f_mA_m}S_m=S$. Combining these two cases, we obtain point 1 of our result.

\end{proof}

\section{Proofs for Section \ref{effect-fi}}\label{effect-fi-proof}

\begin{proof}[{\bf Proof of Proposition \ref{yn-fi-co}}\label{yn-fi-co-proof}]Under assumptions in Proposition \ref{yn-fi-co}, we can prove that the equilibrium interest rate is in $(A_{n-1},A_{n})$ if and only if 
\begin{align}
\sum_{i=n}^m\dfrac{A_{n}S_i}{A_{n}- \gamma_iA_i}<\sum_iS_i<\sum_{i=n}^m\dfrac{A_{n-1}S_i}{A_{n-1}- \gamma_iA_i}
\end{align} 
Then, when $R\in (A_{n-1},A_{n})$, it is determined by \begin{align}\label{Rn} 
\underbrace{\sum_{i=n}^m\dfrac{\gamma_iA_i}{R- \gamma_iA_i}S_i}_{\text{Asset demand}}=\underbrace{\sum_{i=1}^{n-1}S_i}_{\text{Asset supply}}  \text{ or equivalently }
 \underbrace{\sum_{i=n}^m\dfrac{RS_i}{R- \gamma_iA_i}}_{\text{Capital demand}}=\underbrace{S}_{\text{Capital supply}}
\end{align} 
Agents $1,\ldots,n-1$ are lenders while agents $n,\ldots,m$ are borrowers.
 It is easy to see that  $\frac{\partial Y}{\partial \gamma_i}=0$, $\forall i\leq n-1$.  For $i\geq n$, by using condition $\sum_{i=n}^m\dfrac{RS_i}{R- \gamma_iA_i}=\sum_iS$, we get that 
\begin{align}
\frac{\partial R}{\partial \gamma_j}&=\frac{\frac{RA_jS_j}{(R-\gamma_jA_j)^2}}{\Big(\sum_{i=n}^{m}\frac{\gamma_iA_iS_i}{(R-\gamma_iA_i)^2}\Big)}>0, \text{ and notice that }
\sum_j\frac{\partial R}{\partial \gamma_j}\frac{\gamma_j}{R}=1
\end{align}
Then, we can compute that
\begin{align*}
\frac{\partial Y}{\partial \gamma_i}
&=\sum_{j=n}^mA_jS_j \frac{\partial \Big(\dfrac{R}{R-\gamma_jA_j}\Big) }{\partial \gamma_i}=\sum_{j=n}^mA_jS_j \frac{-\gamma_jA_j }{(R-\gamma_jA_j)^2}\frac{\partial R}{\partial \gamma_i}+\frac{RS_iA_i^2}{(R-\gamma_iA_i)^2}\\
&=\sum_{j=n}^mA_jS_j \frac{-\gamma_jA_j }{(R-\gamma_jA_j)^2}\frac{\frac{RA_iS_i}{(R-\gamma_iA_i)^2}}{\Big(\sum_{j=n}^{m}\frac{\gamma_jA_jS_j}{(R-\gamma_jA_j)^2}\Big)}+\frac{RS_iA_i^2}{(R-\gamma_iA_i)^2}\\
&=\frac{\partial R}{\partial \gamma_i}\Big( A_i\sum_{j=n}^{m}\frac{\gamma_jA_jS_j}{(R-\gamma_jA_j)^2} -\sum_{j=n}^m \frac{\gamma_jS_jA_j^2 }{(R-\gamma_jA_j)^2}\Big).
\end{align*}

The first point is a direct consequence of this expression and the fact that $A_m>\cdots>A_{n+1}$. Let us prove the second point. We have, by noticing that $R\in (A_{n-1},A_{n})$ and $A_{t+1}>A_t$, $\forall t$,
\begin{align}
&A_i\sum_{t=n}^{m}\frac{\gamma_{t}A_tS_t}{(R-\gamma_{t}A_t)^2} -\sum_{t=n}^m \frac{\gamma_{t}S_tA_t^2 }{(R-\gamma_{t}A_t)^2}\\
&=\sum_{t=n}^{i-1}\frac{\gamma_{t}A_tS_t}{(R-\gamma_{t}A_t)^2} (A_i-A_t)-\sum_{t=i+1}^{m}\frac{\gamma_{t}A_tS_t}{(R-\gamma_{t}A_t)^2} (A_t-A_i)\\
&\geq \sum_{t=n}^{i-1}\frac{\gamma_{t}A_tS_t}{(A_{n}-\gamma_{t}A_t)^2} (A_i-A_{i-1})-\sum_{t=i+1}^{m}\frac{\gamma_{t}A_tS_t}{(A_{n-1}-\gamma_{t}A_t)^2} (A_m-A_i).
\end{align}
Combining this with the expression of $\frac{\partial Y}{\partial \gamma_i}$, we obtain (\ref{ynfi1}). 

We also have 
\begin{align}
&A_i\sum_{t=n}^{m}\frac{\gamma_{t}A_tS_t}{(R-\gamma_{t}A_t)^2} -\sum_{t=n}^m \frac{\gamma_{t}S_tA_t^2 }{(R-\gamma_{t}A_t)^2}\\
&=\sum_{t=n}^{i-1}\frac{\gamma_{t}A_tS_t}{(R-\gamma_{t}A_t)^2} (A_i-A_t)-\sum_{t=i+1}^{m}\frac{\gamma_{t}A_tS_t}{(R-\gamma_{t}A_t)^2} (A_t-A_i)\\
&< \sum_{t=n}^{i-1}\frac{\gamma_{t}A_tS_t}{(A_{n-1}-\gamma_{t}A_t)^2} (A_i-A_{n})-\sum_{t=i+1}^{m}\frac{\gamma_{t}A_tS_t}{(A_{n}-\gamma_{t}A_t)^2} (A_{i+1}-A_i).
\end{align}
Combining this with the expression of $\frac{\partial Y}{\partial \gamma_i}$, we obtain (\ref{ynfi2}). 
\end{proof}

\begin{proof}[{\bf Proof of Example \ref{ex-f2-2}}]
\label{3agents}
We focus here on the case $max(\gamma_2A_2,\gamma_3A_3)<A_1$ (in this case the interest rate $R$ may take any value in $[A_1,A_m]$). Applying Theorem \ref{cate-r}, we can check that the interest rate is uniquely determined by
\begin{align}
R=\begin{cases}A_1 \text{ if }  S_1 \geq \frac{\gamma_3 A_3}{A_1 - \gamma_3 A_3}S_3 + \frac{\gamma_2A_2}{A_1- \gamma_2A_2}S_2  \\
R_1  \text{ if } \frac{\gamma_3 A_3}{A_2 - \gamma_3 A_3}S_3 + \frac{\gamma_2}{1- \gamma_2}S_2 < S_1 < \frac{\gamma_3 A_3}{A_1 - \gamma_3 A_3}S_3 + \frac{\gamma_2A_2}{A_1- \gamma_2A_2}S_2  \\
A_2  \text{ if } \frac{\gamma_3 A_3}{A_2 - \gamma_3 A_3}S_3-S_2\leq S_1 \leq \frac{\gamma_3 A_3}{A_2 - \gamma_3 A_3}S_3 + \frac{\gamma_2}{1- \gamma_2}S_2 \\
R_2  \text{ if }  \frac{\gamma_3}{1-\gamma_3}S_3 - S_2 < S_1 < \frac{\gamma_3 A_3}{A_2 - \gamma_3 A_3} - S_2\\
A_3  \text{ if } S_1 \leq \frac{\gamma_3}{1-\gamma_3}S_3 - S_2
\end{cases}
\end{align}
where $R_2 = \gamma_3 A_3 \Big( 1 + \frac{S_3}{S_1 + S_2}\Big)$ and $R_1$ is the highest solution of the equation: 
\begin{align}
\label{r3} \dfrac{\gamma_2 A_2}{R- \gamma_2 A_2} S_2 + \dfrac{\gamma_3A_3}{R-\gamma_3A_3}S_3 = S_1.
\end{align}
This equation implies that  
$R(S_2(R-\gamma_3A_3)+S_3(R-\gamma_2A_2))=S(R-\gamma_2A_2)(R-\gamma_3A_3)$, or equivalently
\begin{subequations}
\begin{align}
&S_1R^2-R\Big((S_1+S_2)\gamma_2A_2+(S_1+S_3)\gamma_3A_3\Big)+S\gamma_2A_2\gamma_3A_3=0.
\end{align}
So, the rate $R_1$ is computed by
\begin{align}
R&=\frac{(S_1+S_2)\gamma_2A_2+(S_1+S_3)\gamma_3A_3+\sqrt{\Delta}}{2S_1}\\
\text{where }\Delta&\equiv \big((S_1+S_2)\gamma_2A_2+(S_1+S_3)\gamma_3A_3\big)^2-4S_1S\gamma_2A_2\gamma_3A_3
\end{align}
\end{subequations}
There are 5 different cases. In each case, we can explicitly compute equilibrium outcomes thanks to Lemma \ref{lem1}. 

\end{proof}

\begin{proof}[{\bf Proof of Corollary  \ref{hete-f} (homogeneous credit limit)}]
\label{hete-f-concave-proof}

Since $F_i'(k_t)\geq R$, there are two cases. (1) If $F_i'(k_i)=R$, then we have hence $\frac{\partial k_i}{\partial \gamma}<0$.  (2) If $F_i'(k_i)>R$, then  borrowing constraint of this agent is binding. 

The market clearing condition $\sum_ik_i=\sum_iS_i$ implies that
$$\sum_{i: F_i'(k_i)=R}\frac{\partial k_i}{\partial \gamma}+\sum_{i: F_i'(k_i)>R}\frac{\partial k_i}{\partial \gamma}=0.$$ 

So, we have $\sum_{i: F_i'(k_i)>R}\frac{\partial k_i}{\partial \gamma}>0$.

We now claim that $\frac{\partial k_i}{\partial \gamma}>0$  for any agent with $F_i'(k_i)>R$. For such agents we have $\gamma F_i(k_i^n)-R(k_i^n-S_i)$. Taking the derivative with respect to $\gamma$ of both sides of this equation, we have
\begin{align}
F_i(k_i)+\gamma F_i'(k_i)\frac{\partial k_i}{\partial \gamma}=\frac{\partial R}{\partial \gamma}(k_i-S_i)+R\frac{\partial k_i}{\partial \gamma}\\
\label{trick}i.e., \quad \frac{\partial k_i}{\partial \gamma}=\Big(\frac{\partial R}{\partial \gamma}\frac{\gamma }{R}-1\Big)\frac{F_i(k_i)}{R-\gamma F_i'(k_i)}.
\end{align}
By summing with respect to $i$ such that $F_i'(k_i)>R$ and noticing that  $\sum_{i: F_i'(k_i)>R}\frac{\partial k_i}{\partial \gamma}>0$ and $R-\gamma F_i'(k_i)>0$ $\forall i$, we get that $\frac{\partial R}{\partial \gamma}\frac{\gamma }{R}-1\geq 0$. From this and (\ref{trick}), we obtain $\frac{\partial k_i}{\partial \gamma} >0$ $\forall i$ such that $F_i'(k_i)>R$.

We now observe that 
\begin{align}\frac{\partial Y}{\partial \gamma}
&=\sum_{i: F_i'(k_i)=R}
F_j'(k_i)\frac{\partial k_i}{\partial \gamma}+\sum_{i: F_i'(k_i)>R}F_i'(k_i)\frac{\partial k_i}{\partial \gamma}\geq R\Big(\sum_{i=1}^m \frac{\partial k_i}{\partial \gamma}\Big)=0.
\end{align}
\end{proof}
\section{Proofs of Section \ref{extensions}}
\label{extensions-proof}

 \setcounter{equation}{0} 
\renewcommand{\theequation}{B.\arabic{equation}}

Firstly, we provide a sufficient condition to check whether a sequence of prices and allocations is an intertemporal equilibrium. 

\begin{lemma}
\label{ns}

If the sequences $(R_t, (c_{i,t},k_{i,t},b_{i,t})_i)_{t}$    and $(\lambda_{i,t},\mu_{i,t},\eta_{i,t})_{i,t}$ satisfy  the following conditions:
\begin{enumerate}
\item $c_{i,t}, l_{i,t}, \lambda_{i,t}, \eta_{i,t}, \mu_{i,t+1}$ are non-negative and $R_t>0$ for any $t$.
\item  $c_{i,t}+k_{i,t}+R_tb_{i,t-1}= F_{i,t}(k_{i,t-1}) + b_{i,t},$ and $ R_{t+1}b_{i,t}- \gamma_iF_{i,t}(k_{i,t})= 0$, $\forall i, \forall t$.
\item $\sum_ib_{i,t}=0$, $\forall t$.
\item $\summ_{t=0}^{\infty}\lambda_{i,t}c_{i,t}<\infty, \summ_{t=0}^{\infty}\beta_i^tu_i(c_{i,t})<\infty$. 
\item TVCs: $\lim_{T\to \infty}\beta_i^tu_i'(c_{i,t})(k_{i,t}-b_{i,t})=0$, $\forall i.$
\item FOCs:  $\forall i, \forall t,$
\begin{align*}
\beta_i^tu_i'(c_{i,t})&=\lambda_{i,t}\\
\lambda_{i,t}&=\lambda_{i,t+1} F_{i,t+1}'(k_{i,t})+\mu_{i,t+1}\gamma_iF_{i,t+1}'(k_{i,t})+\eta_{i,t}, \quad \eta_{i,t}k_{i,t}=0\\
\lambda_{i,t}&=R_{t+1}\lambda_{i,t+1}+\mu_{i,t+1}R_{t+1}, \quad \mu_{i,t+1}\big(R_{t+1}b_{i,t}- \gamma_iF_{i,t}(k_{i,t})\big)=0,
\end{align*}
%
\end{enumerate}
then the list $(R_t, (c_{i,t},k_{i,t},b_{i,t})_i)$  is an intertemporal equilibrium.

\end{lemma}

\begin{proof}[{\bf Proof of  Lemma \ref{ns}}] 

Before presenting our proof, we should notice that this result requires neither $u_i(0)=0$ nor $u_i'(0)=\infty$.  Let us now prove our result. It is sufficient to prove the optimality of $(c_i,k_i,b_i)$ for all $i$.  Let $(c_i',k_i',b_i')$ be a plan satisfying all budget and borrowing constraints and $b_{i,-1}'-b_{i,-1}=0=k'_{i,-1}-k_{i,-1}$. We have
$
\summ_{t=0}^T \beta_i^t( u_i(c_{i,t})-u_i(c'_{i,t}))\geq  \summ_{t=0}^T \beta_i^t u_i'(c_{i,t}) (c_{i,t}-c'_{i,t})=\summ_{t=0}^T \lambda_{i,t} (c_{i,t}-c'_{i,t}).$

Budget constraints imply that $c_{i,t}=F_{i,t}(k_{i,t-1}) + b_{i,t} - k_{i,t} - R_tb_{i,t-1}$ and $ c'_{i,t}\leq F_{i,t}(k'_{i,t-1}) + b'_{i,t} - k'_{i,t} - R_tb'_{i,t-1}$, and hence,
\begin{align*}
\lambda_{i,t} (c_{i,t}-c'_{i,t})\geq &\lambda_{i,t} (F_{i,t}(k_{i,t-1}) + b_{i,t} - k_{i,t} - R_tb_{i,t-1}-F_{i,t}(k'_{i,t-1}) - b'_{i,t} + k'_{i,t} + R_tb'_{i,t-1})\\
=& \lambda_{i,t} \big(F_{i,t}(k_{i,t-1})-F_{i,t}(k'_{i,t-1})\big)-\lambda_{i,t}(k_{i,t}-k'_{i,t})+\lambda_{i,t}(b_{i,t}-b'_{i,t})-\lambda_{i,t}R_t(b_{i,t-1}-b'_{i,t-1}).
\end{align*}
According to FOCs, we have 
\begin{align*}
\lambda_{i,t}k'_{i,t}&=\lambda_{i,t+1} F_{i,t+1}'(k_{i,t})k'_{i,t}+\gamma_i\mu_{i,t+1}F_{i,t+1}'(k_{i,t})k'_{i,t}+\eta_{i,t}k'_{i,t}\\
\lambda_{i,t}b'_{i,t}&=R_{t+1}\lambda_{i,t+1}b'_{i,t}+R_{t+1}\mu_{i,t+1}b'_{i,t}
\end{align*}
This implies that 
\begin{align}
\label{add1}
\lambda_{i,t}(k_{i,t}-k'_{i,t})&=\lambda_{i,t+1} F_{i,t+1}'(k_{i,t})(k_{i,t}-k'_{i,t})+\gamma_i\mu_{i,t+1}F_{i,t+1}'(k_{i,t})(k_{i,t}-k'_{i,t})+\eta_{i,t}(k_{i,t}-k'_{i,t})\\
\label{add2}\lambda_{i,t}(b_{i,t}-b'_{i,t})&=R_{t+1}\lambda_{i,t+1}(b_{i,t}-b'_{i,t})+R_{t+1}\mu_{i,t+1}(b_{i,t}-b'_{i,t})
\end{align}
Therefore, we have that 
\begin{align*}
&\sum_{t=0}^T\lambda_{i,t} (c_{i,t}-c'_{i,t})\geq \sum_{t=0}^T\Big(\lambda_{i,t} \big(F_{i,t}(k_{i,t-1})-F_{i,t}(k'_{i,t-1})\big)-\lambda_{i,t}(k_{i,t}-k'_{i,t})\Big)\\
&\quad\quad\quad\quad\quad\quad +\sum_{t=0}^T\Big(\lambda_{i,t}(b_{i,t}-b'_{i,t})-\lambda_{i,t}R_t(b_{i,t-1}-b'_{i,t-1})\Big)\\
\geq&\sum_{t=0}^{T-1}\big(\lambda_{i,t+1} F'_{i,t+1}(k_{i,t})-\lambda_{i,t}\big)(k_{i,t}-k'_{i,t})-\lambda_{i,T}(k_{i,T}-k'_{i,T})\\
&+\sum_{t=0}^{T-1}\big(\lambda_{i,t}-\lambda_{i,t+1}R_{t+1}\big)(b_{i,t}-b'_{i,t})+\lambda_{i,T}(b_{i,T}-b'_{i,T})\\
=&\lambda_{i,T}\big(k'_{i,T}-b'_{i,T}-(k_{i,T}-b_{i,T})\big)+\sum_{t=0}^{T-1}\eta_{i,t}(k'_{i,t}-k_{i,t})\\
&
+\sum_{t=0}^{T-1}\mu_{i,t+1}\Big(-\gamma_iF_{i,t+1}'(k_{i,t})(k_{i,t}-k'_{i,t})+R_{t+1}(b_{i,t}-b'_{i,t})\Big)\\
\end{align*}
We consider $\mu_{i,t+1}\Big(-\gamma_iF_{i,t+1}'(k_{i,t})(k_{i,t}-k'_{i,t})+R_{t+1}(b_{i,t}-b'_{i,t})\Big)$.

\begin{align}
&\mu_{i,t+1}\Big(-\gamma_iF_{i,t+1}'(k_{i,t})(k_{i,t}-k'_{i,t})+R_{t+1}(b_{i,t}-b'_{i,t})\Big) \notag\\
=&\mu_{i,t+1}(R_{t+1}b_{i,t}- \gamma_iF_{i,t}(k_{i,t})-(R_{t+1}b'_{i,t}- \gamma_iF_{i,t}(k'_{i,t}))) \notag\\
&+\mu_{i,t+1}\Big(-(R_{t+1}b_{i,t}- \gamma_iF_{i,t}(k_{i,t}))+(R_{t+1}b'_{i,t}- \gamma_iF_{i,t}(k'_{i,t}))\\
&\quad \quad \quad \quad-\gamma_iF_{i,t+1}'(k_{i,t})(k_{i,t}-k'_{i,t})+R_{t+1}(b_{i,t}-b'_{i,t})\Big) \notag\\
\geq &\mu_{i,t+1}(R_{t+1}b_{i,t}- \gamma_iF_{i,t}(k_{i,t})-(R_{t+1}b'_{i,t}- \gamma_iF_{i,t}(k'_{i,t})))\\
&=\mu_{i,t+1}(\gamma_iF_{i,t}(k'_{i,t})-R_{t+1}b'_{i,t})\geq 0.\label{add31}
\end{align}
It remains to prove that $\liminf_{T\rightarrow \infty}\lambda_{i,T}\big(k'_{i,T}-b'_{i,T}-(k_{i,T}-b_{i,T})\big)\geq 0$.

According to (\ref{add1}) and (\ref{add2}), we have 
\begin{align*}
&\lambda_{i,t}\big(k'_{i,t}-b'_{i,t}-(k_{i,t}-b_{i,t})\big)\\
=&R_{t+1}\lambda_{i,t+1}(b_{i,t}-b'_{i,t})+\mu_{i,t+1}R_{t+1}(b_{i,t}-b'_{i,t})\\
&-\Big(\lambda_{i,t+1} F_{i,t+1}'(k_{i,t})(k_{i,t}-k'_{i,t})+\gamma_i\mu_{i,t+1}F_{i,t+1}'(k_{i,t})(k_{i,t}-k'_{i,t})+\eta_{i,t}(k_{i,t}-k'_{i,t})\Big)\\
=&R_{t+1}\lambda_{i,t+1}(b_{i,t}-b'_{i,t})-\lambda_{i,t+1} F_{i,t+1}'(k_{i,t})(k_{i,t}-k'_{i,t}) +\eta_{i,t}(k_{i,t}-k'_{i,t})\\
&+\mu_{i,t+1}R_{t+1}(b_{i,t}-b'_{i,t})+
\mu_{i,t+1}\gamma_iF_{i,t+1}'(k_{i,t})(k_{i,t}-k'_{i,t})\\
\geq &R_{t+1}\lambda_{i,t+1}(b_{i,t}-b'_{i,t})-\lambda_{i,t+1} F_{i,t+1}'(k_{i,t})(k_{i,t}-k'_{i,t}).
\end{align*}
where we use (\ref{add31}) the fact that $\eta_{i,t}(k_{i,t}-k'_{i,t})=-\eta_{i,t}k'_{i,t}\leq 0$ for the last inequality. 

Since $F_{i,t+1}$ is concave, we have $F_{i,t+1}'(k_{i,t})(k_{i,t}-k'_{i,t})\leq F_{i,t+1}(k_{i,t})-F_{i,t+1}(k'_{i,t})$. So, we get that 
\begin{align*}
\lambda_{i,t}\big(k'_{i,t}-b'_{i,t}-(k_{i,t}-b_{i,t})\big)\geq &R_{t+1}\lambda_{i,t+1}(b_{i,t}-b'_{i,t})-\lambda_{i,t+1}\big(F_{i,t+1}(k_{i,t})-F_{i,t+1}(k'_{i,t})\big)\\
=&\lambda_{i,t+1}\big(R_{t+1}b_{i,t}-F_{i,t+1}(k_{i,t})\big)+\lambda_{i,t+1}\big(F_{i,t+1}(k'_{i,t})-R_{t+1}b'_{i,t}\big)
\end{align*}
We have $F_{i,t+1}(k'_{i,t})-R_{t+1}b'_{i,t}\geq 0$ because $\gamma_iF_{i,t+1}(k'_{i,t})-R_{t+1}b'_{i,t}\geq 0$.

The budget constraint at date $t$ implies that $\lambda_{i,t}(c_{i,t}+k_{i,t}-b_{i,t})=\lambda_{i,t}\big(F_{i,t}(k_{i,t-1}) -R_tb_{i,t-1} \big)$. Since $\lim_{t\to \infty}\lambda_{i,t}c_{i,t}=0=\lim_{t\to \infty}\lambda_{i,t}(k_{i,t}-b_{i,t})$, we get that  $\lim_{t\to \infty}\lambda_{i,t}\big(F_{i,t}(k_{i,t-1}) -R_tb_{i,t-1} \big)=0$. By consequence, we obtain that $\liminf_{T\rightarrow \infty}\lambda_{i,T}\big(k'_{i,T}-b'_{i,T}-(k_{i,T}-b_{i,T})\big)\geq 0$.
\end{proof}
\subsection{Proofs for Section \ref{sec41}}

\begin{proof}[{\bf Proof of Proposition \ref{ss}. Steady state analysis.}] 
Let us focus on an interior equilibrium (i.e., $k_{i,t}>0$, $\forall i,t$), we can write the FOCs
\begin{align*}\beta_i^tu_i'(c_{i,t})&=\lambda_{i,t}\\
\lambda_{i,t}&=F_{i,t}'(k_{i,t})(\lambda_{i,t+1}+\gamma_i\mu_{i,t+1})\\
\lambda_{i,t}&=R_{t+1}(\lambda_{i,t+1}+\mu_{i,t+1})\\
\mu_{i,t+1}(R_{t+1}b_{i,t}-\gamma_iF_i(k_{i,t}))&=0
\end{align*}
where $\mu_{i,t}\geq 0$ is the multiplier with respect to the constraint  $R_tb_{i,t-1}- \gamma_i F_{i,t}(k_{i,t-1})\leq 0$.

According to FOCs, we have that $1\geq R_{t+1}\max_i\frac{\beta_i u_i'(c_{i,t+1})}{u_i'(c_{i,t})}$, $\forall i$. Since $k_{i,t}>0$, $\forall i, \forall t$, there exists an agent, say agent $i$, whose borrowing constraint at date $t+1$ is not binding. It means that $\mu_{i,t+1}=0$. By consequence, we have $1= R_{t+1}\frac{\beta_i u_i'(c_{i,t+1})}{u_i'(c_{i,t})}=R_{t+1}\max_j\frac{\beta_j u_j'(c_{j,t+1})}{u_j'(c_{j,t})}$.  Therefore, we have $R=1/\max_i\{\beta_i\}$ at steady state.

The first-order conditions imply that $\lambda_{i,t}\frac{R_{t+1}-\gamma_iF_{i,t}'(k_{i,t})}{R_{t+1}}=F_{i,t}'(k_{i,t})\lambda_{i,t+1}(1-\gamma_i)$.  By consequence, we obtain point 2.

\end{proof}
 \begin{proof}[{\bf Proof of Lemma \ref{infinite-a11}}]
The maximization problem of agent $i$ is 
\begin{align*}
 &\ma_{(c_i,k_i,b_i)} \sum_{t=0}^{\infty}\beta_i^tu_i(c_{i,t})\\
\text{subject to: }& c_{i,t}+k_{i,t}+R_tb_{i,t-1} \leq A_{i,t}k_{i,t-1} + b_{i,t}\\
&  R_tb_{i,t-1}\leq \gamma_i A_{i,t}(k_{i,t-1})
\end{align*}
 Denote $s_{i,t}=k_{i,t}-b_{i,t}$ the net saving of agent $i$ at date $t$.

Let $R_t=A_{h,t}$, $\forall t$, for some agent $h$.

For agent $h$, we have $c_{h,t}+(k_{h,t}-b_{h,t}) \leq A_{h,t}(k_{h,t-1}-b_{h,t-1})$.   We can compute that
\begin{align*}
 s_{h,0}&=\beta_h w_{h,0}, \quad    s_{h,t}=\beta_h A_{h,t}s_{h,t-1} \text{ } \forall t\geq 1\\
   s_{h,t}&=\beta_h^t A_{h,t}\cdots A_{h,1}s_{i,0}
\end{align*}
For agent $i<h$, since $A_{i,t}<R_t=A_{h,t}$, $\forall t$, we have $k_{i,t}=0$ and hence we find that  
\begin{align*}
 s_{i,0}&=\beta_i w_{i,0}, \quad    s_{i,t}=\beta_i R_ts_{i,t-1} \text{ } \forall t\geq 1\\
   s_{i,t}&=\beta_i^t R_{t}\cdots R_{1}s_{i,0}.
\end{align*}
For agent $j>h$, since $A_{j,t}>R_t=A_{h,t}$, $\forall t$, her borrowing constraint is always binding: $R_tb_{j,t-1}= \gamma_j A_{j,t}k_{j,t-1}$. Therefore, we have $$s_{j,t}=k_{j,t}\Big(1-\frac{\gamma_jA_{j,t+1}}{R_{t+1}}\Big), \quad A_{j,t}k_{j,t-1} - R_tb_{j,t-1}=(1-\gamma_j)A_{j,t}k_{j,t-1}, \forall t\geq 1.$$ From this, we can compute that
\begin{align*}
 s_{j,t}&=\beta_j\frac{(1-\gamma_j)A_{j,t}R_t}{R_t-\gamma_jA_{j,t}}s_{j,t-1}=\Big(\beta_j\frac{(1-\gamma_j)A_{j,t}R_t}{R_t-\gamma_jA_{j,t}}\Big)\cdots \Big(\beta_j\frac{(1-\gamma_j)A_{j,1}R_1}{R_1-\gamma_jA_{j,1}}\Big) s_{j,0}\\
 k_{j,t}&=\frac{1}{1-\frac{\gamma_jA_{j,t+1}}{R_{t+1}}}s_{j,t}=\frac{R_{t+1}}{R_{t+1}-\gamma_jA_{j,t+1}}s_{j,t}\\
  b_{j,t}&= \frac{\gamma_jA_{j,t+1}}{R_{t+1}}k_{j,t}=\frac{\gamma_jA_{j,t+1}}{R_{t+1}-\gamma_jA_{j,t+1}}s_{j,t}
\end{align*}

Therefore, we can find the capital of the agent $h$
\begin{align*}
k_{h,t}=&s_{h,t}+b_{h,t}=s_{h,t}-\sum_{i<h}b_{i,t}-\sum_{j>h}b_{j,t}\\
=&\beta_h^t A_{h,t}\cdots A_{h,1}s_{h,0}+\sum_{i<h}\beta_i^t A_{h,t}\cdots A_{h,1}s_{i,0}\\
&-\sum_{j>h}\frac{\gamma_jA_{j,t+1}}{R_{t+1}-\gamma_jA_{j,t+1}}\Big(\beta_j\frac{(1-\gamma_j)A_{j,t}R_t}{R_t-\gamma_jA_{j,t}}\Big)\cdots \Big(\beta_j\frac{(1-\gamma_j)A_{j,1}R_1}{R_1-\gamma_jA_{j,1}}\Big) s_{j,0}
\end{align*}

In order to keep $k_{h,t}>0$, $\forall t$, we impose that
\begin{align*}
&\sum_{i\leq h}\beta_i^t
s_{i,0}-\sum_{j>h}\beta_j^t\frac{\gamma_jA_{j,t+1}}{A_{h,t+1}-\gamma_jA_{j,t+1}}\frac{(1-\gamma_j)A_{j,t}}{A_{h,t}-\gamma_jA_{j,t}} \cdots \frac{(1-\gamma_j)A_{j,1}}{A_{h,1}-\gamma_jA_{j,1}} s_{j,0}>0.
\end{align*}
which is actually condition (\ref{mh}).

The borrowing constraint of agent $h$ at date $t$ becomes $k_{h,t}\leq \frac{R_{t+1}s_{h,t}}{R_{t+1}-\gamma_hA_{h,t}}=\frac{s_{h,t}}{1-\gamma_h}$. This is equivalent to
\begin{align*}
&\sum_{i\leq h}\beta_i^t A_{h,t}\cdots A_{h,1}s_{i,0}-\sum_{j>h}\frac{\gamma_jA_{j,t+1}}{R_{t+1}-\gamma_jA_{j,t+1}}\Big(\beta_j\frac{(1-\gamma_j)A_{j,t}R_t}{R_t-\gamma_jA_{j,t}}\Big)\cdots \Big(\beta_j\frac{(1-\gamma_j)A_{j,1}R_1}{R_1-\gamma_jA_{j,1}}\Big) s_{j,0}\\
&\leq \beta_h^t A_{h,t}\cdots A_{h,1}s_{i,0}\frac{1}{1-\gamma_h}\\
&\Leftrightarrow \sum_{i\leq h}\beta_i^t s_{i,0}-\sum_{j>h}\frac{\gamma_jA_{j,t+1}}{R_{t+1}-\gamma_jA_{j,t+1}}\Big(\beta_j\frac{(1-\gamma_j)A_{j,t}}{R_t-\gamma_jA_{j,t}}\Big)\cdots \Big(\beta_j\frac{(1-\gamma_j)A_{j,1}}{R_1-\gamma_jA_{j,1}}\Big) s_{j,0}\leq \beta_h^t s_{h,0}\frac{1}{1-\gamma_h}.
\end{align*}

Under these conditions, by applying Lemma \ref{ns}, we can check that the above list $(R_t,(c_{i,t},k_{i,t},b_{i,t})_i)$ is an equilibrium.
 
  We now compute the aggregate production
\begin{align*}
Y_t&=A_{h,t}k_{h,t-1}+\sum_{j>h}A_{h,t}k_{j,t-1}\\
&=A_{h,t}\cdots A_{h,1}\Big(\beta_h^{t-1}s_{h,0}+\sum_{i<h}\beta_i^{t-1}
s_{i,0}-\sum_{j>h}\beta_j^{t-1}\frac{\gamma_jA_{j,t}}{A_{h,t}-\gamma_jA_{j,t}}\frac{(1-\gamma_j)A_{j,t-1}}{A_{h,t-1}-\gamma_jA_{j,t-1}} \cdots \frac{(1-\gamma_j)A_{j,1}}{A_{h,1}-\gamma_jA_{j,1}} s_{j,0}\Big)\\
&+\sum_{j>h}A_{j,t}\frac{R_{t}}{R_{t}-\gamma_jA_{j,t}}\Big(\beta_j\frac{(1-\gamma_j)A_{j,t-1}R_{t-1}}{R_{t-1}-\gamma_jA_{j,t-1}}\Big)\cdots \Big(\beta_j\frac{(1-\gamma_j)A_{j,1}R_1}{R_1-\gamma_jA_{j,1}}\Big) s_{j,0}\\
&=A_{h,t}\cdots A_{h,1}\sum_{i\leq h}\beta_i^{t-1}
s_{i,0}\\
&+A_{h,t}\cdots A_{h,1}\sum_{j>h}\beta_j^{t-1}(1-\gamma_j)^{t}\frac{A_{j,t}}{A_{h,t}-\gamma_jA_{j,t}}\frac{A_{j,t-1}}{A_{h,t-1}-\gamma_jA_{j,t-1}}\cdots \frac{A_{j,1}}{A_{h,1}-\gamma_jA_{j,1}} s_{j,0}.
 \end{align*}
 
\begin{enumerate}
\item When there are 2 agents and $h=2$, i.e., only the most productive agent produces, this condition is obviously satisfied.
\item 
When there are 2 agents and $h=1$. This condition becomes
\begin{align*}
\beta_1^t s_{1,0}-\beta_2^t\frac{\gamma_2A_{2,t+1}}{A_{1,t+1}-\gamma_2A_{2,t+1}}\frac{(1-\gamma_2)A_{2,t}}{A_{1,t}-\gamma_2A_{2,t}}\cdots \frac{(1-\gamma_2)A_{2,1}}{A_{1,1}-\gamma_2A_{2,1}} s_{2,0}>0
\end{align*}
or equivalently
\begin{align*}
 \frac{s_{2,0}}{s_{1,0}}\frac{\beta_1}{\beta_2}\frac{\gamma_2}{1-\gamma_2}\Big(\frac{\beta_2}{\beta_1}\frac{(1-\gamma_2)A_{2,1}}{A_{1,1}-\gamma_2A_{2,1}} \cdots \frac{\beta_2}{\beta_1}\frac{(1-\gamma_2)A_{2,t+1}}{A_{1,t+1}-\gamma_2A_{2,t+1}}\Big) <1 \forall t.
\end{align*}
This happens if  $\sup_t\dfrac{\beta_2}{\beta_1}\dfrac{1-\gamma_2}{\frac{A_{1,t}}{A_{2,t}}-\gamma_2}<1$ and $ \frac{s_{2,0}}{s_{1,0}}\frac{\beta_1}{\beta_2}\frac{\gamma_2}{1-\gamma_2}\leq 1$. 

\end{enumerate}
 \end{proof}

\begin{proof}[{\bf Proof of Proposition \ref{infinite-akm} (m agents)}]
 To investigate the properties of the output and the growth rate, we need an useful lemma whose proof is left for the reader.
  \begin{lemma}\label{useful1} Let $N\geq 1$ be an integer. For each integer $t\geq 1$, we denote  $X_t\equiv \sum_{i=1}^N\alpha_ia_i^t$, where $\alpha_i>0,a_i>0$ for any $i$. We have
 \begin{align}
 \lim_{t\to \infty}\frac{X_{t+1}}{X_t}=\max_{1\leq i\leq N}a_i
 \end{align}
 \end{lemma}
 According to Lemma \ref{infinite-a11}, we have that
 \begin{align*}
Y_t&=A_{h,t}\cdots A_{h,1}\sum_{i\leq h}\beta_i^{t-1}
s_{i,0}\\
&+A_{h,t}\cdots A_{h,1}\sum_{j>h}\beta_j^{t-1}(1-\gamma_j)^{t}\frac{A_{j,t}}{A_{h,t}-\gamma_jA_{j,t}}\frac{A_{j,t-1}}{A_{h,t-1}-\gamma_jA_{j,t-1}}\cdots \frac{A_{j,1}}{A_{h,1}-\gamma_jA_{j,1}} s_{j,0}.
 \end{align*}
 
 When $A_{i,t}=A_i$, $\forall t,\forall i,$, we have that
 \begin{align*}
Y_t&=A_h^t\sum_{i\leq h}\beta_i^{t-1}
s_{i,0}+\sum_{j>h}\beta_j^{t-1}(1-\gamma_j)^{t}A_{j}^t\Big(\frac{A_{h}}{A_{h}-\gamma_jA_{j}}\Big)^t s_{j,0}.
 \end{align*}

 \begin{enumerate}
 \item 
By consequence, we can compute that
  \begin{align*}
\frac{1}{tA_h^{t-1}} \frac{\partial Y_t}{\partial A_h}=\sum_{i\leq h}\beta_i^{t-1}
s_{i,0} -\sum_{j>h}(1-\gamma_j)\frac{\gamma_j A_{j}^2}{(A_h-\gamma_jA_j)^2}\Big(\frac{\beta_j(1-\gamma_j)A_j}{A_{h}-\gamma_jA_{j}}\Big)^{t-1} s_{j,0}.
 \end{align*}
 This implies (\ref{ytah}).
 
 Since $\max_{i\leq h}\beta_i>\max_{j>h}\frac{\beta_j(1-\gamma_j)A_{j}}{A_{h}-\gamma_jA_{j}}$, Lemma \ref{useful1} implies that 
 \begin{align*}\lim_{t\to \infty}\frac{\Big(
 \sum_{i\leq h}\beta_i^{t-1}
s_{i,0} -\sum_{j>h}(1-\gamma_j)\frac{\gamma_j A_{j}^2}{(A_h-\gamma_jA_j)^2}\Big(\frac{\beta_j(1-\gamma_j)A_j}{A_{h}-\gamma_jA_{j}}\Big)^{t-1} s_{j,0}\Big)}{\beta_h^t}=\sum_{i\leq h: \beta_i=\beta_h}s_{i,0}>0
 \end{align*}
 By consequence, there exists a date  $t_0$ such that $\frac{\partial Y_t}{\partial A_1}\geq 0 $, $\forall t> t_0$.
 
 \item  Since $\max_{i\leq h}\beta_i>\max_{j>h}\frac{\beta_j(1-\gamma_j)A_{j}}{A_{h}-\gamma_jA_{j}}$, Lemma \ref{useful1} directly implies that $G_{t+1}\equiv \frac{Y_{t+1}}{Y_t} $ converges to $A_h\max_{1\leq i\leq m}\beta_i$. 
 
 We now look at the formula of $G_{t+1}\equiv \frac{Y_{t+1}}{Y_t} $
 \begin{align*}
 G_{t+1}&\equiv \frac{Y_{t+1}}{Y_t} =\frac{A_h^{t+1}\sum_{i\leq h}\beta_i^{t}
s_{i,0}+\sum_{j>h}\beta_j^{t}(1-\gamma_j)^{t+1}A_{j}^{t+1}\Big(\frac{A_{h}}{A_{h}-\gamma_jA_{j}}\Big)^{t+1} s_{j,0}}{A_h^t\sum_{i\leq h}\beta_i^{t-1}
s_{i,0}+\sum_{j>h}\beta_j^{t-1}(1-\gamma_j)^{t}A_{j}^t\Big(\frac{A_{h}}{A_{h}-\gamma_jA_{j}}\Big)^t s_{j,0}}\\
&=A_h\frac{\sum_{i\leq h}\beta_i^{t}
s_{i,0}+\sum_{j>h}\Big(\frac{\beta_j(1-\gamma_j)A_j}{A_{h}-\gamma_jA_{j}}\Big)^{t+1} \frac{s_{j,0}}{\beta_j}}{\sum_{i\leq h}\beta_i^{t-1}
s_{i,0}+\sum_{j>h}\Big(\frac{\beta_jA_{j}(1-\gamma_j)}{A_{h}-\gamma_jA_{j}}\Big)^t \frac{s_{j,0}}{\beta_j}}.
 \end{align*}
 Let us denote $g(x_{h+1},\ldots,x_m)\equiv \frac{\sum_{i\leq h}\beta_i^{t}
s_{i,0}+\sum_{j>h}x_j^{t+1} \frac{s_{j,0}}{\beta_j}}{\sum_{i\leq h}\beta_i^{t-1}
s_{i,0}+\sum_{j>h}x_j^t \frac{s_{j,0}}{\beta_j}},$

  where $x_j\equiv \frac{\beta_jA_{j}(1-\gamma_j)}{A_{h}-\gamma_jA_{j}} $, for $j>h$.

  Denote $B_h\equiv \max_{j>h}x_j$. Recall that we assume that  $M<\beta_h<1$.
  
  For $d\in\{h+1,\ldots,m\}$, we compute that
  \begin{align*}
  \frac{\partial g}{\partial x_d}&=\frac{(t+1)x_d^{t} \frac{s_{d,0}}{\beta_d}\big(\sum_{i\leq h}\beta_i^{t-1}
s_{i,0}+\sum_{j>h}x_j^t \frac{s_{j,0}}{\beta_j}\big)-tx_d^{t-1}\frac{s_{d,0}}{\beta_d}\big(\sum_{i\leq h}\beta_i^{t}
s_{i,0}+\sum_{j>h}x_j^{t+1} \frac{s_{j,0}}{\beta_j}\big)}{\big(\sum_{i\leq h}\beta_i^{t-1}
s_{i,0}+\sum_{j>h}x_j^t \frac{s_{j,0}}{\beta_j}\big)^2}\\
&=A\Big((t+1)x_d\big(\sum_{i\leq h}\beta_i^{t-1}
s_{i,0}+\sum_{j>h}x_j^t \frac{s_{j,0}}{\beta_j}\big)-t\big(\sum_{i\leq h}\beta_i^{t}
s_{i,0}+\sum_{j>h}x_j^{t+1} \frac{s_{j,0}}{\beta_j}\big)\Big).
  \end{align*}
  where $A\equiv \frac{x_d^{t-1} \frac{s_{d,0}}{\beta_d}}{\big(\sum_{i\leq h}\beta_i^{t-1}
s_{i,0}+\sum_{j>h}x_j^t \frac{s_{j,0}}{\beta_j}\big)^2}$. 
 Applying Lemma \ref{useful1}, we have that
 \begin{align*}
 \lim_{t\to \infty}\frac{(t+1)x_d\big(\sum_{i\leq h}\beta_i^{t-1}
s_{i,0}+\sum_{j>h}x_j^t \frac{s_{j,0}}{\beta_j}\big)}{t\big(\sum_{i\leq h}\beta_i^{t}
s_{i,0}+\sum_{j>h}x_j^{t+1} \frac{s_{j,0}}{\beta_j}\big)\Big)}=\frac{x_d}{\beta_h}<1
 \end{align*}
 which implies that there exists a date $t_1$ such that $ \frac{\partial g}{\partial x_d}<0$ for any $t\geq t_1$. 
 Since  $x_j\equiv \frac{\beta_jA_{j}(1-\gamma_j)}{A_{h}-\gamma_jA_{j}} $ is increasing in $A_j$, we get our result.
\end{enumerate} 
   
 \end{proof}
\subsection{Proofs for Section \ref{sec42}}

\begin{proof}[{\bf Proof of Proposition \ref{credit-infinite-magents}}]
Let us focus on an equilibrium where only the most productive agent produces. The interest rate $R_{1}\in (A_{m-1},A_m)$ and $R_t=A_m$, $\forall t\geq 2$. 

  Denote the individual saving $s_{i,t}\equiv k_{i,t}-b_{i,t}$.  For $i<m$, agent $i$ is lender, $k_{i,t}=0$, $s_{i,t}=-b_{i,t}$, $\forall t$.  We can compute that
\begin{align*}
 s_{i,0}&=\beta_i w_{i,0}, \quad    s_{i,t}=\beta_i R_{t}s_{i,t-1} \text{ } \forall t\geq 1\\
   s_{i,t}&=\beta_i^t R_{t}\cdots R_{1}s_{i,0}.
\end{align*}

For agent $m$, since $A_m>R_1$, her  borrowing constraints at date $0$ is binding: $R_1b_{m,0}= \gamma_m A_{m}k_{m,0}$. Therefore, we have 
\begin{align*}
A_{m}k_{m,0} - R_1b_{m,0}&=(1-\gamma_m)A_{m}k_{m,0}\\
s_{m,0}&=k_{m,0}-b_{m,0}=k_{m,0}\Big(1-\frac{\gamma_mA_{m}}{R_{1}}\Big)\\
 k_{m,0}&=\frac{R_1}{R_1-\gamma_mA_{m}}s_{m,0}, \quad 
  b_{m,0}=\frac{\gamma_mA_{m}}{R_{1}-\gamma_mA_{m}}s_{m,0}\end{align*}

The budget constraints of agent $m$ write
\begin{align*}
c_{m,0}+s_{m,0}&=w_{m,0}\\
c_{m,1}+s_{m,1}&=(1-\gamma_m)A_mk_{m,0}=\frac{(1-\gamma_m)A_mR_1}{R_1- \gamma_mA_m}s_{m,0}\\
c_{m,t}+s_{m,t}&=A_ms_{m,t}, \forall t\geq 2\\
s_{m,t}&=k_{m,t}-b_{m,t}, \forall t\geq 2.
\end{align*}

From this and the FOCs, we can compute the individual saving
\begin{align*}
s_{i,0}&=\beta_iw_{i,0}, \forall i\\
 s_{m,1}&=\beta_m\frac{(1-\gamma_m)A_{m}R_1}{R_1-\gamma_mA_{m}}s_{m,0}\\
 s_{i,t}&=\beta_iA_ms_{i,t-1}, \forall t\geq 2, \forall i=1,\cdots,m.
\end{align*}

We now look at equilibrium. From the market clearing condition $\sum_{i}b_{i,t}=0$, we have that 
\begin{align*}
-\sum_{i\not=m}b_{i,0}&=b_{m,0}\\   
\Leftrightarrow \sum_{i\not=m}s_{i,0}&=\frac{\gamma_mA_{m}}{R_{1}-\gamma_mA_{m}}s_{m,0}\\  \Leftrightarrow R_1&=\gamma_mA_m(1+\frac{s_{m,0}}{\sum_{i\not=m}s_{i,0}})=\gamma_mA_m\frac{S_0}{\sum_{i\not=m}s_{i,0}}.
\end{align*}  
By consequence, we find the saving of all agents: $s_{i,0}=\beta_iw_{i,0}, \forall i$, and 
\begin{align*}
s_{i,1}&=\beta_iR_1s_{i,0}=\beta_i\gamma_mA_m(1+\frac{s_{m,0}}{\sum_{i\not=m}s_{i,0}})s_{i,0}=\beta_i\gamma_mA_mS_0\frac{s_{i,0}}{\sum_{j\not=m}s_{i,0}}, \forall i\not=m\\
 s_{m,1}&=\beta_m\frac{(1-\gamma_m)A_{m}R_1}{R_1-\gamma_mA_{m}}s_{m,0}=\beta_m\frac{(1-\gamma_m)A_{m}\gamma_mA_m(1+\frac{s_{m,0}}{\sum_{i\not=m}s_{i,0}}) }{\gamma_mA_m(1+\frac{s_{m,0}}{\sum_{i\not=m}s_{i,0}}) -\gamma_mA_{m}}s_{m,0}\\
&=\beta_m(1-\gamma_m)A_mS_0\\
 s_{i,t}&=\beta_iA_ms_{i,t-1}=(\beta_iA_m)^{t-1}s_{i,1}, \forall t\geq 2, \forall i=1,2.
\end{align*}
where $S_0\equiv \sum_{i=1}^ms_{i,0}$.

It remains to find the sequence of capital $(k_{i,t})$. We have, $\forall t\geq 1$
\begin{align*}
k_{m,0}&=\sum_{i=1}^ms_{i,0}, \quad k_{m,t}=s_{m,t}+b_{m,t}=s_{m,t}-\sum_{i\not=m}b_{i,t}=\sum_{i=1}^ms_{i,t}, \forall t\geq 1\\
k_{m,1}&=\sum_{i\not=m}\beta_iR_1s_{i,0}+s_{m,1}\\
&=\sum_{i\not=m}\beta_i\gamma_mA_mS_0\frac{s_{i,0}}{\sum_{j\not=m}s_{j,0}}+\beta_m(1-\gamma_m)A_m\sum_{i=1}^ms_{i,0} \\
k_{m,t}&=\sum_is_{i,t}=\sum_i(\beta_iA_m)^{t-1}s_{i,1}, \forall t\geq 1\\
&=\sum_{i\not=m}(\beta_iA_m)^{t-1}\beta_i\gamma_mA_mS_0\frac{s_{i,0}}{\sum_{j\not=m}s_{j,0}}+(\beta_mA_m)^{t-1}\beta_m(1-\gamma_m)A_mS_0,, \forall t\geq 1\\
&=S_0A_m^{t}\Big(\gamma_m\sum_{i\not=m}\beta_i^t\frac{s_{i,0}}{\sum_{j\not=m}s_{j,0}}+\beta_m^t(1-\gamma_m)\Big).
\end{align*}

We now check that the above list $((c_{i,t},k_{i,t},b_{i,t})_i,R_t)_t$ is an equilibrium. We use Lemma \ref{ns}. It is easy to verify the market clearing conditions and the FOCs.

Condition $R_1\in (A_{m-1},A_m)$ is ensured by the assumption that 
$$A_{m-1}<\gamma_mA_m(1+\frac{s_{m,0}}{\sum_{i\not=m}s_{i,0}})<A_m.$$
\begin{itemize}
\item We verify borrowing constraints: $R_{t+1}b_{m,t}\leq \gamma_mA_mk_{m,t}$. This is satisfied for $t=0$. Let us consider $t\geq 1$. Since $R_{t+1}=A_m$, we get that $k_{m,t}-s_{m,t}=b_{m,t}\leq \gamma_mk_{m,t}$, or, equivalently, $(1-\gamma_m)k_{m,t}\leq s_{m,t}$. So, we must prove, for any $t\geq 1$, 
\begin{align*}
&(1-\gamma_m)S_0A_m^{t}\Big(\gamma_m\sum_{i\not=m}\beta_i^t\frac{s_{i,0}}{\sum_{j\not=m}s_{j,0}}+\beta_m^t(1-\gamma_m)\Big)\leq (\beta_mA_m)^{t}(1-\gamma_m)S_0\\
\Leftrightarrow &\Big(\gamma_m\sum_{i\not=m}\beta_i^t\frac{s_{i,0}}{\sum_{j\not=m}s_{j,0}}+\beta_m^t(1-\gamma_m)\Big)\leq \beta_m^{t}\\
\Leftrightarrow &\sum_{i\not=m}\beta_i^t\frac{s_{i,0}}{\sum_{j\not=m}s_{j,0}}\leq \beta_m^{t}
\end{align*}
which is satisfied under our assumption.

\item Tranversality conditions:  $\lim_{T\to \infty}\beta_i^Tu_i'(c_{i,T})(k_{i,T}-b_{i,T})=0$. It is easy to verify these conditions because $\beta_i\in (0,1)$ and $u'(c)=1/c$.

\end{itemize}
\end{proof}

\begin{proof}[{\bf Proof of Lemma \ref{credit-infinite-nh}}]
Let us focus on an equilibrium where only the most productive agent produces. The interest rate $R_{1}\in (A_{n-1},A_n)$ and $R_t=A_h$, $\forall t\geq 2$, where $h\geq n$. 

Denote the individual saving $s_{i,t}\equiv k_{i,t}-b_{i,t}$.  For $i<n$, Agent $i$ is lender, $k_{i,t}=0$, $s_{i,t}=-b_{i,t}$, $\forall t$.  We can compute that
\begin{align*}
 s_{i,0}&=\beta_i w_{i,0}, \quad    s_{i,t}=\beta_i R_{t}s_{i,t-1} \text{ } \forall t\geq 1\\
   s_{i,t}&=\beta_i^t R_{t}\cdots R_{1}s_{i,0}.
\end{align*}

For agent $j\geq n$, since $A_n>R_1$, her  borrowing constraints at date $0$ is binding: $R_1b_{j,0}= \gamma_j A_{j}k_{j,0}$. Therefore, we have 
\begin{align*}
A_{}k_{j,0} - R_1b_{j,0}&=(1-\gamma_j)A_{j}k_{j,0}\\
s_{j,0}&=k_{j,0}-b_{j,0}=k_{j,0}\Big(1-\frac{\gamma_jA_{j}}{R_{1}}\Big)\\
 k_{j,0}&=\frac{R_1}{R_1-\gamma_jA_{j}}s_{j,0}, \quad 
  b_{j,0}=\frac{\gamma_jA_{j}}{R_{1}-\gamma_jA_{j}}s_{j,0}\end{align*}

The budget constraints of agent $j=n$ write
\begin{align*}
c_{j,0}+s_{j,0}&=w_{j,0}, \forall j\\
c_{j,1}+s_{j,1}&=(1-\gamma_j)A_jk_{j,0}=\frac{(1-\gamma_j)A_jR_1}{R_1- \gamma_jA_j}s_{j,0}, \forall j\geq n\\
c_{n,t}+s_{n,t}&=A_hs_{n,t}, \forall t\geq 2\\
s_{n,t}&=k_{n,t}-b_{n,t}, \forall t\geq 2.
\end{align*}

From this and the FOCs, we can compute that
\begin{align*}
s_{i,0}&=\beta_iw_{i,0}, \forall i\\
 s_{j,1}&=\beta_j\frac{(1-\gamma_j)A_{j}R_1}{R_1-\gamma_jA_{j}}s_{j,0}, \forall j\geq n\\
 \quad s_{j,t}&=\beta_jA_hs_{j,t-1}, \forall t\geq 2\\
 &=\beta_j^{t-1}A_h^{t-1}s_{j,1}=\beta_j^{t-1}A_h^{t-1}\beta_j\frac{(1-\gamma_j)A_{j}R_1}{R_1-\gamma_jA_{j}}s_{n,0}\\
 \text{for $j\leq h$, } \quad s_{j,t} &=\beta_j^{t}A_h^{t-1}\frac{(1-\gamma_j)A_jR_1}{R_1-\gamma_jA_{j}}s_{j,0}, \forall t\geq 1.
\end{align*}

We now look at the equilibrium  $R_1$. From the market clearing condition $\sum_{i}b_{i,0}=0$, we have that 
\textcolor{blue}{
\begin{align*}
-\sum_{i< n}b_{i,0}&=\sum_{j\geq n}b_{j,0}
\Leftrightarrow \sum_{i<n}s_{i,0}=\sum_{j\geq n}\frac{\gamma_jA_{j}}{R_{1}-\gamma_jA_{j}}s_{j,0}\Leftrightarrow S_0= \sum_{j\geq n}\frac{R_1}{R_{1}-\gamma_jA_{j}}s_{j,0}.
\end{align*}  }
Since $R_1\in (A_{n-1},A_n)$, this condition requires that 
\begin{align*}
\sum_{j\geq n}\frac{\gamma_jA_{j}}{A_{n}-\gamma_jA_{j}}s_{j,0}<\sum_{i<n}s_{i,0}<\sum_{j\geq n}\frac{\gamma_jA_{j}}{A_{n-1}-\gamma_jA_{j}}s_{j,0}.
\end{align*}
In a particular case where $n=m$, we find that  $R_1=\gamma_mA_m(1+\frac{s_{m,0}}{\sum_{i\not=m}s_{i,0}})=\gamma_mA_m\frac{S_0}{\sum_{i\not=m}s_{i,0}}.$

Now, consider agent $j>h$ and date $t\geq 1$. 
For agent $j>h$, since $A_{j,t}>R_t=A_{h,t}$, $\forall t$, her borrowing constraint is always binding: $R_tb_{j,t-1}= \gamma_j A_{j,t}k_{j,t-1}$. Therefore, we have $$s_{j,t}=k_{j,t}\Big(1-\frac{\gamma_jA_{j,t+1}}{R_{t+1}}\Big), \quad A_{j,t}k_{j,t-1} - R_tb_{j,t-1}=(1-\gamma_j)A_{j,t}k_{j,t-1}, \forall t\geq 1.$$ From this, we can compute that
\begin{align*}
 s_{j,0}&=\beta_j w_{j,0},\\
 s_{j,t}&=\beta_j\frac{(1-\gamma_j)A_{j,t}R_t}{R_t-\gamma_jA_{j,t}}s_{j,t-1}, \forall t\geq 2 \\
 &=\Big(\beta_j\frac{(1-\gamma_j)A_{j,t}R_t}{R_t-\gamma_jA_{j,t}}\Big)\cdots \Big(\beta_j\frac{(1-\gamma_j)A_{j,1}R_1}{R_1-\gamma_jA_{j,1}}\Big) s_{j,0}\\
 k_{j,t}&=\frac{1}{1-\frac{\gamma_jA_{j,t+1}}{R_{t+1}}}s_{j,t}=\frac{R_{t+1}}{R_{t+1}-\gamma_jA_{j,t+1}}s_{j,t}\\
 &=\frac{R_{t+1}}{R_{t+1}-\gamma_jA_{j,t+1}}\Big(\beta_j\frac{(1-\gamma_j)A_{j,t}R_t}{R_t-\gamma_jA_{j,t}}\Big)\cdots \Big(\beta_j\frac{(1-\gamma_j)A_{j,1}R_1}{R_1-\gamma_jA_{j,1}}\Big) s_{j,0}\\
  b_{j,t}&= \frac{\gamma_jA_{j,t+1}}{R_{t+1}}k_{j,t}=\frac{\gamma_jA_{j,t+1}}{R_{t+1}-\gamma_jA_{j,t+1}}s_{j,t}\\
  &=\frac{\gamma_jA_{j,t+1}}{R_{t+1}-\gamma_jA_{j,t+1}}\Big(\beta_j\frac{(1-\gamma_j)A_{j,t}R_t}{R_t-\gamma_jA_{j,t}}\Big)\cdots \Big(\beta_j\frac{(1-\gamma_j)A_{j,1}R_1}{R_1-\gamma_jA_{j,1}}\Big) s_{j,0}.
\end{align*}

From the market clearing condition $\sum_{i}b_{i,t}=0$, we have $\sum_{i}s_{i,t}=\sum_ik_{i,t}$ which implies that
\begin{align}
k_{h,t}=&\sum_{i}s_{i,t}-\sum_{i\not=h}k_{i,t}=
\sum_{i}s_{i,t}-\sum_{i>h}k_{i,t} \text{ (since $k_{i,t}=0,\forall i<h$)}\\
=&\sum_{i<n}s_{i,t}+\sum_{n\leq j\leq h}s_{i,t}-\sum_{j> h}b_{j,t}\\
=&\sum_{i<n}\beta_i^t R_{t}\cdots R_{1}s_{i,0}+\sum_{n\leq j\leq h}\beta_j^{t-1}A_h^{t-1}\beta_j(1-\gamma_j)\frac{A_jR_1}{R_1-\gamma_jA_{j}}s_{j,0}\\
&-\sum_{j>h}\frac{\gamma_jA_{j,t+1}}{R_{t+1}-\gamma_jA_{j,t+1}}\Big(\beta_j\frac{(1-\gamma_j)A_{j,t}R_t}{R_t-\gamma_jA_{j,t}}\Big)\cdots \Big(\beta_j\frac{(1-\gamma_j)A_{j,1}R_1}{R_1-\gamma_jA_{j,1}}\Big) s_{j,0}\\
=&\sum_{i<n}\beta_i^t A_h^{t-1} R_{1}s_{i,0}+\sum_{n\leq j\leq h}\beta_j^{t}A_h^{t-1}(1-\gamma_j)\frac{A_jR_1}{R_1-\gamma_jA_{j}}s_{j,0}\\
&-\sum_{j>h}A_h^{t-1}\Big(\beta_j\frac{(1-\gamma_j)A_{j}}{A_h-\gamma_jA_{j}}\Big)^t \frac{\gamma_jA_jR_1}{R_1-\gamma_jA_j} s_{j,0}, \quad \forall t\geq 1.
\end{align}
We then compute the output as in the statement of our result.

Since $k_{h,t}\geq 0, \forall t$, we must have 
\begin{align*}
\sum_{i<n}\beta_i^t R_{1}s_{i,0}+\sum_{n\leq j\leq h}\beta_j^{t}(1-\gamma_j)\frac{A_jR_1}{R_1-\gamma_jA_{j}}s_{j,0}
&-\sum_{j>h}\Big(\beta_j\frac{(1-\gamma_j)A_{j}}{A_h-\gamma_jA_{j}}\Big)^t \frac{\gamma_jA_jR_1}{R_1-\gamma_jA_j}) s_{j,0}\geq 0, \forall t\geq 1.
\end{align*}


We now check that the above list $((c_{i,t},k_{i,t},b_{i,t})_i,R_t)_t$ is an equilibrium. We use Lemma \ref{ns}. It is easy to verify the market clearing conditions and the FOCs.

\begin{itemize}
\item Condition $R_1\in (A_{n-1},A_n)$ is ensured by (\ref{nh1}). Condition $k_{h,t}\geq 0$ is ensured by (\ref{nh2}).

\item We verify borrowing constraints: $R_{t+1}b_{h,t}\leq \gamma_hA_hk_{h,t}$. This is satisfied for $t=0$. Let us consider $t\geq 1$. Since $R_{t+1}=A_h$, $\forall t\geq 1$, this becomes $k_{h,t}-s_{h,t}=b_{h,t}\leq \gamma_hk_{h,t}$, or, equivalently, $(1-\gamma_h)k_{h,t}\leq s_{h,t}$. So, we must prove, for any $t\geq 1$, 
\begin{align*}
&\sum_{i<n}\beta_i^t A_h^{t-1} R_{1}s_{i,0}+\sum_{n\leq j\leq h}\beta_j^{t}A_h^{t-1}(1-\gamma_j)\frac{A_jR_1}{R_1-\gamma_jA_{j}}s_{j,0}\\
&-\sum_{j>h}A_h^{t-1}\Big(\beta_j\frac{(1-\gamma_j)A_{j}}{A_h-\gamma_jA_{j}}\Big)^t \frac{\gamma_jA_jR_1}{R_1-\gamma_jA_j}) s_{j,0}, \quad \forall t\geq 1\\
<&\frac{1}{1-\gamma_h}\beta_h^{t}A_h^{t-1}(1-\gamma_h)\frac{A_hR_1}{R_1-\gamma_hA_{h}}s_{j,0}=\beta_h^{t}A_h^{t-1}\frac{A_hR_1}{R_1-\gamma_hA_{h}}s_{j,0}
\end{align*}
which is satisfied under our assumption.

\item Tranversality conditions:  $\lim_{T\to \infty}\beta_i^Tu_i'(c_{i,T})(k_{i,T}-b_{i,T})=0$. It is easy to verify these conditions because $\beta_i\in (0,1)$ and $u'(c)=1/c$.

\end{itemize}
\end{proof}

\begin{proof}[{\bf Proof of Proposition \ref{ytanalysis}}]
{\bf Part 1: }
The aggregate output at date $1$ equals $Y_1=\sum_{j\geq n}A_jk_{j,0}$. By using the same technique in Proposition \ref{yn-fi-co}, we can provide conditions under which the aggregate output $Y_1$ is increasing or decreasing in the credit limit of producers.  

{\bf Part 2: }We now look at the output from second date on. For any $t\geq 1$, the aggregate output is computed by
\begin{align*}
\frac{1}{A_h^t} Y_{t+1}=&\sum_{i<n}\beta_i^t  R_{1}s_{i,0}+\sum_{n\leq j\leq h}\beta_j^{t}\frac{(1-\gamma_j)A_jR_1}{R_1-\gamma_jA_{j}}s_{j,0}+
\sum_{j> h}\Big(\beta_j\frac{(1-\gamma_j)A_{j}}{A_h-\gamma_jA_{j}}\Big)^{t}\frac{(1-\gamma_j)A_jR_1}{R_1-\gamma_jA_{j}}s_{j,0}
\end{align*}
Note that $R_1$ does not depend on $\gamma_i$ with $i<n$. So, the output does not depend on any agent $i<n$, who are not producer in equilibrium.


From (\ref{ahR1}), we get that 
\begin{align}\label{ahR12}
\Big(\sum_{j\geq n}\frac{\gamma_jA_{j}}{(R_{1}-\gamma_jA_{j})^2}s_{j,0}\Big)\frac{\partial R_1}{\partial \gamma_v}=\frac{A_vR_1}{(R_1-\gamma_vA_v)^2}s_{v,0}.
\end{align}
Thus, $\frac{\partial R_1}{\partial \gamma_v}>0$. 

{\bf Part 2.1.} For $v\in \{n,\ldots,h\}$, we compute that
\begin{align*}
\frac{1}{A_h^t}\dfrac{\partial Y_{t+1}}{\partial \gamma_v}=&\frac{\partial R_1}{\partial \gamma_v}\Big(\sum_{i<n}\beta_i^t s_{i,0}-\sum_{n\leq j\leq h}\beta_j^{t}\frac{(1-\gamma_j)\gamma_jA_j^2}{(R_1-\gamma_jA_{j})^2}s_{j,0}-
\sum_{j> h}\Big(\beta_j\frac{(1-\gamma_j)A_{j}}{A_h-\gamma_jA_{j}}\Big)^{t}\frac{(1-\gamma_j)\gamma_jA_j^2}{(R_1-\gamma_jA_{j})^2}s_{j,0}
\Big)\\
&+\beta_v^{t}\frac{A_vR_1(A_v-R_1)}{(R_1-\gamma_vA_{v})^2}s_{v,0}
\end{align*}
Combining with  (\ref{ahR12}), we get that
\begin{align*}
&\frac{1}{A_h^t}\dfrac{\partial Y_{t+1}}{\partial \gamma_v}\frac{1}{\frac{\partial R_1}{\partial \gamma_v}}=\Big(\sum_{i<n}\beta_i^t s_{i,0}-\sum_{n\leq j\leq h}\beta_j^{t}\frac{(1-\gamma_j)\gamma_jA_j^2}{(R_1-\gamma_jA_{j})^2}s_{j,0}-
\sum_{j> h}\Big(\beta_j\frac{(1-\gamma_j)A_{j}}{A_h-\gamma_jA_{j}}\Big)^{t}\frac{(1-\gamma_j)\gamma_jA_j^2}{(R_1-\gamma_jA_{j})^2}s_{j,0}
\Big)\\
&+\beta_v^{t}\frac{A_vR_1(A_v-R_1)}{(R_1-\gamma_vA_{v})^2}s_{v,0}\frac{\sum_{j\geq n}\frac{\gamma_jA_{j}}{(R_{1}-\gamma_jA_{j})^2}s_{j,0}}{\frac{A_vR_1}{(R_1-\gamma_vA_v)^2}s_{v,0}}\\
&=\sum_{i<n}\beta_i^t s_{i,0}-\sum_{n\leq j\leq h}\beta_j^{t}\frac{(1-\gamma_j)\gamma_jA_j^2}{(R_1-\gamma_jA_{j})^2}s_{j,0}-
\sum_{j> h}\Big(\beta_j\frac{(1-\gamma_j)A_{j}}{A_h-\gamma_jA_{j}}\Big)^{t}\frac{(1-\gamma_j)\gamma_jA_j^2}{(R_1-\gamma_jA_{j})^2}s_{j,0}
\Big)\\
&+\beta_v^{t}(A_v-R_1)\sum_{j\geq n}\frac{\gamma_jA_{j}}{(R_{1}-\gamma_jA_{j})^2}s_{j,0}\\
&=\sum_{i<n}\beta_i^t s_{i,0}-\sum_{n\leq j\leq h}\beta_j^{t}\frac{(1-\gamma_j)\gamma_jA_j^2}{(R_1-\gamma_jA_{j})^2}s_{j,0}-
\sum_{j> h}\Big(\beta_j\frac{(1-\gamma_j)A_{j}}{A_h-\gamma_jA_{j}}\Big)^{t}\frac{(1-\gamma_j)\gamma_jA_j^2}{(R_1-\gamma_jA_{j})^2}s_{j,0}
\Big)\\
&+\beta_v^{t}(A_v-R_1)\sum_{j\geq n}\frac{\gamma_jA_{j}}{(R_{1}-\gamma_jA_{j})^2}s_{j,0}.
\end{align*}

We now assume that 
$\beta_h>\max_{i\not=h}\beta_i$.
When $v\not=h$, it is easy to see that  $\dfrac{\partial Y_{t+1}}{\partial \gamma_v}<0$ for $t$ high enough.

When $v=h$, observe that
\begin{align}
&\beta_v^{t}\frac{(1-\gamma_v)\gamma_vA_v^2}{(R_1-\gamma_vA_{v})^2}s_{v,0}-\beta_v^{t}(A_v-R_1)\frac{\gamma_vA_{v}}{(R_{1}-\gamma_vA_{v})^2}s_{v,0}\\&=\frac{\beta_v^ts_{v,0}}{(R_{1}-\gamma_vA_{v})^2}\gamma_vA_v\Big((1-\gamma_v)A_v-(A_v-R_1) \Big)\\
&=\frac{\beta_v^ts_{v,0}}{(R_{1}-\gamma_vA_{v})^2}\gamma_vA_v\Big(R_1-\gamma_vA_v\Big)=\frac{\beta_v^ts_{v,0}}{(R_{1}-\gamma_vA_{v})}\gamma_vA_v.
\end{align}
By consequence, if $\beta_h>\max_{i\not=h}\beta_i$, then there exists $t_0$ such that $\dfrac{\partial Y_{t+1}}{\partial \gamma_h}<0,\forall t\geq t_0$.

{\bf Part 2.2.}  For agent $v>h$, we compute that
\begin{align*}
\frac{1}{A_h^t}\dfrac{\partial Y_{t+1}}{\partial \gamma_v}=&\frac{\partial R_1}{\partial \gamma_v}\Big(\sum_{i<n}\beta_i^t s_{i,0}-\sum_{n\leq j\leq h}\beta_j^{t}\frac{(1-\gamma_j)\gamma_jA_j^2}{(R_1-\gamma_jA_{j})^2}s_{j,0}-
\sum_{j> h}\Big(\beta_j\frac{(1-\gamma_j)A_{j}}{A_h-\gamma_jA_{j}}\Big)^{t}\frac{(1-\gamma_j)\gamma_jA_j^2}{(R_1-\gamma_jA_{j})^2}s_{j,0}
\Big)\\
&+\Big(\frac{\beta_v(1-\gamma_v)A_v}{A_h-\gamma_vA_v}\Big)^t\frac{A_vR_1\big[t(R_1-\gamma_vA_v)(A_v-A_h)+(A_v-R_1)(A_h-\gamma_vA_v)\big]}{(A_h-\gamma_vA_v)(R_1-\gamma_vA_v)^2}s_{v,0}
\end{align*}
\begin{align*}
&\frac{1}{A_h^t}\dfrac{\partial Y_{t+1}}{\partial \gamma_v}\frac{1}{\frac{\partial R_1}{\partial \gamma_v}}
=\Big(\sum_{i<n}\beta_i^t s_{i,0}-\sum_{n\leq j\leq h}\beta_j^{t}\frac{(1-\gamma_j)\gamma_jA_j^2}{(R_1-\gamma_jA_{j})^2}s_{j,0}-
\sum_{j> h}\Big(\beta_j\frac{(1-\gamma_j)A_{j}}{A_h-\gamma_jA_{j}}\Big)^{t}\frac{(1-\gamma_j)\gamma_jA_j^2}{(R_1-\gamma_jA_{j})^2}s_{j,0}
\Big)\\
&+\Big(\frac{\beta_v(1-\gamma_v)A_v}{A_h-\gamma_vA_v}\Big)^t\frac{A_vR_1\big[t(R_1-\gamma_vA_v)(A_v-A_h)+(A_v-R_1)(A_h-\gamma_vA_v)^2\big]}{(A_h-\gamma_vA_v)(R_1-\gamma_vA_v)}s_{v,0}\frac{\sum_{j\geq n}\frac{\gamma_jA_{j}}{(R_{1}-\gamma_jA_{j})^2}s_{j,0}}{\frac{A_vR_1}{(R_1-\gamma_vA_v)^2}s_{v,0}}\\
&=\Big(\sum_{i<n}\beta_i^t s_{i,0}-\sum_{n\leq j\leq h}\beta_j^{t}\frac{(1-\gamma_j)\gamma_jA_j^2}{(R_1-\gamma_jA_{j})^2}s_{j,0}-
\sum_{j> h}\Big(\beta_j\frac{(1-\gamma_j)A_{j}}{A_h-\gamma_jA_{j}}\Big)^{t}\frac{(1-\gamma_j)\gamma_jA_j^2}{(R_1-\gamma_jA_{j})^2}s_{j,0}
\Big)\\
&+\Big(\frac{\beta_v(1-\gamma_v)A_v}{A_h-\gamma_vA_v}\Big)^t\frac{A_vR_1\big[t(R_1-\gamma_vA_v)(A_v-A_h)+(A_v-R_1)(A_h-\gamma_vA_v)\big]}{(A_h-\gamma_vA_v)(R_1-\gamma_vA_v)^2\frac{A_vR_1}{(R_1-\gamma_vA_v)^2}}\Big(\sum_{j\geq n}\frac{\gamma_jA_{j}}{(R_{1}-\gamma_jA_{j})^2}s_{j,0}\Big)\\
&=\Big(\sum_{i<n}\beta_i^t s_{i,0}-\sum_{n\leq j\leq h}\beta_j^{t}\frac{(1-\gamma_j)\gamma_jA_j^2}{(R_1-\gamma_jA_{j})^2}s_{j,0}-
\sum_{j> h}\Big(\beta_j\frac{(1-\gamma_j)A_{j}}{A_h-\gamma_jA_{j}}\Big)^{t}\frac{(1-\gamma_j)\gamma_jA_j^2}{(R_1-\gamma_jA_{j})^2}s_{j,0}
\Big)\\
&+\Big(\frac{\beta_v(1-\gamma_v)A_v}{A_h-\gamma_vA_v}\Big)^t\Big(\frac{t(R_1-\gamma_vA_v)(A_v-A_h)}{(A_h-\gamma_vA_v)}+A_v-R_1\Big)\Big(\sum_{j\geq n}\frac{\gamma_jA_{j}}{(R_{1}-\gamma_jA_{j})^2}s_{j,0}\Big)
\end{align*}

\end{proof}

\subsubsection{Additional results}
\label{additional}

\begin{proposition}[\textcolor{blue}{\bf  equilibrium with $R_1\in (A_{m-1},A_m), R_t=A_h, \forall t\geq 2$, $h<m$}]
\label{m-1mh}  Assume that $u_i(c)=ln(c)$, $\forall i,\forall c>0$, $F_{i,t}(k)=A_ik$, $\forall i,\forall k\geq 0$ with $\max_i{\gamma_iA_i}<A_1<A_2<\ldots<A_m$, and 
 \begin{subequations}\begin{align}
 \label{magentmh}
\frac{\beta_h^ts_{h,0}}{1-\gamma_h}&\geq \sum_{i\leq h}\beta_i^t s_{i,0}-\sum_{j=h+1}^{m-1}\frac{\gamma_jA_{j}}{A_h-\gamma_jA_{j}}\Big(\beta_j\frac{(1-\gamma_j)A_{j}}{A_h-\gamma_jA_{j}}\Big)^{t-1}\beta_j s_{j,0}\notag\\
&\quad \quad \quad \quad \quad \quad \quad \quad -\Big(\beta_m\frac{(1-\gamma_m)A_{m}}{A_h-\gamma_mA_{m}}\Big)^{t}(S_0- s_{m,0}) \quad  \geq 0\\
\gamma_m&<\frac{\sum_{i\not=m}s_{i,0}}{S_0}\\
\frac{A_{m-1}}{A_m}&<\gamma_m\frac{S_0}{\sum_{i\not=m}s_{i,0}},
\end{align} \end{subequations}
where $h\leq m-1$. 

Then, there exists an equilibrium where the interest rates are determined by
\begin{align}
R_1&=\gamma_mA_m\frac{S_0}{\sum_{i\not=m}s_{i,0}} \in (A_{m-1},A_m), \quad 
R_t=A_h, \forall t\geq 2,
\end{align}
where $S_0\equiv \sum_{i=1}^ms_{i,0}$.
\begin{enumerate}
\item 
In such an equilibrium, only agent $m$ produces at date 1 but (m-h+1) agents produces from date $2$ on. The individual capital is given by
\begin{align*}
k_{j,0}&=\begin{cases}0, &\forall j<m\\
S_0, &\forall j=m
\end{cases}\\
k_{j,t}&=\begin{cases}0, &\forall j<h\\
\sum_{i\leq h}\beta_i^t A_h^{t-1}R_{1}s_{i,0}-\sum_{j=h+1}^{m-1}\frac{\gamma_jA_{j}}{A_h-\gamma_jA_{j}}\Big(\beta_j\frac{(1-\gamma_j)A_{j}A_h}{A_h-\gamma_jA_{j}}\Big)^{t-1}\beta_jR_1 s_{j,0}\\
-\frac{\gamma_mA_{m}}{A_h-\gamma_mA_{m}}\Big(\beta_m\frac{(1-\gamma_m)A_{m}A_h}{A_h-\gamma_mA_{m}}\Big)^{t-1}\Big(\beta_m\frac{(1-\gamma_m)A_{m}R_1}{R_1-\gamma_mA_{m}}\Big) s_{m,0},& \text{ for j=h}\\
\frac{A_h}{A_h-\gamma_jA_{j}}\Big(\beta_j\frac{(1-\gamma_j)A_{j}A_h}{A_h-\gamma_jA_{j}}\Big)^{t-1}\beta_jR_1 s_{j,0} &\forall h<j<m\\
\frac{A_h}{A_h-\gamma_mA_{m}}\Big(\beta_m\frac{(1-\gamma_m)A_{m}A_h}{A_h-\gamma_mA_{m}}\Big)^{t-1}\Big(\beta_m\frac{(1-\gamma_m)A_{m}R_1}{R_1-\gamma_mA_{m}}\Big) s_{m,0} &\text{ for } j=m.
\end{cases}
\end{align*}
\item The aggregate output is increasing in the credit limit $\gamma_j$ of each producer $j$.
\end{enumerate}

\end{proposition}

\begin{proof}[{\bf Proof of Proposition \ref{m-1mh}}]
Let us focus on an equilibrium where only the most productive agent produces. The interest rate $R_{1}\in (A_{m-1},A_m)$ and $R_t=A_h$, $\forall t\geq 2$.

  Denote the individual saving $s_{i,t}\equiv k_{i,t}-b_{i,t}$.

First, we observe that \begin{align}
s_{i,0}=\beta_iw_{i,0}, \forall i.
\end{align}
At date $0$, since $R_1\in (A_{m-1},A_m)$, we have 
\begin{align*}
k_{i,0}&=0, \quad b_{i,0}=-s_{i,0}, \forall i<m\\
A_{m}k_{m,0} - R_1b_{m,0}&=(1-\gamma_m)A_{m}k_{m,0}\\
s_{m,0}&=k_{m,0}-b_{m,0}=k_{m,0}\Big(1-\frac{\gamma_mA_{m}}{R_{1}}\Big)\\
 k_{m,0}&=\frac{R_1}{R_1-\gamma_mA_{m}}s_{m,0}, \quad 
  b_{m,0}=\frac{\gamma_mA_{m}}{R_{1}-\gamma_mA_{m}}s_{m,0}\end{align*}

We now look at equilibrium. From the market clearing condition $\sum_{i}b_{i,0}=0$, we have that 
\begin{align*}
-\sum_{i\not=m}b_{i,0}&=b_{m,0} \Leftrightarrow \sum_{i\not=m}s_{i,0}=\frac{\gamma_mA_{m}}{R_{1}-\gamma_mA_{m}}s_{m,0}\\  \Leftrightarrow R_1&=\gamma_mA_m(1+\frac{s_{m,0}}{\sum_{i\not=m}s_{i,0}})=\gamma_mA_m\frac{S_0}{\sum_{i\not=m}s_{i,0}}.
\end{align*}  

\textcolor{blue}{Let $h\leq m-1$.} For each agent $i<h$, we have $k_{i,t}=0$, $s_{i,t}=-b_{i,t}$, $\forall t$.  

Since $R_t=A_h, \forall t\geq 2$, we can compute that, for any $i\leq h$,
\begin{align*}
 s_{i,0}&=\beta_i w_{i,0}, \quad    s_{i,t}=\beta_i R_{t}s_{i,t-1} \text{ } \forall t\geq 1\\
   s_{i,t}&=\beta_i^t A_h^{t-1}R_{1}s_{i,0}, t\geq 1.
\end{align*}
The capital $k_{h,t}$ will be determined by the market clearing condition. 

For any $m$, we have 

For each agent $j$ with $h<j<m$ and , their borrowing constraints bind at any date $t\geq 1$: $R_{t+1}b_{j,t}= \gamma_j A_{j}k_{j,t}$. Therefore, we have $s_{j,t}=k_{j,t}\Big(1-\frac{\gamma_jA_{j,t+1}}{R_{t+1}}\Big), \forall t\geq 1.$ From this, we can compute that
\begin{align*}
 s_{j,0}&=\beta_j w_{j,0}, \quad k_{j,0}=0,\\
  s_{j,1}&=\beta_j R_1 s_{j,0},\quad s_{j,t}=\beta_j\frac{(1-\gamma_j)A_{j,t}R_t}{R_t-\gamma_jA_{j,t}}s_{j,t-1}, \forall t\geq 2 \\
  s_{j,t}&=\Big(\beta_j\frac{(1-\gamma_j)A_{j}A_h}{A_h-\gamma_jA_{j}}\Big)^{t-1}\beta_jR_1 s_{j,0},\forall t\geq 1\\
 k_{j,t}&=\frac{R_{t+1}}{R_{t+1}-\gamma_jA_{j,t+1}}s_{j,t}=\frac{A_h}{A_h-\gamma_jA_{j}}s_{j,t}=\frac{A_h}{A_h-\gamma_jA_{j}}\Big(\beta_j\frac{(1-\gamma_j)A_{j}A_h}{A_h-\gamma_jA_{j}}\Big)^{t-1}\beta_jR_1 s_{j,0}\\
  b_{j,t}&=\frac{\gamma_jA_{j,t+1}}{R_{t+1}-\gamma_jA_{j,t+1}}s_{j,t}=\frac{\gamma_jA_{j}}{A_h-\gamma_jA_{j,t}}s_{j,t}, \forall t\geq 1.
\end{align*}

For each agent $j=m$, we have 
\begin{align*}
 s_{m,0}&=\beta_m w_{m,0}, \quad s_{m,t}=\beta_m\frac{(1-\gamma_m)A_{m,t}R_t}{R_t-\gamma_mA_{m,t}}s_{m,t-1}, \forall t\geq 1 \\
  s_{m,t}&=\Big(\beta_m\frac{(1-\gamma_m)A_{m}A_h}{A_h-\gamma_mA_{m}}\Big)^{t-1}\Big(\beta_m\frac{(1-\gamma_m)A_{m}R_1}{R_1-\gamma_mA_{m}}\Big) s_{m,0},\forall t\geq 1\\
 k_{m,t}&=\frac{R_{t+1}}{R_{t+1}-\gamma_mA_{m,t+1}}s_{m,t}=\frac{A_h}{A_h-\gamma_mA_{m}}s_{j,t}\\
 &=\frac{A_h}{A_h-\gamma_mA_{m}}\Big(\beta_m\frac{(1-\gamma_m)A_{m}A_h}{A_h-\gamma_mA_{m}}\Big)^{t-1}\Big(\beta_m\frac{(1-\gamma_m)A_{m}R_1}{R_1-\gamma_mA_{m}}\Big) s_{m,0}\\
  b_{m,t}&=\frac{\gamma_mA_{m,t+1}}{R_{t+1}-\gamma_mA_{m,t+1}}s_{m,t}=\frac{\gamma_mA_{m}}{A_h-\gamma_mA_{m}}s_{m,t}, \forall t\geq 1.
\end{align*}

From the market clearing condition $\sum_{i}b_{i,t}=0$, we have $\sum_{i}s_{i,t}=\sum_ik_{i,t}$ which implies that
\begin{align}
k_{h,t}=&\sum_{i}s_{i,t}-\sum_{i\not=h}k_{i,t}=
\sum_{i}s_{i,t}-\sum_{i>h}k_{i,t} \text{ (since $k_{i,t}=0,\forall i<h$)}\\
=&\sum_{i\leq h}s_{i,t}-\sum_{j> h}b_{j,t}\\
=&\sum_{i\leq h}\beta_i^t A_h^{t-1}R_{1}s_{i,0}-\sum_{j=h+1}^{m-1}\frac{\gamma_jA_{j}}{A_h-\gamma_jA_{j}}\Big(\beta_j\frac{(1-\gamma_j)A_{j}A_h}{A_h-\gamma_jA_{j}}\Big)^{t-1}\beta_jR_1 s_{j,0}\\
&-\frac{\gamma_mA_{m}}{A_h-\gamma_mA_{m}}\Big(\beta_m\frac{(1-\gamma_m)A_{m}A_h}{A_h-\gamma_mA_{m}}\Big)^{t-1}\Big(\beta_m\frac{(1-\gamma_m)A_{m}R_1}{R_1-\gamma_mA_{m}}\Big) s_{m,0}.
\end{align}
We will verify that $0\leq k_{h,t}  $ and $R_{t+1}b_{h,t}\leq \gamma_hA_hk_{h,t}$. 
We now check that the above list $((c_{i,t},k_{i,t},b_{i,t})_i,R_t)_t$ is an equilibrium. We use Lemma \ref{ns}. It is easy to verify the market clearing conditions and the FOCs.

Condition $R_1\in (A_{m-1},A_m)$ is ensured by the assumption that 
$$A_{m-1}<\gamma_mA_m(1+\frac{s_{m,0}}{\sum_{i\not=m}s_{i,0}})<A_m.$$
\begin{itemize}
\item We verify borrowing constraints: $R_{t+1}b_{i,t}\leq \gamma_iA_ik_{i,t}$ and $k_{i,t}\geq 0$, $\forall t\geq 0$. It is clear for any $j>m$. Let us consider agent $h$. This is satisfied for $t=0$. Let us consider $t\geq 1$. Since $R_{t+1}=A_h$, $\forall t\geq 1$, this becomes $k_{h,t}-s_{h,t}=b_{h,t}\leq \gamma_hk_{h,t}$, or, equivalently, $(1-\gamma_h)k_{h,t}\leq s_{h,t}$. Note that $0\leq k_{h,t}  \leq \frac{s_{h,t}}{1-\gamma_h}$ becomes
\begin{align}
\frac{\beta_h^ts_{h,0}}{1-\gamma_h}\geq &\sum_{i\leq h}\beta_i^t s_{i,0}-\sum_{j=h+1}^{m-1}\frac{\gamma_jA_{j}}{A_h-\gamma_jA_{j}}\Big(\beta_j\frac{(1-\gamma_j)A_{j}}{A_h-\gamma_jA_{j}}\Big)^{t-1}\beta_j s_{j,0}\\
&-\frac{\gamma_mA_{m}}{A_h-\gamma_mA_{m}}\Big(\beta_m\frac{(1-\gamma_m)A_{m}}{A_h-\gamma_mA_{m}}\Big)^{t-1}\Big(\beta_m\frac{(1-\gamma_m)A_{m}}{R_1-\gamma_mA_{m}}\Big) s_{m,0}\geq 0\\
\Leftrightarrow
\frac{\beta_h^ts_{h,0}}{1-\gamma_h}\geq &\sum_{i\leq h}\beta_i^t s_{i,0}-\sum_{j=h+1}^{m-1}\frac{\gamma_jA_{j}}{A_h-\gamma_jA_{j}}\Big(\beta_j\frac{(1-\gamma_j)A_{j}}{A_h-\gamma_jA_{j}}\Big)^{t-1}\beta_j s_{j,0}\\
&-\Big(\beta_m\frac{(1-\gamma_m)A_{m}}{A_h-\gamma_mA_{m}}\Big)^{t}\Big(\frac{\gamma_mA_{m}}{R_1-\gamma_mA_{m}}\Big) s_{m,0}\geq 0.
\end{align}
Since  $R_1=\gamma_mA_m(1+\frac{s_{m,0}}{\sum_{i\not=m}s_{i,0}})$, this is equivalent to
\begin{align}
\frac{\beta_h^ts_{h,0}}{1-\gamma_h}\geq &\sum_{i\leq h}\beta_i^t s_{i,0}-\sum_{j=h+1}^{m-1}\frac{\gamma_jA_{j}}{A_h-\gamma_jA_{j}}\Big(\beta_j\frac{(1-\gamma_j)A_{j}}{A_h-\gamma_jA_{j}}\Big)^{t-1}\beta_j s_{j,0}\\
&-\frac{\gamma_mA_{m}}{A_h-\gamma_mA_{m}}\Big(\beta_m\frac{(1-\gamma_m)A_{m}}{A_h-\gamma_mA_{m}}\Big)^{t-1}\Big(\beta_m\frac{(1-\gamma_m)A_{m}}{R_1-\gamma_mA_{m}}\Big) s_{m,0}\geq 0\\
\Leftrightarrow
\frac{\beta_h^ts_{h,0}}{1-\gamma_h}\geq &\sum_{i\leq h}\beta_i^t s_{i,0}-\sum_{j=h+1}^{m-1}\frac{\gamma_jA_{j}}{A_h-\gamma_jA_{j}}\Big(\beta_j\frac{(1-\gamma_j)A_{j}}{A_h-\gamma_jA_{j}}\Big)^{t-1}\beta_j s_{j,0}\\
&-\Big(\beta_m\frac{(1-\gamma_m)A_{m}}{A_h-\gamma_mA_{m}}\Big)^{t}(S_0- s_{m,0})\geq 0.
\end{align}
\item Tranversality conditions:  $\lim_{T\to \infty}\beta_i^Tu_i'(c_{i,T})(k_{i,T}-b_{i,T})=0$. It is easy to verify these conditions because $\beta_i\in (0,1)$ and $u'(c)=1/c$.

\end{itemize}

The aggregate output at date $1$ is $Y_1=A_mS_0$. We now compute
\begin{align*}
Y_{t+1}=&\sum_{i\geq h}A_ik_{i,t}=A_h\Big(\sum_{i\leq h}\beta_i^t A_h^{t-1}R_{1}s_{i,0}-\sum_{j=h+1}^{m-1}\frac{\gamma_jA_{j}}{A_h-\gamma_jA_{j}}\Big(\beta_j\frac{(1-\gamma_j)A_{j}A_h}{A_h-\gamma_jA_{j}}\Big)^{t-1}\beta_jR_1 s_{j,0}\\
&-\frac{\gamma_mA_{m}}{A_h-\gamma_mA_{m}}\Big(\beta_m\frac{(1-\gamma_m)A_{m}A_h}{A_h-\gamma_mA_{m}}\Big)^{t-1}\Big(\beta_m\frac{(1-\gamma_m)A_{m}R_1}{R_1-\gamma_mA_{m}}\Big) s_{m,0}\Big)\\
&+\sum_{j=h+1}^{m-1}A_j\frac{A_h}{A_h-\gamma_jA_{j}}\Big(\beta_j\frac{(1-\gamma_j)A_{j}A_h}{A_h-\gamma_jA_{j}}\Big)^{t-1}\beta_jR_1 s_{j,0}\\
&+A_m\frac{A_h}{A_h-\gamma_mA_{m}}\Big(\beta_m\frac{(1-\gamma_m)A_{m}A_h}{A_h-\gamma_mA_{m}}\Big)^{t-1}\Big(\beta_m\frac{(1-\gamma_m)A_{m}R_1}{R_1-\gamma_mA_{m}}\Big) s_{m,0}.
\end{align*}
Hence, we get \begin{align*}
\frac{Y_{t+1}}{A_h^t}=&\sum_{i\leq h}\beta_i^t R_{1}s_{i,0}+\sum_{j=h+1}^{m-1}\frac{A_j(1-\gamma_j)}{A_h-\gamma_jA_{j}}\Big(\beta_j\frac{(1-\gamma_j)A_{j}}{A_h-\gamma_jA_{j}}\Big)^{t-1}\beta_jR_1 s_{j,0}\\
&+\frac{A_m(1-\gamma_m)}{A_h-\gamma_mA_{m}}\Big(\beta_m\frac{(1-\gamma_m)A_{m}}{A_h-\gamma_mA_{m}}\Big)^{t-1}\Big(\beta_m\frac{(1-\gamma_m)A_{m}R_1}{R_1-\gamma_mA_{m}}\Big) s_{m,0}\\
=&\sum_{i\leq h}\beta_i^t R_{1}s_{i,0}+\sum_{j=h+1}^{m-1}\Big(\beta_j\frac{(1-\gamma_j)A_{j}}{A_h-\gamma_jA_{j}}\Big)^{t}R_1 s_{j,0}+\Big(\beta_m\frac{(1-\gamma_m)A_{m}}{A_h-\gamma_mA_{m}}\Big)^{t}\frac{(1-\gamma_m)A_{m}R_1}{R_1-\gamma_mA_{m}} s_{m,0}.
\end{align*}
Since $R_1=\gamma_mA_m\frac{S_0}{\sum_{i\not=m}s_{i,0}},$ we have that $$\frac{R_1}{R_1-\gamma_mA_{m}}s_{m,0}=\frac{\gamma_mA_m\frac{S_0}{\sum_{i\not=m}s_{i,0}}}{\gamma_mA_m\frac{S_0}{\sum_{i\not=m}s_{i,0}}-\gamma_mA_{m}}s_{m,0}=S_0.$$
Therefore, we get that
\begin{align*}
\frac{Y_{t+1}}{A_h^tA_mS_0}=&\gamma_m\frac{\sum_{i\leq h}\beta_i^t s_{i,0}+\sum_{j=h+1}^{m-1}\Big(\beta_j\frac{(1-\gamma_j)A_{j}}{A_h-\gamma_jA_{j}}\Big)^{t} s_{j,0}}{\sum_{i<m}s_{i,0}}+\Big(\beta_m\frac{(1-\gamma_m)A_{m}}{A_h-\gamma_mA_{m}}\Big)^{t} (1-\gamma_m).
\end{align*}
From this, we can see that $\frac{Y_{t+1}}{\partial \gamma_j}>0,\forall h+1<j<m-1$ since $\beta_j\frac{(1-\gamma_j)A_{j}}{A_h-\gamma_jA_{j}}$ is increasing in $\gamma_j$. The intuition is simple: the credit limits of these agents do not affect the equilibrium interest rate while it allows these producers to borrow more and produce more.

We now look at the effect of $\gamma_m$. In terms of interest, this credit limit positively affects the interest rate $R_1$ and hence the savings of any agents. 

\begin{align*}
\frac{\frac{\partial Y_{t+1}}{\partial \gamma_m}}{A_h^tA_mS_0}=&\Big(\frac{\sum_{i\leq h}\beta_i^t s_{i,0}}{\sum_{i<m}s_{i,0}}+\frac{\sum_{j=h+1}^{m-1}\Big(\beta_j\frac{(1-\gamma_j)A_{j}}{A_h-\gamma_jA_{j}}\Big)^{t} s_{j,0}}{\sum_{i<m}s_{i,0}}\Big)+(\beta_mA_m)^t \Big(\frac{(1-\gamma_m)^{t+1}}{(A_h-\gamma_mA_{m})^t}\Big)\\
\end{align*}
We compute
\begin{align*}
&\frac{\partial \Big(\frac{(1-\gamma_m)^{t+1}}{(A_h-\gamma_mA_{m})^t}\Big)}{\partial \gamma_m}
=\frac{-(t+1)(1-\gamma_m)^{t}(A_h-\gamma_mA_{m})^t+tA_m(A_h-\gamma_mA_{m})^{t-1}(1-\gamma_m)^{t+1}}{(A_h-\gamma_mA_{m})^{2t}}\\
=&(1-\gamma_m)^{t}\frac{-(t+1)(A_h-\gamma_mA_{m})+tA_m(1-\gamma_m)}{(A_h-\gamma_mA_{m})^{t+1}}\\
=&(1-\gamma_m)^{t}\frac{-(t+1)A_h-\gamma_mA_{m}+tA_m}{(A_h-\gamma_mA_{m})^{t+1}}=(1-\gamma_m)^{t}\frac{t(A_m-A_h)-(A_h-\gamma_mA_m)}{(A_h-\gamma_mA_{m})^{t+1}}\\
=&(1-\gamma_m)^{t}\frac{-(t+1)A_h-\gamma_mA_{m}+tA_m}{(A_h-\gamma_mA_{m})^{t+1}}=\frac{t(A_m-A_h)(1-\gamma_m)^{t}}{(A_h-\gamma_mA_{m})^{t+1}}-\frac{(1-\gamma_m)^{t}}{(A_h-\gamma_mA_{m})^{t}}.
\end{align*}
Therefore, we get that 
\begin{align*}
\Big(\frac{\sum_{i<m}s_{i,0}}{A_h^tA_mS_0}\Big)\frac{\partial Y_{t+1}}{\partial \gamma_m}=&
\sum_{i\leq h}\beta_i^t s_{i,0}+\sum_{j=h+1}^{m-1}\Big(\beta_j\frac{(1-\gamma_j)A_{j}}{A_h-\gamma_jA_{j}}\Big)^{t} s_{j,0}-\Big(\beta_m\frac{(1-\gamma_m)A_{m}}{A_h-\gamma_mA_{m}}\Big)^{t}\sum_{i<m}s_{i,0}\\
&+\frac{(\beta_m(1-\gamma_m)A_{m})^t}{(A_h-\gamma_mA_{m})^{t+1}}(A_m-A_h)t\sum_{i<m}s_{i,0}.
\end{align*}
This is strictly positive thanks to the assumption (\ref{magentmh}).

\end{proof}

\begin{proposition}[\textcolor{blue}{\bf equilibrium with $R_t\in (A_{m-1,t},A_{m,t}),\forall t$}]
\label{m-1m}
Assume that $F_{i,t}(k)=A_ik$, $\forall i,\forall k\geq 0$ with $\max_i{\gamma_iA_i}<A_1<A_2<\ldots<A_m$, and utility function $u_i(c)=ln(c)$ $\forall i$. 

Assume also that
\begin{align}
\frac{A_{m-1,1}}{A_{m,1}}&<\gamma_m\frac{S_0}{\sum_{i<m}s_{i,0}}<1\\
\frac{A_{m-1,t+1}}{A_{m,t+1}}&<\gamma_m+(1-\gamma_m)\beta_m\frac{\sum_{i<m}\beta_i^{t-1} s_{i,0}}{\sum_{i\not=m}\beta_i^t s_{i,0}}, \forall t\geq 1\\
\beta_m\frac{\sum_{i<m}\beta_i^{t-1} s_{i,0}}{\sum_{i\not=m}\beta_i^t s_{i,0}}&<1, \forall t\geq 1.
\end{align}
Then, there exists an equilibrium whose the interest rates are 
\begin{align}
R_1&=\gamma_mA_{m,1}\frac{S_0}{\sum_{i<m}s_{i,0}}\\
R_{t+1}&=A_{m,t+1}\Big(\gamma_m+(1-\gamma_m)\beta_m\frac{\sum_{i<m}\beta_i^{t-1} s_{i,0}}{\sum_{i<m}\beta_i^t s_{i,0}}\Big), \forall t\geq 1,
\end{align} Observe that $R_t\in (A_{m-1,t},A_{m,t}), \forall t$. When $A_{m,t}$ converges to $A_m$, then we have 
\begin{align}
\lim_{t\to\infty}R_t=A_m\Big(\gamma_m+\frac{\beta_m}{\beta_{i_0}}(1-\gamma_m)\Big).
\end{align}

In this equilibrium, the aggregate output is increasing in the credit limit $\gamma_m$ at any date.

\end{proposition}
\begin{proof}[{\bf Proof of Proposition \ref{m-1m}}]
Consider an equilibrium with $R_t\in (A_{m-1,t},A_{m,t}), \forall t$. 

For agent $i<m$, since $A_{i,t}<R_t$, $\forall t$, we have $k_{i,t}=0$ and hence we find that  
\begin{align*}
 s_{i,0}&=\beta_i w_{i,0}, \quad    s_{i,t}=\beta_i R_ts_{i,t-1} \text{ } \forall t\geq 1\\
   -b_{i,t}=s_{i,t}&=\beta_i^t R_{t}\cdots R_{1}s_{i,0}.
\end{align*}

For agent $j>h$, since $A_{j,t}>R_t=A_{h,t}$, $\forall t$, her borrowing constraint is always binding: $R_tb_{j,t-1}= \gamma_j A_{j,t}k_{j,t-1}$. Therefore, we have $$s_{j,t}=k_{j,t}\Big(1-\frac{\gamma_jA_{j,t+1}}{R_{t+1}}\Big), \quad A_{j,t}k_{j,t-1} - R_tb_{j,t-1}=(1-\gamma_j)A_{j,t}k_{j,t-1}, \forall t\geq 1.$$ From this, we can compute that
\begin{align*}
 s_{j,0}&=\beta_j w_{j,0},\\
 s_{j,t}&=\beta_j\frac{(1-\gamma_j)A_{j,t}R_t}{R_t-\gamma_jA_{j,t}}s_{j,t-1}=\Big(\beta_j\frac{(1-\gamma_j)A_{j,t}R_t}{R_t-\gamma_jA_{j,t}}\Big)\cdots \Big(\beta_j\frac{(1-\gamma_j)A_{j,1}R_1}{R_1-\gamma_jA_{j,1}}\Big) s_{j,0}\\
 k_{j,t}&=\frac{1}{1-\frac{\gamma_jA_{j,t+1}}{R_{t+1}}}s_{j,t}=\frac{R_{t+1}}{R_{t+1}-\gamma_jA_{j,t+1}}s_{j,t}\\
 &=\frac{R_{t+1}}{R_{t+1}-\gamma_jA_{j,t+1}}\Big(\beta_j\frac{(1-\gamma_j)A_{j,t}R_t}{R_t-\gamma_jA_{j,t}}\Big)\cdots \Big(\beta_j\frac{(1-\gamma_j)A_{j,1}R_1}{R_1-\gamma_jA_{j,1}}\Big) s_{j,0}\\
  b_{j,t}&= \frac{\gamma_jA_{j,t+1}}{R_{t+1}}k_{j,t}=\frac{\gamma_jA_{j,t+1}}{R_{t+1}-\gamma_jA_{j,t+1}}s_{j,t}\\
  &=\frac{\gamma_jA_{j,t+1}}{R_{t+1}-\gamma_jA_{j,t+1}}\Big(\beta_j\frac{(1-\gamma_j)A_{j,t}R_t}{R_t-\gamma_jA_{j,t}}\Big)\cdots \Big(\beta_j\frac{(1-\gamma_j)A_{j,1}R_1}{R_1-\gamma_jA_{j,1}}\Big) s_{j,0}.
\end{align*}

The market clearing condition writes $\sum_ib_{i,t}=0$, i.e., $\sum_{i>h}b_{i,t}=-\sum_{i<h}b_{i,t}=0$. This becomes
\begin{align*}
\sum_{j>h}\frac{\gamma_jA_{j,t+1}}{R_{t+1}-\gamma_jA_{j,t+1}}\Big(\beta_j\frac{(1-\gamma_j)A_{j,t}R_t}{R_t-\gamma_jA_{j,t}}\Big)\cdots \Big(\beta_j\frac{(1-\gamma_j)A_{j,1}R_1}{R_1-\gamma_jA_{j,1}}\Big) s_{j,0}=\sum_{i<h}\beta_i^t R_{t}\cdots R_{1}s_{i,0}\\
\Leftrightarrow \sum_{j>h}\frac{\gamma_jA_{j,t+1}}{R_{t+1}-\gamma_jA_{j,t+1}}\Big(\beta_j\frac{(1-\gamma_j)A_{j,t}}{R_t-\gamma_jA_{j,t}}\Big)\cdots \Big(\beta_j\frac{(1-\gamma_j)A_{j,1}}{R_1-\gamma_jA_{j,1}}\Big) s_{j,0}=\sum_{i<h}\beta_i^t s_{i,0}\\
\text{Date 1: } \sum_{j>h}\frac{\gamma_jA_{j,1}}{R_{1}-\gamma_jA_{j,1}}s_{j,0}=\sum_{i<h} s_{i,0}\\
\text{Date 2: } \sum_{j>h}\frac{\gamma_jA_{j,2}}{R_{2}-\gamma_jA_{j,2}}\Big(\beta_j\frac{(1-\gamma_j)A_{j,1}}{R_1-\gamma_jA_{j,1}}\Big) s_{j,0}=\sum_{i<h}\beta_i s_{i,0}\\
\text{Date 3: } \sum_{j>h}\frac{\gamma_jA_{j,3}}{R_{3}-\gamma_jA_{j,3}}\Big(\beta_j\frac{(1-\gamma_j)A_{j,2}}{R_2-\gamma_jA_{j,2}}\Big)\Big(\beta_j\frac{(1-\gamma_j)A_{j,1}}{R_1-\gamma_jA_{j,1}}\Big) s_{j,0}=\sum_{i<h}\beta_i^2 s_{i,0}
\end{align*}
Let us focus on the equilibrium with $R_t\in (A_{m,t-1},A_{m,t})$, i.e., $j=m$. We have that

\begin{align*}
\frac{\gamma_mA_{m,t+1}}{R_{t+1}-\gamma_mA_{m,t+1}}\Big(\beta_m\frac{(1-\gamma_m)A_{m,t}}{R_t-\gamma_mA_{m,t}}\Big)\cdots \Big(\beta_m\frac{(1-\gamma_m)A_{m,1}}{R_1-\gamma_mA_{m,1}}\Big) s_{m,0}=\sum_{i<m}\beta_i^t s_{i,0}\\
\text{$R_1$: } \frac{\gamma_mA_{m,1}}{R_{1}-\gamma_mA_{m,1}}s_{m,0}=\sum_{i<m} s_{i,0}\\
\text{$R_2$: } \frac{\gamma_mA_{m,2}}{R_{2}-\gamma_mA_{m,2}}\Big(\beta_m\frac{(1-\gamma_m)A_{m,1}}{R_1-\gamma_mA_{m,1}}\Big) s_{m,0}=\sum_{i<m}\beta_i s_{i,0}.
\end{align*}
Therefore, we can find the interest rate. First, the interest rate $R_1=\gamma_mA_{m,1}\frac{S_0}{\sum_{i<m}s_{i,0}}$. For date $t\geq 1$, we have
\begin{align*}
\frac{\sum_{i\not=m}\beta_i^t s_{i,0}}{\sum_{i\not=m}\beta_i^{t-1} s_{i,0}}&=\frac{\frac{\gamma_mA_{m,t+1}}{R_{t+1}-\gamma_mA_{m,t+1}}\Big(\beta_m\frac{(1-\gamma_m)A_{m,t}}{R_t-\gamma_mA_{m,t}}\Big)\cdots \Big(\beta_m\frac{(1-\gamma_m)A_{m,1}}{R_1-\gamma_mA_{m,1}}\Big) s_{m,0}}{\frac{\gamma_mA_{m,t}}{R_{t}-\gamma_mA_{m,t}}\Big(\beta_m\frac{(1-\gamma_m)A_{m,t-1}}{R_{t-1}-\gamma_mA_{m,t-1}}\Big)\cdots \Big(\beta_m\frac{(1-\gamma_m)A_{m,1}}{R_1-\gamma_mA_{m,1}}\Big) s_{m,0}}\\
&=\frac{\frac{\gamma_mA_{m,t+1}}{R_{t+1}-\gamma_mA_{m,t+1}}\Big(\beta_m\frac{(1-\gamma_m)A_{m,t}}{R_t-\gamma_mA_{m,t}}\Big)}{\frac{\gamma_mA_{m,t}}{R_{t}-\gamma_mA_{m,t}}}=\frac{\gamma_mA_{m,t+1}}{R_{t+1}-\gamma_mA_{m,t+1}}\frac{\beta_m(1-\gamma_m)}{\gamma_m}\\
&=\frac{\beta_m(1-\gamma_m)A_{m,t+1}}{R_{t+1}-\gamma_mA_{m,t+1}}.
\end{align*}

To sum up, we obtain that:
\begin{align}
R_1&=\gamma_mA_{m,1}\frac{S_0}{\sum_{i<m}s_{i,0}}\\
R_{t+1}&=A_{m,t+1}\Big(\gamma_m+(1-\gamma_m)\beta_m\frac{\sum_{i\not=m}\beta_i^{t-1} s_{i,0}}{\sum_{i\not=m}\beta_i^t s_{i,0}}\Big), \forall t\geq 1.
\end{align}

We need to check that $A_{m-1,t}<R_{t}<A_{m,t}$, $\forall t$. It means that
\begin{align}
\frac{A_{m-1,1}}{A_{m,1}}<\gamma_m\frac{S_0}{\sum_{i<m}s_{i,0}}<1\\
\frac{A_{m-1,t+1}}{A_{m,t+1}}<\gamma_m+(1-\gamma_m)\beta_m\frac{\sum_{i\not=m}\beta_i^{t-1} s_{i,0}}{\sum_{i\not=m}\beta_i^t s_{i,0}}<1, \forall t\geq 1\\
\beta_m\frac{\sum_{i\not=m}\beta_i^{t-1} s_{i,0}}{\sum_{i\not=m}\beta_i^t s_{i,0}}<1, \forall t\geq 1.
\end{align}
So, we can see that $R_t$ is increasing in $\gamma_m$ for any $t$.

Let $\frac{A_{m-1,t+1}}{A_{m,t+1}}$ converges  to $\frac{A_{m-1}}{A_{m}}$, $\forall t$ and $\beta_{i_0}=\max_{i<m}\beta_i$. We must have
\begin{align}
\frac{A_{m-1}}{A_{m}}\leq \gamma_m+(1-\gamma_m)\frac{\beta_m}{\beta_{i_0}}\leq 1.
\end{align}

We now find the capital at date $t$. We have
\begin{align*}
k_{m,0}&=\sum_{i}s_{i,0}\\
k_{m,1}&=\sum_{i<m}s_{i,1}+s_{m,1}=\sum_{i<m}\beta_iR_1s_{i,0}+\beta_m\frac{(1-\gamma_m)A_{m,1}R_1}{R_1-\gamma_mA_{m,1}}s_{m,0}\\
&=R_1\Big(\sum_{i<m}\beta_is_{i,0}+\beta_m\frac{(1-\gamma_m)A_{m,1}}{R_1-\gamma_mA_{m,1}}s_{m,0}\Big)=R_1\Big(\sum_{i<m}\beta_is_{i,0}+\beta_m\frac{1-\gamma_m}{\gamma_m}
\sum_{i<m} s_{i,0}\Big)\\
&=\gamma_mA_{m,1}\frac{S_0}{\sum_{i<m}s_{i,0}}\Big(\sum_{i<m}\beta_is_{i,0}+\beta_m\frac{1-\gamma_m}{\gamma_m}
\sum_{i<m} s_{i,0}\Big)\\
&=A_{m,1}S_0\Big(\gamma_m\frac{\sum_{i<m}\beta_is_{i,0}}{\sum_{i<m}s_{i,0}}+\beta_m(1-\gamma_m)\Big).
\end{align*}

Since $\beta_m\frac{\sum_{i\not=m}\beta_i^{t-1} s_{i,0}}{\sum_{i\not=m}\beta_i^t s_{i,0}}<1, \forall t\geq 1$, we see that $k_{m,1}$ is increasing in $\gamma_m$.

For any date $t\geq 2$, by using $\frac{\gamma_mA_{m,t+1}}{R_{t+1}-\gamma_mA_{m,t+1}}\Big(\beta_m\frac{(1-\gamma_m)A_{m,t}}{R_t-\gamma_mA_{m,t}}\Big)\cdots \Big(\beta_m\frac{(1-\gamma_m)A_{m,1}}{R_1-\gamma_mA_{m,1}}\Big) s_{m,0}=\sum_{i<m}\beta_i^t s_{i,0}$, we get that
\begin{align*}
 k_{m,t}&=\frac{R_{t+1}}{R_{t+1}-\gamma_mA_{m,t+1}}\Big(\beta_m\frac{(1-\gamma_m)A_{m,t}R_t}{R_t-\gamma_mA_{m,t}}\Big)\cdots \Big(\beta_m\frac{(1-\gamma_m)A_{m,1}R_1}{R_1-\gamma_mA_{m,1}}\Big) s_{m,0}\\
&=R_{t+1}\ldots R_1\frac{1}{\gamma_mA_{m,t+1}}\frac{\gamma_mA_{m,t+1}}{R_{t+1}-\gamma_mA_{m,t+1}}\Big(\beta_m\frac{(1-\gamma_m)A_{m,t}R_t}{R_t-\gamma_mA_{m,t}}\Big)\cdots \Big(\beta_m\frac{(1-\gamma_m)A_{m,1}R_1}{R_1-\gamma_mA_{m,1}}\Big) s_{m,0}\\
&=R_{t+1}\ldots R_1\frac{1}{\gamma_mA_{m,t+1}}\sum_{i<m}\beta_i^t s_{i,0}\\
&=R_{t+1}\ldots R_2\frac{1}{A_{m,t+1}}\Big(\sum_{i<m}\beta_i^t s_{i,0}\Big) A_{m,1}\frac{S_0}{\sum_{i<m}s_{i,0
}}\\
&=R_{t}\ldots R_2 A_{m,1}S_0\frac{\sum_{i<m}\beta_i^t s_{i,0}}{\sum_{i<m}s_{i,0
}}\Big(\gamma_m+(1-\gamma_m)\beta_m\frac{\sum_{i\not=m}\beta_i^{t-1} s_{i,0}}{\sum_{i\not=m}\beta_i^t s_{i,0}}\Big), \forall t\geq 1,
\end{align*}
where the last equality follows $R_{t+1}=A_{m,t+1}\Big(\gamma_m+(1-\gamma_m)\beta_m\frac{\sum_{i\not=m}\beta_i^{t-1} s_{i,0}}{\sum_{i\not=m}\beta_i^t s_{i,0}}\Big)$.
We see that this is increasing in the credit limit $\gamma_m$, and so is the aggregate output. This is also increasing in agents' productivity.

\end{proof}

}


{\small \section{Characterization of equilibrium in a two-period model}\label{34-proof}
 \setcounter{equation}{0} 
\renewcommand{\theequation}{A.\arabic{equation}}

  \subsection{Proof of Theorem 1}
  \label{append-linear}
We have the following result which characterizes the optimal solution of agents.   
   \begin{lemma}[individual choice - linear production function] \label{lem1}
Assume that $F_i(K)=A_iK$. Let $R>0$ be given. The solution for agent i's maximization problem is described as follows.
\begin{enumerate}
\item If $R\leq \gamma_iA_i$, then there is no solution ($k_i=\infty$).

\item If $A_i > R >\gamma_iA_i$, then agent i borrows from the financial market and the borrowing constraint is binding. We have
$k_i = \frac{R}{R-\gamma_iA_i}S_i,  a_i = \frac{\gamma_iA_i}{R- \gamma_iA_i}S_i, 
\pi_i=A_ik_i-Rb_i=\frac{R(1-\gamma_i)}{R-\gamma_iA_i}A_iS_i$.

\item If $A_i = R$, then the solutions for the agent's problem include all pairs $(k_i, b_i)$ such that $ -S_i \leq b_i \leq \frac{\gamma_i}{1-\gamma_i}S_i$ and $k_i= b_i +S_i$.
\item If $A_i < R$, then agent i does not produce goods and invest all her initial wealth in the  financial market: $k_i = 0, b_i = - S_i. $
\end{enumerate}
\end{lemma}

According to Definition of $\mathbb{D}_{n}, \mathbb{B}_{n}$,
\begin{align}\mathbb{D}_{n}&\equiv \sum_{i=n}^{m}\frac{A_{n}S_i}{A_{n}-\gamma_iA_i}, \text{ }\forall n\geq 1, \quad
\mathbb{B}_n \equiv \sum_{i=n+1}^{m}\frac{A_nS_i}{A_n-\gamma_iA_i}, \text{ }\forall n\geq 1,
\end{align}
 we observe that 
\begin{align}\label{bd}
\frac{S_m}{1-\gamma_m}=\mathbb{D}_m<\cdots <\mathbb{D}_{n+1}<\mathbb{B}_{n}<\mathbb{D}_{n}<\mathbb{B}_{n-1}<\cdots < \mathbb{B}_1=\sum_{i=2}^{m}\frac{A_1S_i}{A_1-\gamma_iA_i}.
\end{align}

Theorem \ref{cate-r} is a direct consequence of the existence of equilibrium and  Lemmas  \ref{lemma2}-\ref{lem-rn} below. First, the following result is a direct consequence of 
Lemma \ref{lem1}.
\begin{lemma}\label{lemma2} Assume that $A_1<A_2<\cdots<A_m$.
 If $\max_{i} (\gamma_iA_i)\geq A_n$ and there exists an equilibrium, then  $R>A_n$.
\end{lemma}

By comparing $\mathbb{B}_n,\mathbb{D}_n$ with the aggregate capital supply $S\equiv \sum_{i=1}^mS_i$, we obtain the following result.

\begin{lemma} \label{intermediate}Assume that $A_1<A_2<\cdots<A_m$. Denote $S\equiv \sum_{i=1}^mS_i$ the aggregate capital. Consider an equilibrium.
\begin{enumerate}
\item If $A_n>\max_i(\gamma_iA_i)$ and  $R>A_n$, then $\mathbb{B}_n>S$.  
Consequently, if  $A_n>\max_i(\gamma_iA_i) $ and $
\mathbb{B}_n\leq S$, then   $R\leq A_n$.

\item If  $A_{n}>\max_i(\gamma_iA_i)$ and $R<A_{n}$, then $S>\mathbb{D}_n$.  
Consequently, if   $A_{n}>\max_i(\gamma_iA_i)$ and $S\leq \mathbb{D}_n$, then $R\geq A_{n}$. 
\end{enumerate}
\end{lemma}

\begin{proof}

\begin{enumerate}
\item Since $R>A_i$ for any $i=1,\ldots,n$, Lemma \ref{lem1}  implies that $k_i=0, a_i=-S_i$ $\forall i=1,\ldots,n$. Hence, we have, by using market clearing condition, 
\begin{align}\sum_{i=1}^nS_i=-\sum_{i=1}^na_i=\sum_{i=n+1}^ma_i\leq \sum_{i=n+1}^m\dfrac{\gamma_iA_i}{R- \gamma_iA_i}S_i<\sum_{i=n+1}^m\dfrac{\gamma_iA_i}{A_n- \gamma_iA_i}S_i
\end{align} 
where the first inequality follows $b_i\leq \frac{\gamma_i A_iS_i}{R-\gamma_iA_i}$ while the last inequality follows $R>A_n>\max_i(\gamma_iA_i)$ and the fact that the function $Func(R)\equiv \sum_{i=n+1}^m\frac{\gamma_iA_i}{R- \gamma_iA_i}S_i$ is decreasing in $(\max_{i}(\gamma_iA_i),+\infty)$. Notice that this function is not decreasing in the interval $(0,\infty)$.

\item Since $R<A_{n}$, again Lemma \ref{lem1} implies that $
k_i = \frac{R}{R-\gamma_iA_i}S_i$ and  $a_i = \frac{\gamma_iA_i}{R- \gamma_iA_i}S_i  \forall i\geq n$.  We have
\begin{align}
\sum_{i=n}^m\dfrac{A_{n}S_i}{A_{n}- \gamma_iA_i}S< \sum_{i=n}^m\dfrac{RS_i}{R- \gamma_iA_i}=\sum_{i=n}^mk_i\leq \sum_{i=1}^mS_i=S
\end{align}
where the first inequality follows $A_{n}>R>\max_i(\gamma_iA_i)$.
 \end{enumerate}
\end{proof}

\begin{lemma}\label{lem-an} $R=A_n$ if and only if $A_n>\max_i(\gamma_iA_i)$ and $\mathbb{B}_n\leq S\leq \mathbb{D}_n$.\footnote{We need condition $A_n>M\equiv \max_i(\gamma_iA_i)$ because that $R>\max_i(\gamma_iA_i)$. Condition  $\sum_{i=n+1}^m\frac{A_nS_i}{A_n-\gamma_iA_i}\leq S$ ensures that $R\leq A_n$ while condition $S\leq \sum_{i=n}^m\frac{A_nS_i}{A_n-\gamma_iA_i}$ ensures that $R\geq A_n$.}
\end{lemma}

\begin{proof}
If $R=A_n$, we have $k_i=0 \text{ }\forall i\leq n-1$ and  $k_i=\frac{RS_i}{R-\gamma_iA_i}\text{ } \forall i \geq n+1$. This implies that $A_n=R>\max_{i}(\gamma_iA_i)$. 
Since $0\leq k_n\leq \frac{RS_i}{R-f_nA_n}$, we have 
\begin{align}\sum_{i=n+1}^m\frac{RS_i}{R-\gamma_iA_i}\leq \sum_ik_i=\sum_{i=n}^mk_i\leq \sum_{i=n}^m\frac{RS_i}{R-\gamma_iA_i}=\sum_{i=n}\frac{A_nS_i}{A_n-\gamma_iA_i}
\end{align} 

By converse, suppose that $A_n>\max_i(\gamma_iA_i)$ and $\sum_{i=n+1}^m\frac{A_nS_i}{A_n-\gamma_iA_i}\leq S\leq \sum_{i=n}^m\frac{A_nS_i}{A_n-\gamma_iA_i}$. Applying points 1 and 2 of Lemma \ref{intermediate}, we have $R\geq A_n$ and $R\leq A_n$. Hence $R=A_n$. 
\end{proof}

By combining Lemma \ref{intermediate} and the fact that  $R>\max_i(\gamma_iA_i)$, we obtain the following result.

\begin{lemma} \label{intermediate1}Assume that $A_1<A_2<\cdots<A_m$. Consider an equilibrium. If 
 $R\in (A_n,A_{n+1})$, then $A_{n+1}>\max_i(\gamma_iA_i)$ and $R=R^L_n$ (hence $R^L_n\in(A_n,A_{n+1})$).
\end{lemma}

We now identify the necessary and sufficient conditions under which $R=R^L_n$.

\begin{lemma}\label{lem-rn}$R=R^L_n\not=A_n$ if and only if one of the following conditions is satisfied:
\begin{enumerate}
\item $\max_i(\gamma_iA_i)< A_n< r^L_n< A_{n+1}$, or equivalently $\max_i(\gamma_iA_i)< A_n$ and $\mathbb{D}_{n+1}<S<\mathbb{B}_n$
\item $A_n\leq \max_i(\gamma_iA_i)<R^L_n<A_{n+1}$, or equivalently $A_n\leq \max_i(\gamma_iA_i)<R^L_n$ and $\mathbb{D}_{n+1}<S$.

\end{enumerate} 
In any case, we have that $R^L_n\in [A_n,A_{n+1})$.
\end{lemma}

\begin{proof}
{\bf Part 1}. Assume that $R=R^L_n\not=A_n$. By definition  of $R$ and $R^L_n$, we have $\sum_{i=n+1}^m\frac{RS_i}{R- \gamma_iA_i}=S$, and $R^L_n>\max_{i}(\gamma_iA_i)$. We will prove that $R=R^L_n\in (A_n,A_{n+1})$.

If $R\leq A_n$, then $R<A_{n+1}$, and hence $k_i=\frac{RS_i}{R- \gamma_iA_i}$ $\forall i\geq n+1$. Since  $\sum_{i=n+1}^m\frac{RS_i}{R- \gamma_iA_i}=S=\sum_{i}k_i$. We have $k_i=0$ $\forall i\leq n$, and hence $k_n=0$. This implies that $R\geq A_n$. Therefore, we have $R=A_n$, a contradiction. Thus, we have $R>A_n$.

If $R\geq A_{n+1}$, we have $k_i=0$ $\forall i\leq n$. Hence $S=\sum_ik_i\leq \sum_{i=n+1}^m\frac{RS_i}{R- \gamma_iA_i}.$ Since $\sum_{i=n+1}^m\frac{RS_i}{R- \gamma_iA_i}=S$, we have $k_i=\frac{RS_i}{R- \gamma_iA_i}$ $\forall i\geq n+1$. Hence $A_{n+1}\geq R$. So, $R=A_{n+1}$. We have just proved that $R\leq A_{n+1}$. By definition of $R$, we get that $A_{n+1}>\max_{i}(\gamma_iA_i)$. If $R^L_n=A_{n+1}$, then applying Lemma \ref{lem-an}, we have  $\sum_{i=n+2}^m\frac{A_{n+1}S_i}{A_{n+1}-\gamma_iA_i}=\mathbb{B}_{n+1}\leq S$. However, by definition of $R^L_n$, we have $\sum_{i=n+1}^m\frac{A_{n+1}S_i}{A_{n+1}- \gamma_iA_i}=S$, contradiction. Therefore, we obtain $R^L_n<A_{n+1}$. 

We have just proved that $R^L_n\in (A_n,A_{n+1})$.  
Applying point 2 of Lemma \ref{intermediate}, we have $S>\mathbb{D}_{n+1}$.  There are two cases:
\begin{enumerate}
\item $\max_{i}(\gamma_iA_i)\geq A_n$. In this case, we have $A_n\leq \max_{i}(\gamma_iA_i)<R^L_n<A_{n+1}$.

\item $\max_{i}(\gamma_iA_i)< A_n$. We get $\max_{i}(\gamma_iA_i)< A_n< R^L_n< A_{n+1}$. Notice that, in this case, $R^L_n\in (A_n,A_{n+1})$ is equivalent to $\mathbb{D}_{n+1}<S< \mathbb{B}_n$.
\end{enumerate}
{\bf Part 2}. Conversely, assume that (i) $A_n\leq \max_{i}(\gamma_iA_i)<R^L_n<A_{n+1}$ or (ii) $\max_{i}(\gamma_iA_i)< A_n< R^L_n< A_{n+1}$.

\begin{enumerate}
\item If $A_n\leq \max_{i}(\gamma_iA_i)<R^L_n<A_{n+1}$. Condition $A_n\leq \max_{i}(\gamma_iA_i)$ implies that $R>A_n$. Then $k_i=0$ $\forall i\leq n$, and hence $S=\sum_{i=n+1}^mk_i\leq \sum_{i=n+1}^m\frac{RS_i}{R- \gamma_iA_i}$

 By definition $R^L_n$, we have 
$S=\sum_{i=n+1}^m\frac{R^L_nS_i}{R^L_n- \gamma_iA_i}$. Since the function $f(X)\equiv \sum_{i=n+1}^m\frac{XS_i}{X- \gamma_iA_i}$ is decreasing in the interval $(\max_{i\geq n+1}(\gamma_iA_i),\infty)$ and $R,R^L_n>\max_i(\gamma_iA_i)$, we have $R\leq R^L_n$. This implies that $R\in (A_n,A_{n+1})$. Therefore, Lemma \ref{intermediate1} implies that $R=R^L_n$.

\item If $\max_{i}(\gamma_iA_i)< A_n$ and $\mathbb{D}_{n+1}<S< \mathbb{B}_n$. We have $S<\mathbb{D}_n$ because $\mathbb{D}_n>\mathbb{B}_n$. According to point 2 of Lemma \ref{intermediate}, we have $R\geq A_n$. 

Condition $S>\mathbb{D}_{n+1}$ implies that $S>\mathbb{B}_{n+1}$ because $\mathbb{D}_{n+1}>\mathbb{B}_{n+1}$. According to point 1 of Lemma \ref{intermediate}, we have $R\leq A_{n+1}$. 

If $R=A_{n+1}$, then Lemma \ref{lem-an} implies that $S\leq \mathbb{D}_{n+1}$. This is a contradiction because $S>\mathbb{D}_{d+1}$.  

If $R=A_n$, Lemma \ref{lem-an} implies that $S\in [\mathbb{B}_n,\mathbb{D}_n]$. However, $S\leq \mathbb{B}_n$. Thus, we have $S=\mathbb{B}_n=\sum_{i=n+1}^m\frac{A_nS_i}{A_n-\gamma_iA_i}$. Since $A_n>\max_i(\gamma_iA_i)$, then $A_n=R^L_n$, a contradiction.

Summing up, we have $R\in (A_n,A_{n+1})$. By applying point 3 of Lemma \ref{intermediate1}, we have $R=R^L_n$. 
\end{enumerate}
\end{proof}


\begin{remark} {\normalfont We can check that the regimes in Theorem \ref{cate-r} are not overlap, and the union of these regimes is equal to the set of economies satisfying $A_1<\cdots<A_m$, or, formally,
\begin{subequations}
\begin{align}
\label{setE1}&\mathbf{E}=
\cup_{i=1}^{m}\mathcal{A}_i\cup \cup_{i=1}^{m-1}\mathcal{R}_i\\
\label{setE2}&\mathcal{X}\cap\mathcal{Y}=\emptyset \text{ } \forall X,Y\in \{\mathcal{A}_1,\ldots,\mathcal{A}_m,\mathcal{R}_1,\ldots,\mathcal{R}_{m-1}\} \text{ and } X\not=Y.
\end{align}
\end{subequations}


Denote $M\equiv \max_i(\gamma_iA_i)$. By definition, we see that:
\begin{enumerate}
\item The economy $\mathcal{E} \equiv (F_i, \gamma_i,S_i)_{i=1,\ldots,m} \in \mathcal{A}_1$ if and only if $A_1>\max_i(\gamma_iA_i)$ and $S>\mathbb{B}_1$. 

\item $\mathcal{E}\in \mathcal{A}_m$ if and only if $S\leq \mathbb{D}_m$.

\item $\mathcal{E}\in \mathcal{A}_n$ with $n\in \{2,\ldots,m-1\}$ if and only if  $A_n>\max_i(\gamma_iA_i)$ and $\mathbb{B}_n\leq S\leq \mathbb{D}_n$.

\item $\mathcal{R}_n\equiv \mathcal{R}_{n,1}\cup \mathcal{R}_{n,2}$ with $n\in \{1,\ldots,m-1\}$ where
\begin{enumerate}
\item  $\mathcal{R}_{n,1}$ is the set of economies such that $A_{n}>\max_i(\gamma_iA_i)$ and $\mathbb{D}_{n+1}<S< \mathbb{B}_n$.

\item $\mathcal{R}_{n,2}$ is the set of economies such that $A_{n+1}>\max_i(\gamma_iA_i)\geq A_n$ and $\mathbb{D}_{n+1}<S$.
\end{enumerate}

\end{enumerate}

{\bf We now prove (\ref{setE1})} which implies the existence of equilibrium. It suffices to verify that $\mathbf{E}\subset
\cup_{i=1}^{m}\mathcal{A}_i\cup \cup_{i=1}^{m-1}\mathcal{R}_i$. Let us consider an economy $\mathcal{E}$. There are only two cases.
\begin{enumerate}
\item $\max_i(\gamma_iA_i)<A_1$. In this case, we have $\max_i(\gamma_iA_i)<A_n$ $\forall n$. Therefore, it is easy to see that $\mathcal{E}\in
\cup_{i=1}^{m}\mathcal{A}_i\cup \cup_{i=1}^{m-1}\mathcal{R}_{i,1}\subset 
\cup_{i=1}^{m}\mathcal{A}_i\cup \cup_{i=1}^{m-1}\mathcal{R}_{i}$.

\item There exists $n\in \{1,\ldots,m-1\}$ such that $A_{n+1}>\max_i(\gamma_iA_i)\geq A_n$. There are two sub-cases.
\begin{enumerate}
\item $S>\mathbb{D}_{n+1}$. In this case, $\mathcal{E}\in \mathcal{R}_{n+1,2}$.
\item $S\leq \mathbb{D}_{n+1}$. Recall that $M<A_{n+1}$. In this case, we will prove that $\mathcal{E}\in
\cup_{i=n+1}^{m}\mathcal{A}_i\cup \cup_{i=n+1}^{m-1}\mathcal{R}_i$. Indeed, since $S\leq \mathbb{D}_{n+1}$, there are $2(m-n)-1$ cases.
\begin{enumerate}
\item If there exists $i\in \{n+1,m-1\}$ such that $\mathbb{B}_{i}\leq S\leq \mathbb{D}_{i}$. Then $\mathcal{E}\in \mathcal{A}_{i}$ because $A_i\geq A_{n+1}>\max_i(\gamma_iA_i)$.

\item If there exists $i\in \{n+1,m-1\}$ such that  $\mathbb{D}_{i+1}\leq S\leq \mathbb{B}_{i}$. Then $\mathcal{E}\in \mathcal{R}_{i,1}$ because $A_i\geq A_{n+1}>\max_i(\gamma_iA_i)$.
\item Last, if $S\leq \mathbb{D}_m$, then $\mathcal{E}\in \mathcal{R}_{m}$.
\end{enumerate}

\end{enumerate}

\end{enumerate}

{\bf Proof of (\ref{setE2}).} Observe that the equilibrium interest rate is unique if (\ref{setE2}) holds. We have to prove that:
\begin{subequations}
\begin{align}
\mathcal{A}_n\cap\mathcal{A}_h&=\emptyset \text{ }\forall n\not=h\\
\mathcal{A}_n\cap\mathcal{R}_{h,1}&=\emptyset \text{ }\forall n,h\\ 
\mathcal{A}_n\cap\mathcal{R}_{h,2}&=\emptyset \text{ }\forall n,h\\
\mathcal{R}_n\cap\mathcal{R}_{h}&=\emptyset \text{ }\forall n\not=h.
\end{align}
\end{subequations}

Following (\ref{bd}), it is easy to see that the two first equalities hold.

We now prove that $\mathcal{A}_n\cap\mathcal{R}_{h,2}=\emptyset \text{ }\forall n,h$.  Suppose that there exists $\mathcal{E}\in \mathcal{A}_{n}\cap \mathcal{R}_{h,2}$. It means that (1) $A_n>\max_i(\gamma_iA_i)$ and $\mathbb{B}_n\leq S\leq \mathbb{D}_n$, and (ii) $A_{h+1}>\max_i(\gamma_iA_i)\geq A_h$ and $\mathbb{D}_{h+1}<S$. From these conditions we get $A_n>A_h$, and hence $n\geq h+1$.  Thus, we obtain
$S>\mathbb{D}_{h+1}\geq \mathbb{D}_n\geq S$, a contradiction. Therefore, we have $\mathcal{A}_n\cap\mathcal{R}_{h,2}=\emptyset \text{ }\forall n,h$.

Last, we prove $\mathcal{R}_n\cap\mathcal{R}_{h}=\emptyset$, or equivalently $\mathcal{R}_{n,i}\cap \mathcal{R}_{h,j}=\emptyset$ $\forall i,j\in \{1,2\}$, $\forall n\not=h$. Without loss of generality, we can assume that $n<h$. It is easy to see that $\mathcal{R}_{n,1}\cap \mathcal{R}_{h,1}=\emptyset$ and $\mathcal{R}_{n,2}\cap \mathcal{R}_{h,2}=\emptyset$. We now prove that $\mathcal{R}_{n,1}\cap \mathcal{R}_{h,2}=\emptyset$ and $\mathcal{R}_{n,2}\cap \mathcal{R}_{h,1}=\emptyset$. 
\begin{enumerate}
\item  Suppose that there exists $\mathcal{E}\in \mathcal{R}_{n,1}\cap \mathcal{R}_{h,2}$. It means that $A_{n}>\max_i(\gamma_iA_i)$; $\mathbb{D}_{n+1}<S<\mathbb{B}_n$; $A_{h+1}>\max_i(\gamma_iA_i) \geq A_h$;  $\mathbb{D}_{h+1}<S$. Since $h>n$, then $A_h>A_n>\max_i(\gamma_iA_i)$. This is a contradiction because $\max_i(\gamma_iA_i)\geq A_h$. So, we have $ \mathcal{R}_{n,1}\cap \mathcal{R}_{h,2}=\emptyset$.

\item Suppose that there exists $\mathcal{E}\in \mathcal{R}_{n,2}\cap \mathcal{R}_{h,1}$. It means that $A_{n+1}>\max_i(\gamma_iA_i)\geq A_n$; $\mathbb{D}_{n+1}<S$;$A_{h}>\max_i(\gamma_iA_i)$; $\mathbb{D}_{h+1}<S<\mathbb{B}_h$.

Since $h\geq n+1$, we have $\mathbb{B}_h\leq \mathbb{B}_{n+1}<\mathbb{D}_{n+1}<S<\mathbb{B}_h$, a contradiction.
\end{enumerate}

}
\end{remark}

\begin{remark}In Theorem \ref{cate-r}, we assume that $A_1<\cdots<A_m$. However, we can characterize the set of equilibria in the general case where some agents have the same productivity. Indeed, without lost of generality, we can (1) rank that $A_i\leq A_{i+1}$, $\forall i$, and assume that (2) the set $\{A_i: i\in \{1, \ldots, m\}\}$ has the cardinal $p$, $p\leq m$ and its distinct values are $(A_{i_t})_{t=1}^p$, where $A_1=A_{i_1}<A_{i_2}<\cdots<A_{i_{p}}=A_m$. We can decompose that
\begin{align*}
A_1,A_2,\ldots,A_m&=\underbrace{A_1,\ldots,A_{1}}_\text{$i_1$ times},\underbrace{A_{i_1+1},\ldots,A_{i_1+i_2}}_\text{$i_2$ times},\ldots, \underbrace{A_{i_1+\cdots+i_{p-1}},\ldots,A_{m}}_\text{$i_m$ times}
\end{align*}

Let us denote $
\mathbb{A}_t\equiv A_{i_t}, \quad \mathbb{S}_t\equiv \sum_{i: A_i=A_{i_t}}S_i.$ Then, we can use the same argument  in Theorem \ref{cate-r}  (but we replace $m$ by $p$, $A_i$ by $\mathbb{A}_i$, $S_i$ by $\mathbb{S}_i$) to determine the unique equilibrium interest rate. However, there may be multiple equilibrium allocations when one of the sets $\{i: A_i=A_{i_1}\}, \ldots, \{i: A_i=A_{i_p}\}$ has multiple elements. 
\end{remark}


\subsection{Proof of Theorem 2 (economy with strictly concave technologies)}
\label{general1-proof}

\subsubsection*{Individual optimal choice}

Before present the properties of individual optimal choice, we introduce some notations:
\begin{definition}\label{knb}  Given $R,\gamma_i,A_i,S_i$, denote $k^n_i=k^n_i(R/A_i)$ the unique solution to the equation  $A_if_i'(k)=R$ and $k^b_i=k^b_i(\frac{R}{\gamma_iA_i},S_i)$ the unique solution to $R(k-S_i)=\gamma_iA_if_i(k)$.
  \end{definition}
  
 $k^b_i$ ($k^n_i$, respectively) represents the capital level of agent $i$ when her borrowing constraint is binding (not binding, respectively). Under assumptions in Lemma \ref{assum-concave}, we can verify that: (1) $k^n_i$ is strictly decreasing in $R/A_i$. Moreover, $\lim_{R/A_1\to 0}k^n_i=+\infty$, and  $\lim_{R/A_1\to \infty}k^n_i=0$.  (2)
 $k^b_i$ is strictly increasing in $S_i$ but strictly decreasing in $\frac{R}{\gamma_iA_i}$. Moreover, $\lim_{R/A_i\to 0}k^b_i=+\infty$, and  $\lim_{R/A_i\to \infty}k^b_i=S_i$. 
 
The following result characterizes the solution of the problem $(P_i)$.
 \begin{lemma}[individual choice - strictly concave production function]\label{indi-general}Under Assumption \ref{assum-concave}, there exists a unique solution to the problem $(P_i)$. The optimal capital $k_i$ is increasing in TFP $A_i$, credit limit $\gamma_i$ but decreasing in the interest rate $R$. 
\begin{enumerate}
\item If $R\frac{k^n_i(R/A_i)-S_i}{A_if_i(k^n_i(R/A_i))} \geq \gamma_i$, then credit constraint is binding and the capital level is $k_i=k^b_i$. Moreover, $ k_i= k^b_i\leq k^n_i$.

\item If $R\frac{k^n_i(R/A_i)-S_i}{A_if_i(k^n_i(R/A_i))} < \gamma_i$, then credit constraint is not binding and $k=k^n_i$. In this case, we have $k_i=k^n_i< k^b_i$.
\end{enumerate}

\end{lemma}

\begin{proof}[{\bf Proof of Lemma \ref{indi-general}}]

 


 Since $F'_i(0) = \infty$, we have $k_i > 0$  at optimum. The Lagrange function is 
\begin{equation*}
L= F_i(k_i) - Rb_i + \lambda_i (S_i+b_i-k_i) +\mu_i(\gamma_i F_i(k_i)-Rb_i)
\end{equation*}
It is easy to see that $(k_i,b_i)$ is a solution if and only if there exists $(\lambda_i,\mu_i)$ such that
\begin{align*}
[k]&: (1 +\mu_i \gamma_i)F_i'(k_i) =\lambda_i\\
[a]&: (1+\mu_i)R = \lambda_i, \quad \mu_i \geq 0, \text{ and } \mu_i(\gamma_i  F_i(k_i)-R_ib_i)=0.
\end{align*}
These equations imply that:
\begin{equation}
A_if_i(k_i)=F_i'(k_i) = R\dfrac{1+\mu_i}{1+\gamma_i \mu_i} \geq R. \label{eq4}
\end{equation}

 Since $F_i'$ is decreasing, we have $k_i\leq k^n_i(R/A_i)$.

We consider two cases.

  \textbf{Case 1}: 
The credit constraint is binding: $\gamma F_i(k_i)=Rb_i$. In this case, $(k_i, b_i)$ is the solutions of the following equations:
\begin{align}
b_i = k_i - S_i\\
\label{k-general}\gamma F_i(k_i)=R(k_i-S_i),\quad i.e., \quad \frac{\gamma_i }{R}=\frac{k_i}{F_i(k_i)}-\frac{S_i}{F_i(k_i)}.
\end{align}
Consider the function $k/F_i(k)$. Its derivative equals $\frac{F_i(k)-kF_i'(k)}{(F_i(k))^2}$ which is non-negative because $F$ is concave. So, the function $G_i(k)\equiv \frac{k-S_i}{F_i(k)}$ is strictly increasing in $k$. Moreover, $\lim_{k\to 0}G_i(k)<\gamma_i /R$ and $G_i(\infty)>\gamma_i /R$ (because $F_i'(\infty)<1$). Therefore, there exists a unique solution $k_i$ of equation (\ref{k-general}), and this is positive. It is actually $k^b_i$.

We now investigate condition $k_i\leq k^n_i$. Since $G_i(k_i)=\gamma_i /R$, condition $k_i\leq k^n_i$ is equivalent to $G_i(k^n_i)\geq \gamma_i /R$ (because $G_i(k_i)=\gamma_i /R$) or, equivalently, $R\frac{k^n_i(R/A_i)-S_i}{F_i(k^n_i(R/A_i))}\geq \gamma_i$. 

Conversely, assume that  $R\frac{k^n_i(R/A_i)-S_i}{F_i(k^n_i(R/A_i))}\geq \gamma_i$. We choose $k_i=k^b_i$. Then, by definition of $k^b_i$, we have $k_i\in (S_i,\infty)$. Therefore, we have
\begin{align*}
R>R(1-\frac{S_i}{k_i})=\gamma_i\frac{F_i(k_i)}{k_i}\geq \gamma_iF_i'(k_i)
\end{align*} 
where the last inequality follows the fact that $F_i$ is concave.  It means that $R>\gamma_iF_i'(k_i)$. So, we can define $\mu_i,\lambda_i$ by \begin{align*}
1-\frac{F_i'(k_i)}{R}=\mu_i\Big(\frac{F_i'(k_i)}{R}-\gamma_i\Big), \quad 
\lambda_i=R(1+\mu_i).
\end{align*}
Therefore, $(\lambda_i,\mu_i)$ and $(k_i,b_i)$ satisfy conditions $[k]$ and $[b]$ above. It means that $(k_i,b_i)$ is a solution.

\textbf{Case 2}: $\gamma_i F_i(k_i)>Rb_i$.  In this case, we have $\mu_i=0$, and hence $F_i'(k_i) =R$, i.e, $k_i=k^n_i$. It remains to check that this value of $k_i$ satisfies the condition: $\gamma_i F_i(k_i)>Rb_i=R(S_i-k_i)$, i.e., $\gamma_i /R>G_i(k^n_i)$. 

Observe that if $RG_i(k^n_i)<(\geq)\gamma_i $, then $G_i(k^n_i)<(\geq)\gamma_i /R=G_i(k^b_i)$, which implies that $k^n_i<(\geq)k^b_i$. 

The converse is easy. Notice that, in this case, agent borrows (i.e., $b_i > 0$) if and only if $k_i>S$ or equivalently $k^n_i>S$. This means that her wealth is low and/or interest rate is low and/or her productivity is high. 


\end{proof}


Under Assumptions \ref{assum-concave} and \ref{Hi}, the function $\frac{(k-S_i)f_i'(k)}{f_i(k)}$ is strictly increasing in $k$. Therefore, the function $G_i(x)\equiv \frac{(k^n_i(x)-S_i)x}{f_i(k^n_i(x))}$ is strictly decreasing in $x$. Moreover, we can check that $\lim_{x \rightarrow +\infty}G_i(x)=-\infty$, $\lim_{x \rightarrow 0}G_i(x)=\lim_{k \to \infty}\frac{kf_i'(k)}{f_i(k)}$.  By consequence, we obtain the following result.

\begin{lemma}\label{gammaalpha}Let Assumptions \ref{assum-concave} and \ref{Hi}  be satisfied. 
Then, if agent i's borrowing constraint is binding, we must have $\gamma_i\leq \lim_{x\rightarrow \infty}\frac{xF_i'(x)}{F_i(x)}$. By consequence, when $F_i(k)=A_ik^{\alpha_i}$ and $\gamma_i> \alpha_i$, then agent i's borrowing constraint is not binding.
\end{lemma}

The following result show the interaction between interest rate, credit limit $\gamma_i$ and borrowing constraint.

\begin{lemma} \label{Ridef}Let Assumptions \ref{assum-concave}, \ref{Hi} and \ref{gammaalpha-assum}  be satisfied. We can define $R_i$ the unique value satisfying 
\begin{align}\label{Ri1}
H_i(R_i)\equiv R_i \frac{k^n_i(R_i/A_i)-S_i}{A_if_i(k^n_i(R_i/A_i))}= \gamma_i.
\end{align}
Then, we have that:
\begin{enumerate}
\item Agent $i$'s borrowing constraint is binding if and only if $H_i(R)\geq \gamma_i$ which is equivalent to $R\leq R_i\equiv H_i^{-1}(\gamma_i)$. 

\item $R_i/A_i$ does not depend on $A_i$, and $\lim_{A_i\to \infty}R_i=\infty$, $\lim_{A_i\to 0}R_i=0$. $R_i$  is increasing in productivity $A_i$ but decreasing in $\gamma_i$ and in $S_i$. 

\item We also have $k^n_i(R_i/A_i)=k^b_i(R_i/A_i)$.

\end{enumerate}\end{lemma}

 The threshold $R_i$ is exogenous. It represents the subjective interest rate of agent below which agent borrows so that her(his) borrowing constraint is binding.  Point 2 of Lemma \ref{Ridef} indicates that the credit constraint of agent $i$ is more likely to bind if the interest rate, her initial wealth and credit limit are low, and/or her productivity is high.

 \begin{remark}\label{cdtr} Under Cobb-Douglas technology, i.e., $F_i(k)=A_ik^{\alpha}$, we can compute that $H_i(R)=\alpha\big(1-\big(\frac{R}{\alpha A_i S_i^{\alpha - 1}}\big)^{\frac{1}{1-\alpha}}\big)$ is decreasing in $R$  and  $H_i(0)=\alpha$. So, if $\alpha_i<\gamma_i$, then borrowing constraint is not binding, whatever the level of interest rate $R$. When $H_i(0)>\gamma_i$, i.e., $\alpha >\gamma_i$, we have $R_i = \alpha A_i S_i^{\alpha - 1}\big(1-\frac{\gamma_i}{\alpha}\big)^{1-\alpha}.$ 

\end{remark}  

\subsubsection*{Proof of Theorem \ref{general1}}  
\textcolor{blue}{To simplify notations, we write $k^n_i(R)$ and $k^b_i(R)$ instead of $k^n_i(\frac{R}{A_i})$ and  $k^b_i(\frac{R}{\gamma_iA_i}, S_i)$}. 
 We also introduce the so-called aggregate capital demand function:
\begin{align*}
B_n(R)\equiv
\begin{cases} 
\sum_{i=1}^nk^n_i(R)+\sum_{i=n+1}^mk^b_i(R)& \text{ if }  n\leq m-1\\
\sum_{i=1}^mk^n_i(R) &\text{ if } n=m.
\end{cases}
\end{align*}
  
\begin{lemma}\label{bn}
$B_n(R_n)>B_{n+1}(R_{n+1})=B_n(R_{n+1})$.
\end{lemma}
\begin{proof}
Indeed, since $R_n<R_{n+1}$, we notice that
\begin{align*}
&B_n(R_n)\equiv
\sum_{i=1}^nk^n_i(R_{n})+\sum_{i=n+1}^mk^b_i(R_n)\\
&>B_n(R_{n+1})=\sum_{i=1}^{n}k^n_i(R_{n+1})+\sum_{i=n+1}^mk^b_i(R_{n+1})=\sum_{i=1}^{n+1}k^n_i(R_{n+1})+\sum_{i=n+2}^mk^b_i(R_{n+1})
\end{align*}
where the last equality follows $k^b_{n+1}(R_{n+1})= k^n_{n+1}(R_{n+1})$. Therefore, $B_n(R_n)>B_{n+1}(R_{n+1})=B_n(R_{n+1})$  $\forall n$.
\end{proof}


We state an intermediate step whose proof is based on Lemma \ref{indi-general} and Lemma \ref{Ridef}.

\begin{lemma}\label{cate-pre} Let assumptions in Theorem \ref{general1} be satisfied. 
Consider an equilibrium $((k_i,b_i)_i,R)$ and an index $n\in \{1,\ldots,m-1\}$.

\begin{enumerate}
\item If $R>R_m$, Lemma \ref{indi-general}  implies that credit constraint of any agent is not binding. So, the equilibrium coincides to that of the economy without credit constraints. Therefore,  we have $R=R^*>R_m$.

\item  If $R>R_n$, then credit constraint of any agent $i\leq n$ is not binding. Hence $k_i=k^n_i(R) <k^n_i(R_n)$ $\forall i\leq n$. Condition $R>R_n$ also implies that $k^b_i(R)<k^b_i(R_n)$. Therefore, we have 
$ \sum_iS_i< B_n(R_n).$

\item If $R\leq R_{n+1}$, then credit constraint of any agent $i\geq n+1$ is binding, and hence $k_i=k^b_i(R)\geq k^b_i(R_{n+1})$ $\forall i\geq n+1$. Moreover, we have $k_i\geq k^n_i(R)\geq k^n_i(R_{n+1})$.
Therefore, we have $ \sum_iS_i\geq B_n(R_{n+1}).$

\end{enumerate}

\end{lemma}


We now prove Theorem \ref{general1}. Let us consider an equilibrium. Since there is at least one agent whose credit constraint is not binding, we have $R>R_1$.

{\bf Step 1}. Suppose that $R\in (R_n,R_{n+1}]$.  So, credit constraint of any agent $i\geq n+1$ is binding and that of any agent $i\leq n$ is not binding. Hence, the capital demand is
\begin{align}
\sum_ik_i=\sum_{i=1}^nk^n_i(R)+ \sum_{i=n+1}^mk^b_i(R).
\end{align}
Therefore, the equilibrium interest rate is determined by 
\begin{align}
 \sum_{i=1}^nk^n_i(R)+\sum_{i=n+1}^mk^b_i(R)=S\equiv \sum_iS_i.
\end{align}
The left-hand side is decreasing in $r$, and hence this equation has a unique solution.  

Since $R\in (R_n,R_{n+1}]$, we have
\begin{align*}
\sum_{i=1}^nk^n_i(R_n)+\sum_{i=n+1}^mk^b_i(R_{n}) >  \sum_iS_i\geq \sum_{i=1}^nk^n_i(R_{n+1})+\sum_{i=n+1}^mk^b_i(R_{n+1}).
\end{align*}
Conversely, if this condition holds, by using properties of functions $k^b_i,k^n_i$, we can easily prove that $R\in (R_{n},R_{n+1}]$. Indeed, if $R>R_{n+1}$, then point $2$ of Lemma \ref{cate-pre} implies that $S<B_{n+1}(R_{n+1})$. This contradicts to $S\geq B_{n+1}(R_{n+1})$.  If $R\leq R_n$, then point $3$ of Lemma \ref{cate-pre} implies that $S\geq B_{n-1}(R_{n})=B_{n}(R_{n})$. This contradicts to $S< B_{n}(R_{n})$. Therefore, we obtain  $R\in (R_{n},R_{n+1}]$.

{\bf Step 2}. 
We now suppose that $R^*>R_m$. We will prove that credit constraint of any agent is not binding. Suppose that the set $$\mathcal{B}=\{i\in\{1,\ldots,m\}: \text{ agent i's borrowing constraint is binding}\}$$ is not empty. Let $n: 1\leq n\leq m-1$ be the highest element in $\mathcal{B}$, i.e., credit constraint of any agent $i\geq n+1$ is binding while that of any agent $i\leq n$ is not. We have $R\in (R_{n},R_{n+1}]$. So,
 $k^b_i(R)\geq k^b_i(R_{n}+1)>k(R_m)$ and  
$k^n_i(R)\geq k^n_i(R_{n+1})\geq k^n_i(R_m)$.  Hence,  we get that
\begin{align}
 \sum_iS_i= \sum_{i=1}^nk^n_i(R)+ \sum_{i=n+1}^mk^b_i(R)\geq \sum_{i=1}^mk^n_i(R_m).
\end{align}
However, by definition of $R^*$, we have 
\begin{align}
 \sum_iS_i= \sum_{i=1}^mk^n_i(R^*)<\sum_{i=1}^mk^n_i(R_m). 
\end{align}
This is a contradiction. 

{\bf Step 3}. We now prove that $R_n\leq R^*$ $\forall n\leq m-1$. Indeed, in the regime $\mathcal{R}_n$, for any $i\geq n+1$, agent $i$'s credit constraint is binding. Hence, Lemma \ref{indi-general} follows that $k^b_i(R_n)\leq k^n_i(R_n)$ $\forall i\geq n+1$. Consequently, we get  that
\begin{align*}
\sum_{i=1}^mk^n_i(R^*)=S= \sum_{i=1}^nk^n_i(R_n)+\sum_{i=n+1}^mk^b_i(R_n)\leq \sum_{i=1}^mk^n_i(R_n)
\end{align*}
 which implies that $R^*\geq R_n$.

\section{The existence of intertemporal equilibrium}
\label{existencequilibrium}
 The proof is similar to the one in \cite{blvp18}. 
 Notice that we cannot directly use  a method of \cite*{bblvs15} or  \cite*{lvp16} because the financial asset in our model is a short-lived asset with zero supply. 


The idea is that we can bound the individual demand for the financial asset, and so can prove the existence of equilibrium by adapting the  method of \cite*{bblvs15} and  \cite*{lvp16}: (1) we prove the existence of equilibrium for each $T-$ truncated economy $\ee^T$; (2) we show that this sequence of equilibria converges for the product topology to an equilibrium of our economy $\ee$.

We write the maximization problem of agent $i$ with the presence of consumption prices:
\begin{subequations}
\begin{align}
 &\ma_{(c_i,k_i,b_i)} \sum_{t=0}^{\infty}\beta_i^tu_i(c_{i,t})\\
\label{bc-i}\text{subject to: }& p_t(c_{i,t}+k_{i,t})+R_tb_{i,t-1} \leq p_tF_{i,t}(k_{i,t-1}) + b_{i,t}\\
\label{cc-i}&  R_{t+1}b_{i,t}\leq \gamma_ip_{t+1}F_{i,t+1}(k_{i,t}),\\
&c_{i,t}\geq 0, \quad k_{i,t}\geq 0,
\end{align}
\end{subequations}
An intertemporal equilibrium is a list $((c_{i,t},k_{i,t},b_{i,t})_i,p_t,R_t)_{t\geq 0}$ satisfying conditions: (1) given $(p_t,R_t)_{t\geq 0}$, the allocation $(c_{i,t},k_{i,t},b_{i,t})$ is a solution of the above maximization problem, (2) markets clear: $\sum_ib_{i,t}=0$, $\sum_{i}(c_{i,t}+k_{i,t})=\sum_iF_{i,t}(k_{t-1})$  $\forall t\geq 0$, and (3) $p_t,R_t>0$ for any $t\geq 0$. 

We will prove the existence of economy in this setup. Then, from this equilibrium, as  $p_t>0$, we can change the variables ($p_t'=1,b_{i,t}'=\frac{b_{i,t}}{p_t}, R_t'=\frac{R_tp_{t-1}}{p_t})$ to get an equilibrium in the original economy.


\subsection{Existence of equilibrium for  truncated economies}\label{4.8.1}
For each $T\geq 1$, we define $T-$truncated economy ${\ee}^T$ as ${\ee}$ but there are no activities from period $T+1$ to the infinity, i.e.,  $c_{i,t+1}=k_{i,t}=b_{i,t}=0$ for every $i=1,\ldots, m$ and for any $t\geq T$.

We then define the bounded economy ${\ee}_b^T$ as ${\ee}^T$ but consumption level $(c_{i,t})_{t\leq T}$, physical capital $(k_{i,t})_{t\leq T}$, and asset holding $(b_{i,t})_{t\leq T}$ are respectively bounded in the following sets:
\begin{eqnarray*}\mathcal{C}&:=&[-B_c,B_c]^{T+1}, \quad
\mathcal{K}:= [-B_k,B_k]^{T+1}, \quad  \mathcal{B}:= [-B_b,B_b]^{T+1},
\end{eqnarray*}
where 
\begin{align}
B_c,B_k&>\ma_{t\leq T}B_{K,t}; \quad 
B_b=mB_c.
\end{align}

The economy ${\ee}_b^T$ depends on bounds $B_c,B_k, B_b$, so we write ${\ee}_b^T(B_c,B_k, B_b)$.

Let us define 
\begin{align} 
\label{bounds}\mathcal{X}_b&\equiv\mathcal{C} \times    \mathcal{K}   \times  \mathcal{B}, \quad \quad  \mathcal{X}\equiv(\mathcal{X}_b)^{m}\\
\mathcal{P} &\equiv\{z_0=(p_t,R_t)_{t\leq T}: R_{0}=0; \quad  p_{t}, R_{t} \geq 0;\quad p_{t}+R_{t}=1, \forall t\leq T  \}\\
\Phi& \equiv \mathcal{P} \times \mathcal{X}.
\end{align}

An element $z\in \Phi$ is in the form  $z=(z_i)_{i=0}^{m}$ where $z_0=(p_t,R_t)_{t\leq T}, z_i=(c_{i,t}, k_{i,t}, b_{i,t})_{t\leq  T}$ for each $i=1,\ldots,m$.

The following remark is to ensure that the asset volume $(b_{i,t})$ is bounded.
\begin{remark}
If $z\in \Phi$ is an equilibrium for the economy ${\ee}^T$, then, by using the fact that  $p_{t}+R_{t}=1$, we obtain that $b_{i,t}\leq B_c$ for any $i,t$. Indeed, this is true for $t=0$ because 
\begin{align}-b_{i,0}\leq p_{0}F_{i,0}(k_{i,-1})\leq p_{0}B_c=B_c
\end{align}
and then, for any $t\geq 0$, we have 
\begin{align}-b_{i,t}\leq p_{t} F_{i,t}(k_{i,t-1})-R_{t}b_{i,t-1}\leq (p_{t}+R_{t})B_c=B_c
\end{align}
Since $\summ_{i=1}^mb_{i,t}=0$, we get that $b_{i,t}\in [-B_b,B_b]$ for any $i$ and any $t$, where $B_b\equiv mB_c>(m-1)B_c$.
\end{remark}

\begin{proposition} Under Assumptions (1-4), there exists an equilibrium $(p, R, (c_i,k_i,b_i)_{i=1}^m)$, with $p_{t}+R_{t}=1$, $\forall t$, for the economy  ${\ee}_b^T(B_c, B_k, B_b)$. 
This is actually an equilibrium for the economy ${\ee}^T(B_c, B_k, B_b)$
\end{proposition}

\begin{proof}We firstly define
\begin{align*}B_i^{T}(p,R)&:=\big\{ (c_{i,t}, k_{i,t}, b_{i,t})_{t\leq T}\in \mathcal{X}_b: \text{(a) } k_{i,t}=b_{i,t}=0 \text{ } \forall t\in D_T,\\
\text{(b) }& {p}_{0}(c_{i,0}+k_{i,0})\leq {p}_{0}F_{i,0}(k_{i,-1})+b_{i,0}\\
\text{(c) } & \text{for each } t: 1\leq t(t)\leq T:\\
&R_{t}b_{i,t-1} \leq \gamma_i p_{t} F_{i,t}(k_{i,t-1})\\
&{p}_{t}(c_{i,t}+k_{i,t}) + R_{t} b_{i,t-1}\leq p_{t}F_{i,t}(k_{i,t-1}) + b_{i,t}\big\}.
\end{align*}

We also define $C_i^{T}(p,R)$ as follows.
\begin{align*}C_i^{T}(p,R)&:=\big\{ (c_{i,t}, k_{i,t}, b_{i,t})_{t\leq T}\in \mathcal{X}: \text{(a) } k_{i,t}=b_{i,t}=0 \text{ } \forall t\in D_T,\\
\text{(b) }& {p}_{0}(c_{i,0}+k_{i,0})< {p}_{0}F_{i,0}(k_{i,-1})+b_{i,0}\\
\text{(c) } & \text{for each } t: 1\leq t(t)\leq T:\\
&R_{t}b_{i,t-1} < \gamma_i p_{t} F_{i,t}(k_{i,t-1})\\
&{p}_{t}(c_{i,t}+k_{i,t}) + R_{t} b_{i,t-1}< p_{t}F_{i,t}(k_{i,t-1}) + b_{i,t}\big\}.
\end{align*}

\begin{lemma}\label{4.s1}  $C_i^{T}(p,R)\not=\emptyset$ and $ \bar{C}_i^{T}(p,R)= B_i^{T}(p,R)$.
\end{lemma}

\begin{proof}Since $k_{i,-1}>0$ and $p_{0}=1$, we always have $ p_{0}F_{i,0}(k_{i,-1})>0$. Therefore, we can choose $(c_{i,0}, k_{i,0}, b_{i,0})\in \rr_+^2\times \rr_-$,  and then $(c_{i,t}, k_{i,t}, b_{i,t})\in \rr_+^2\times \rr_-$ such that this plan belongs to $C_i^{T}(p,q,r)$. Note that $p_{t}F_{i,t}(k_{i,t-1})-R_{t} b_{i,t-1}>0$ if $k_{i,t-1}>0, b_{i,t-1}<0$ and $(p_{t}, R_{t})\not=(0,0,0)$.
\end{proof}

\begin{lemma} $C_i^{T}(p,R)$ is lower semi-continuous correspondence on $\mathcal{P}$. $B^T_i(p,R)$ is continuous on $\mathcal{P}$ with compact convex values.
\end{lemma}
\begin{proof}It is clear since $C^T_i(p,R)$ is nonempty and has open graph. 
\end{proof}

We now define correspondences.  First, we define $\varphi_0$ (for additional agent $0$) $: \mathcal{X} \rightarrow 2^{\mathcal{P}}$ by
\begin{align*}
\varphi_0((z_i)_{i=1}^{m}) &:=\argmax_{(p,R)\in \mathcal{P}} \Big\{\summ_{t\leq T}\Big[p_{t}\summ_{i=1}^m\big(c_{i,t}+k_{i,t}-F_{i,t}(k_{i,t-1})\big)\Big]+ \summ_{t\leq T}\Big[R_{t}\summ_{i=1}^mb_{i,t-1}\Big]\Big\}.
\end{align*}

Second, for each $i=1,\ldots, m$, we define $\varphi_i: \mathcal{P}  \rightarrow 2^{\mathcal{X}}$
\begin{eqnarray*}\varphi_i((p,R)) &:=&\argmax_{(c_i, k_i,b_i)\in C^T_i(p,R)} \Big\{  \summ_{t=0}^{T} u_{i,t}(c_{i,t})\Big\}.
\end{eqnarray*}

\begin{lemma}The correspondence $\varphi_i$ is  upper semi-continuous and non-empty, convex, compact valued for each $i=0, 1, \ldots, m+1$.
\end{lemma}
\begin{proof} This is a direct consequence of the Maximum Theorem.
\end{proof}

According to the Kakutani Theorem, there exists $(\bar{p}, \bar{R}, (\bar{c}_i, \bar{k}_i, \bar{b}_i)_{i=1}^m)$ such that 
\begin{align}(\bar{p}, \bar{R}) &\in  \varphi_0( (\bar{c}_i, \bar{k}_i, \bar{b}_i)_{i=1}^m)\\
\label{4.optic}(\bar{c}_i,  \bar{k}_i, \bar{b}_i) &\in \varphi_i((\bar{p},\bar{R})).
\end{align}
Denote, for each $t\geq 0$,
\begin{align*}
\bar{X}_{t}&:=\summ_{i=1}^m(\bar{c}_{i,t}+k_{i,t}-F_{i,t}(\bar{k}_{i,t-1})), 
\quad
\bar{Z}_{t}:=\summ_{i=1}^m\bar{b}_{i,t}.
\end{align*}

Therefore, for every $(p,q,r)\in \mathcal{P}$, we have
\begin{align}
\summ_{t\leq T}(p_{t}-\bar{p}_t)\bar{X}_{t} +\summ_{t\leq T}(R_{t}-\bar{R}_t)\bar{Z}_{t-1}\leq 0
\end{align}
%

Consider date  t, by summing  budget constraints over $i$, we get that
\begin{eqnarray*}\bar{p}_{t}\bar{X}_{t}+\bar{R}_{t}\bar{Z}_{t-1}\leq \bar{Z}_{t}.
\end{eqnarray*}

By consequence, we have, for each $t\leq T$ and for every $(p_{t},R_{t})\geq 0$ with $p_{t}+R_{t}=1$,
\begin{eqnarray*}{p}_{t}\bar{X}_{t}+ R_{t}\bar{Z}_{t-1}\leq \bar{p}_{t}\bar{X}_{t}+ \bar{R}_{t}\bar{Z}_{t-1}\leq \bar{Z}_{t}.
\end{eqnarray*}

Since at date T, we have $\bar{Z}_{T}= 0$. So, $\bar{p}_{T}\bar{X}_{T}+\bar{R}_{T}\bar{Z}_{T-1}\leq 0$. Hence, ${p}_{T}\bar{X}_{T}+ R_{T}\bar{Z}_{T-1}\leq 0$ for any $(p_{T},R_{T})\geq 0$ with $p_{T}+R_{T}=1$. This implies that $\bar{X}_{T}, \bar{Z}_{T-1}\leq 0$. Repeating this argument, we obtain that $\bar{X}_{t},  \bar{Z}_{t}\leq 0   \text{ } \forall t\leq T$ which means that
\begin{align*}
\summ_{i=1}^m (\bar{c}_{i,t}+\bar{k}_{i,t})&\leq \summ_{i=1}^m F_{i,t}(\bar{k}_{i,t}) \\
\summ_{i=1}^m\bar{b}_{i,t}&\leq 0.
\end{align*}

\begin{lemma}\label{4.key}$\bar{p}_{t}>0, \bar{R}_{t}>0$ for any $t\leq T$.
\end{lemma}
\begin{proof} By definition of $B_{K,t}$, we see that  $\summ_{i=1}^m\bar{c}_{i,t}\leq  B_{K,t}<B_c,$ or any $t$. This allows us to prove that $\bar{p}_{t}>0$  for any $t$. Indeed, if $\bar{p}_{t}=0$ then $c_{i,t}=B_c>B_{K,t}$, a contradiction.

If $\bar{R}_{t}=0$, then $\bar{b}_{i,t-1}=-B_a$ for any $i$, which implies that $\summ_{i=1}^m\bar{b}_{i,t-1}< 0$, contradiction. Therefore, $\bar{R}_{t}>0$.

\end{proof}

\begin{lemma}$\bar{X}_{t}=\bar{Z}_{t}=0$.
\end{lemma}
\begin{proof} Using $\bar{p}_{t}\bar{X}_{t}+ \bar{R}_{t}\bar{Z}_{t-1}\leq 0$ and Lemma \ref{4.key}. 
\end{proof}

\begin{lemma} The optimality of $(\bar{c}_i,  \bar{k}_i, \bar{b}_i).$
\end{lemma}
\begin{proof} It is clear since $(\bar{c}_i,  \bar{k}_i, \bar{b}_i) \in \varphi_i((\bar{p},\bar{q}, \bar{r})).$
\end{proof}

We have just proved that $(\bar{p},  \bar{R}, (\bar{c}_i,  \bar{k}_i, \bar{b}_i)_{i=1}^m)$ is an equilibrium for the economy ${\ee}^T_b$. We now prove that this equilibrium for the economy ${\ee}^T(B_c, B_k, B_b)$. The market clearing conditions are obviously satisfied. It remains to prove the optimality of the allocation $\bar{z}_i=(\bar{c}_i,  \bar{k}_i, \bar{b}_i)$. Suppose the contra

Let $z_i\equiv ({c}_i,  {k}_i, {b}_i)$ be in the budget set of the $T-$truncated economy.  Since $(\bar{c}_i,  \bar{k}_i, \bar{b}_i)$ belongs the interior of $\mathcal{X}_b$, there exists $\lambda \in (0,1)$ such that $\lambda z^i +(1-\lambda) \bar{z}_i \in \mathcal{X}_b$. Of course, $\lambda z^i +(1-\lambda) \bar{z}_i$ is in the budget set of the economy ${\ee}^T_b$.  Denote $U_i(c)\equiv \sum_{t\leq T}u_{i,t}(c_{t})$.  

 We have $\lambda U^i (c^i) +(1-\lambda) U^i (\bar{c}_i)\leq U^i (\lambda c^i +(1-\lambda) \bar{c}_i) \leq U^i (\bar{c}_i)$, which implies that $U^i (c^i)\leq U^i (\bar{c}_i)$.
 
\end{proof}


\subsection{Existence of equilibrium for the infinite-horizon economy}\label{4.8.3}


For simplicity of notation, in what follows, we write $F_i$ instead of $F_{i,t}$. 
\begin{proposition}
Under Assumptions (1-4) and 7, there exists an equilibrium for the economy $\tilde{\ee}.$ 
\end{proposition}
\begin{proof}

We present a proof in the spirit of \cite*{lvp16}.

We have shown that there exists an equilibrium, say $\big(\bar{p}^T,  \bar{R}^T, (\bar{c}^T_{i},  \bar{k}^T_{i}, \bar{b}^T_{i})_i \big)$,  for each $T-$horizon truncated economy ${\ee}^T$.  Recall that  $\bar{p}^T_{t}+\bar{R}^T_{t}=1$ for any $t\leq T$.

 It is clear that the sequence $\big(\bar{p}^T,  \bar{r}^T, (\bar{c}^T_{i},  \bar{k}^T_{i}, \bar{b}^T_{i})_i \big)_T$ is bounded for the product topology. Since the set of time is a countable set, we can assume that,  without loss of generality,  
\begin{eqnarray*}&&\big(\bar{p}^T, \bar{R}^T, (\bar{c}^T_{i},  \bar{k}^T_{i}, \bar{b}^T_{i})_i \big))\xrightarrow{T\rightarrow \infty}\big(\bar{p}, \bar{R}, (\bar{c}_{i},  \bar{k}_{i}, \bar{b}_{i})_i \big)
\end{eqnarray*}
for the product topology.

We will prove that $\big(\bar{p}, \bar{R}, (\bar{c}_{i},  \bar{k}_{i}, \bar{b}_{i})_i \big)$ is an equilibrium for the economy ${\ee}.$ The market clearing conditions are trivial. We will prove that all prices are strictly positive and the allocation $(\bar{c}_{i},  \bar{k}_{i}, \bar{b}_{i})$ is optimal.

Let $(c_i, k_i, b_i)$ be a feasible allocation of the problem $P_i(\bar{p}, \bar{R})$. We prove that $U_i(c_i)\leq U_i(\bar{c}_i)$. 
%
Let define $(c_{i,t}', k_{i,t}', b_{i,t}')_{t\leq T}$ as follows:
\begin{align*}
 (c_{i,t}', k_{i,t}', b_{i,t}')&=(c_{i,t}, k_{i,t}, b_{i,t}) \text{ if } t \leq T-1\\
 (c'_{i,t}, k_{i,t}', b_{i,t}')&=(F_{i,t}(k_{i,t-1}), 0,0) \text{ if } t =T.
\end{align*}
%
We see that $(c_{i,t}', k_{i,t}', b_{i,t}')_{t\leq T}$ belongs to $B_i^T(\bar{p}, \bar{R})$.

Since $k_{i,-1}>0$ and $\bar{p}_{0}=1$, we have $\bar{p}_{0}F_{i,0}(k_{i,-1})>0$, and hence  $C_i^T(\bar{p},\bar{R})\not=\emptyset$.  Therefore there exists a sequence $\Big((c_{i,t}^n, k_{i,t}^n, b_{i,t}^n)_{t\leq T}\Big)_{n=0}^{\infty}\in C_i^T(\bar{p}, \bar{R}) $, with $k_{i,T}^n=0$, $b_{i,T}^n=0$, and this sequence converges to $(c_{i,t}', k_{i,t}', b_{i,t}')_{t\leq T}$ when $n$ tends to infinity. We have, for each $t\leq T$,
\begin{align*}
\bar{p}_{t}(c^n_{i,t}+k^n_{i,t})+{b}^n_{i,t}&< \bar{p}_{t}F_{i,t}(k^n_{i,t-1})+R_{t}{b}^n_{i,t-1}\\
f_i\bar{p}_{t}F_{i,t}(k^n_{i,t-1})+\bar{R}_{t}{b}^n_{i,t-1}&>0.
\end{align*}
Since $(\bar{p}^T, \bar{R}^T)$ converges to $(\bar{p}, \bar{R})$, we can chose $s_0$ high enough such that (i) $s_0>T$ and (ii) for every $s\geq s_0$, we have
\begin{align*}
\bar{p}^s_{t}(c^n_{i,t}+k^n_{i,t})+{b}^n_{i,t}&< \bar{p}^s_{t}F_{i,t}(k^n_{i,t-1})+R_{t}^s{b}^n_{i,t-1}\\
f_i\bar{p}^s_{t}F_{i,t}(k^n_{i,t-1})+\bar{R}^s_{t}{b}^n_{i,t-1}&>0.
\end{align*}

For $s\geq s_0$, define 
$(c_{i,t}^n(s), k_{i,t}^n(s), b_{i,t}^n(s))_{t\leq s}$ as follows:
\begin{align}
\text{for $t\leq T$: } &c_{i,t}^n(s)=c_{i,t}^n, \quad k_{i,t}^n(s)=k_{i,t}^n, \quad b_{i,t}^n(s)=b_{i,t}^n\\
\text{for $t= T+1,\ldots, s$: } & c_{i,t}^n(s)= k_{i,t}^n(s)= b_{i,t}^n(s)=0
\end{align}
Condition (ii) implies that $(c_{i,t}^n(s), k_{i,t}^n(s), b_{i,t}^n(s))_{t\leq s}\in C_i^s(\bar{p}^s, \bar{R}^s)$ for $s\geq s_0$. Therefore, by the definition of equilibrium in the T-truncated economy and $u_i(0)=0$, we get 
$$\summ_{t\leq T}\beta_i^tu_i(c^n_{i,t})= \summ_{t\leq s}\beta_i^tu_i({c}^s_{i,t}(s))\leq  \summ_{t\leq s}\beta_i^tu_i(\bar{c}^s_{i,t}).$$
%


Let $s$ tend to infinity, we obtain $\summ_{t\leq T}\beta_i^tu_i(c^n_{i,t})\leq \summ_{t\leq T}\beta_i^tu_i(\bar{c}_{i,t})$ for any $n$ and for any $T$ high enough. 

Let $n$ tend to infinity, we have
$\summ_{t\leq T}\beta_i^tu_i(c'_{i,t})\leq\summ_{t\leq T}\beta_i^tu_i(\bar{c}_{i,t})$ for any $T$. We write clearly this as follows:
$$\summ_{t\leq T-1}\beta_i^tu_i(c_{i,t})+\beta_i^Tu_i(F_{i,T}(k_{T,t-1}))\leq\summ_{t\leq T}\beta_i^tu_i(\bar{c}_{i,t}).$$

Let $T$ tend to infinity and note that $\lim_{T\to \infty}\beta_i^Tu_i(F_{i,T}(k_{i,T-1}))=0$, we get that 
\begin{align*}\summ_{t\geq 0}\beta_i^tu_i(c_{i,t})\leq\summ_{t\geq 0}\beta_i^tu_i(\bar{c}_{i,t}).
\end{align*}
So, we have proved the optimality of $(\bar{c}_i, \bar{k}_i, \bar{b}_i).$
 


 
 Now, we prove that $\bar{p}_{t}>0.$ Indeed, if $\bar{p}_{t}=0$, the agent $i$ can freely improve her consumption to obtain a level of utility, which is higher than  $\sum_{t\geq 0}u_{i,t}(\bar{c}_{i,t})$. This contradicts the optimality of  $(\bar{c}_i, \bar{k}_i, \bar{b}_i).$

 

We have $\bar{R}_{t}$ is strictly positive because otherwise we can choose another allocation such that $b_{i,t-1}=\infty$ and at the date $t$, we have the consumption $c'_{i,t+1}=\infty$, which make the utility of agent $i$ infinity, contradiction. 


\end{proof}

}

\end{document}